\documentclass[a4paper, titlepage, oneside, 11pt]{amsart}
\usepackage{amsaddr,amsmath,amssymb,amsfonts,amstext,amscd,cite,geometry,graphicx,latexsym,mathptmx,xcolor,pstricks,mathrsfs,amsthm,enumitem,multicol}
\usepackage[utf8]{inputenc}
\usepackage{newunicodechar}
\usepackage[all]{xy}
\usepackage[colorlinks=true,linkcolor=blue]{hyperref}
\AtBeginDocument{\hypersetup{citecolor=blue}}

\numberwithin{equation}{section}
\allowdisplaybreaks
\newtheorem{lem}{Lemma}[section]
\newtheorem{prop}[lem]{Proposition}
\newtheorem{thm}[lem]{Theorem}
\newtheorem{cor}[lem]{Corollary}

\newtheorem{df}[lem]{Definition}
\newtheorem{rem}[lem]{Remark}

\newtheorem{assu}[lem]{Assumption}

\def\R{{\mathbb{R}}}

\def\Z{{\mathbb{Z}}}

\newcommand{\BA}{\circle{4}}
\newcommand{\BB}{\circle*{4}}

\makeatletter
\renewcommand*{\eqref}[1]{%
  \hyperref[{#1}]{\textup{\tagform@{\ref*{#1}}}}%
}
\makeatother

\newcommand{\eqbox}[2]{
\bigskip
\framebox{\parbox{0.9\textwidth}{\textbf{#1}
#2}}
\bigskip}

\begin{document}

\title[Bi-infinite solutions for KdV- and Toda-type discrete integrable systems]{Bi-infinite solutions for KdV- and Toda-type\\discrete integrable systems based on path encodings}

\author[D.~A.~Croydon]{\vspace{-10pt}David A. Croydon}
\address{Research Institute for Mathematical Sciences, Kyoto University, Kyoto 606--8502, Japan} \email{croydon@kurims.kyoto-u.ac.jp}

\author[M.~Sasada]{\vspace{-10pt}Makiko Sasada}
\address{Graduate School of Mathematical Sciences, University of Tokyo, 3-8-1, Komaba, Meguro-ku, Tokyo, 153--8914, Japan}
\email{sasada@ms.u-tokyo.ac.jp}

\author[S.~Tsujimoto]{\vspace{-10pt}Satoshi Tsujimoto}
\address{Department of Applied Mathematics and Physics, Graduate School of Informatics, Kyoto University, Sakyo-ku, Kyoto 606--8501, Japan}
\email{tujimoto@i.kyoto-u.ac.jp}

\begin{abstract}
\vspace{10pt}
We define bi-infinite versions of four well-studied discrete integrable models, namely the ultra-discrete KdV equation, the discrete KdV equation, the ultra-discrete Toda equation, and the discrete Toda equation. For each equation, we show that there exists a unique solution to the initial value problem when the given data lies within a certain class, which includes the support of many shift ergodic measures. Our unified approach, which is also applicable to other integrable systems defined locally via lattice maps, involves the introduction of a path encoding (that is, a certain antiderivative) of the model configuration, for which we are able to describe the dynamics more generally than in previous work on finite size systems, periodic systems and semi-infinite systems. In particular, in each case we show that the behaviour of the system is characterized by a generalization of the classical `Pitman's transformation' of reflection in the past maximum, which is well-known to probabilists. The picture presented here also provides a means to identify a natural `carrier process' for configurations within the given class, and is convenient for checking that the systems we discuss are all-time reversible. Finally, we investigate links between the different systems, such as showing that bi-infinite all-time solutions for the ultra-discrete KdV (resp.\ Toda) equation may appear as ultra-discretizations of corresponding solutions for the discrete KdV (resp.\ Toda) equation.
\bigskip

\noindent
\emph{Data sharing not applicable to this article as no datasets were generated or analysed during the current study.}
\end{abstract}

\keywords{KdV equation, Toda lattice, discrete integrable system, Pitman's transformation}

\subjclass[2010]{37K10 (primary), 35Q53, 37K60 (secondary)}

\date{\today}

\maketitle

\setcounter{tocdepth}{3}
\tableofcontents
\newpage

\section{Introduction}

Two central and widely-studied examples of classical integrable systems are the Korteweg-de Vries (KdV) equation and the Toda lattice. These were introduced as models for shallow water waves \cite{KdV} and a one-dimensional crystal \cite{Toda}, respectively, and are described precisely by the following equations.

\eqbox{Korteweg-de Vries (KdV) equation}{\begin{equation}\nonumber
\partial_t u = 6u \partial_x u -\partial_{xxx} u,
\end{equation}
where $u=(u(x,t))_{x,t\in\mathbb{R}}$.}

\vspace{-10pt}
\eqbox{Toda lattice equation}{\begin{equation}\nonumber
\begin{cases}
\frac{d}{dt} p_n &=e^{-(q_n-q_{n-1})}-e^{-(q_{n+1}-q_n)},\\
\frac{d}{dt} q_n &=p_n,
\end{cases}
\end{equation}
for all $n\in\mathbb{Z}$, where $p_n=(p_n(t))_{t\in\mathbb{R}}$, $q_n=(q_n(t))_{t\in\mathbb{R}}$.}

\noindent
For each of these systems, an important question concerns the existence and uniqueness of solutions of the defining equations for given initial data. In the case of the KdV equation, many results in this direction have been established when the system is started from a smooth function that is periodic or decaying at infinity (see for instance \cite{BonaSmith, Kenig, Kishi}). Moreover, in recent years, significant progress has been made in understanding the Cauchy problem for more general initial data. Indeed, in \cite{Kotani}, Kotani solved the problem for a class of initial data that incorporates ergodic functions. And, in \cite{KMV}, Killip, Murphy and Visan constructed a solution of the KdV equation started from white noise, and further showed the invariance in distribution of this solution under the KdV dynamics. We note that there has also been extensive work in solving the KdV equation, including demonstrating the invariance of white noise, on the one-dimensional torus \cite{Q, Bourgain}. As for the Toda lattice equation, existence and uniqueness of the dynamics is standard for initial data that increase at infinity at a sub-exponential rate \cite{LL}. Consequently, one has that for random initial configurations whose law is spatially shift ergodic, the dynamics are almost-surely uniquely well-defined. Additionally, the dynamics are known to admit a class of distributionally invariant random configurations whose laws are given by generalized Gibbs measures \cite{SP}. The initial value problem has also been considered for variations of the model that include the periodic Toda lattice and the finite Toda lattice (see \cite{Krichever, Moser}).

The purpose of this paper is to develop a general approach for constructing solutions of the corresponding initial value problems for a variety of discretizations of the KdV equation and Toda lattice, namely the discrete KdV equation \cite{H:DE1}, the ultra-discrete KdV equation \cite{TTMS,TH}, the discrete Toda lattice \cite{H:DE2,HTI}, and the ultra-discrete Toda lattice \cite{TM1995,NTS}. (See equations \eqref{DKDV}, \eqref{UDKDV}, \eqref{DTODA}, and \eqref{UDTODA} below for precise definitions, and Figure \ref{dis} for a schematic diagram showing connections between these and the KdV and Toda lattice equations. Moreover, for a brief overview of the main contributions of this article to the study of these equations, see the first sentences of Theorems \ref{udkdvthm}, \ref{dkdvthm}, \ref{udtodathm} and \ref{dtodathm}.) Such discrete systems have been well-studied, in terms of providing numerical algorithms for simulating the continuous systems, as integrable systems in their own right (see \cite{IKT, TTeng} for some introductory reviews of material in this direction), and with respect to their applicability to other problems in numerical analysis (see \cite{NTSeps, Sogo}, for example). On the other hand, in the context of dynamical systems, these discrete models have not been well-analyzed compared to the original continuous ones. Specifically, the existence and the uniqueness of solutions for general initial data has not been explored. Rather, the literature has tended to focus on periodic configurations, or those that decay quickly at infinity, including those based on solitons. Whilst it is anyway of interest to explore solutions beyond these special cases as we do in the present article, an important motivation for considering bi-infinite solutions specifically is to provide a framework for studying the dynamics of the discrete systems started from random initial configurations. Indeed, it is natural to ask what the invariant measures of the models in question are, and in a follow-up work we present some progress in this direction (see \cite{CSirf}, and the survey article \cite{CSmsj}).

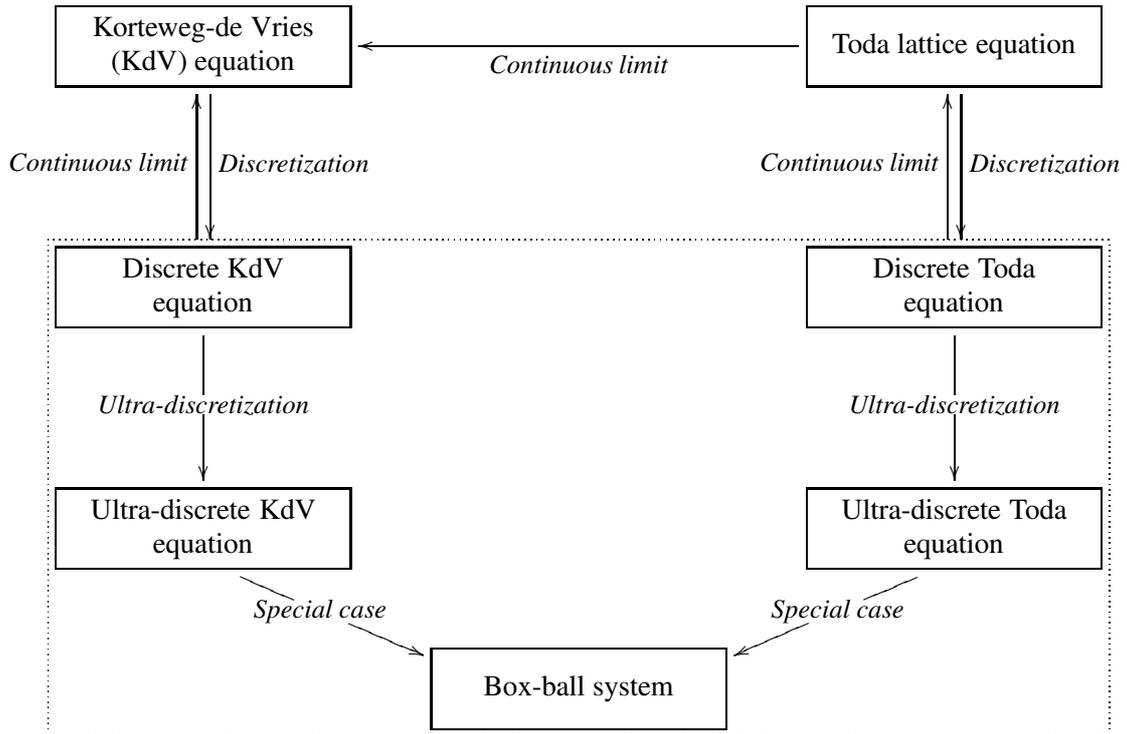
\begin{figure}
\centerline{\xymatrix{\framebox(110,30){\parbox{100pt}{\centering{Korteweg-de Vries (KdV) equation}}}
\ar@<.5ex>[dd]^<<<<<<<<<{\mbox{\small{\emph{Discretization}}}}&&\framebox(110,30){\parbox{100pt}{\centering{Toda lattice equation}}}\ar@<.5ex>[dd]^<<<<<<<<<{\mbox{\small{\emph{Discretization}}}}\ar@<-2.5ex>[ll]^{\mbox{\small{\emph{Continuous limit}}}}\\
&&\\
\framebox(110,30){\parbox{100pt}{\centering{Discrete KdV equation}}}\ar[dd]|<<<<<<<<<{\mbox{\small{\emph{Ultra-discretization}}}}\ar@<.5ex>[uu]^<<<<<<<<<<{\mbox{\small{\emph{Continuous limit}}}}&&\framebox(110,30){\parbox{100pt}{\centering{Discrete Toda equation}}}\ar[dd]|<<<<<<<<<{\mbox{\small{\emph{Ultra-discretization}}}}\ar@<.5ex>[uu]^<<<<<<<<<<{\mbox{\small{\emph{Continuous limit}}}}\\
&&\\
\framebox(110,30){\parbox{100pt}{\centering{Ultra-discrete KdV equation}}}\ar@<0ex>[dr]|<<<<<<<<<<<{\mbox{\small{\emph{Special case}}}}&&\framebox(110,30){\parbox{100pt}{\centering{Ultra-discrete Toda equation}}}\ar@<-0ex>[dl]|<<<<<<<<<<<{\mbox{\small{\emph{Special case}}}}\\
&\framebox(110,30){\parbox{100pt}{\centering{Box-ball system}}}&\hspace{109pt}\save "3,1"."6,3"*[F.]\frm{}\restore}}
\caption{Classical integrable systems derived from the KdV and Toda equations; the discrete integrable systems studied in this article are enclosed by the dotted line. Ultra-discretization is discussed within our framework in Subsection \ref{udsec}, and links to the box-ball system are recalled in Subsection \ref{specialsec}.}\label{dis}
\end{figure}

The program that we follow here will be one that generalizes that of \cite{CKST}, in which the initial value problem was studied for a special case of the ultra-discrete KdV equation called the box-ball system (BBS), as introduced in \cite{takahashi1990}. To help place the present work into context, let us briefly review the construction of BBS dynamics that appears in \cite{CKST}, in which the key observation is that each discrete-time step in the evolution of the BBS can be described by a certain transformation of an encoding of the current configuration by a nearest-neighbour path on $\mathbb{Z}$. More specifically, configurations of the BBS are sequences of the form $\eta=(\eta_n)_{n\in\mathbb{Z}}\in\{0,1\}^\mathbb{Z}$, where $\eta_n=1$ denotes the presence of a ball in the $n$th box, and $\eta_n=0$ the lack thereof. The configuration after one step of the dynamics, $T\eta$ say, is given by the following formal version of the ultra-discrete KdV equation:
\begin{equation}\label{bbsudkdv1}
(T\eta)_{n}=\min\left\{1-\eta_{n},\sum_{m=-\infty}^{n-1}\left(\eta_m - (T\eta)_m\right)\right\}.
\end{equation}
Although we can make sense of the above formulation for configurations that have $\eta_n=0$ eventually as $n\rightarrow\infty$ by supposing that $(T\eta)_n=0$ for integers that lie below the smallest $n$ such that $\eta_n=1$, describing the dynamics for configurations that lie outside this class of initial data requires further thought. In this direction, it is convenient to introduce an auxiliary sequence of variables $W=(W_n)_{n\in\mathbb{Z}}$ taking values in $\mathbb{Z}_+=\{0,1,\dots\}$, and observe that for the class of configurations just described, the equation at \eqref{bbsudkdv1} is equivalent to
\begin{equation}\label{bbsudkdv2}
\begin{cases}
(T\eta)_n& = \min\{1-\eta_n,W_{n-1}\},\\
W_n & =\eta_n+W_{n-1}-(T\eta)_n,
\end{cases}
\end{equation}
subject to the boundary condition $\lim_{n\rightarrow-\infty}W_n=0$; we note that the sequence $W$ can be interpreted as a `carrier process', with $W_n$ representing the number of balls shifted from $\{\dots,n-1,n\}$ to $\{n+1,n+2,\dots\}$, and the boundary condition at $-\infty$ means that the carrier is initially empty. For general $\eta$, the question then essentially becomes one of determining the existence and uniqueness of a corresponding carrier $W$ (as well as considering what the appropriate boundary condition for the carrier is as $n\rightarrow-\infty$). In \cite{CKST}, it was shown that one solution to this problem is given by first introducing a `path encoding' of the configuration, $S=(S_n)_{n\in\mathbb{{Z}}}\in \mathbb{Z}^\mathbb{Z}$ say, which is the antiderivative obtained by fixing $S_0=0$ and defining
\begin{equation}\label{incs}
S_n-S_{n-1}=1-2\eta_n,\qquad \forall n\in\mathbb{Z},
\end{equation}
and secondly using this to construct a specific carrier process via the formula
\begin{equation}\label{originalW}
W_n:=M_n-S_n,\qquad \forall n\in\mathbb{Z},
\end{equation}
where $M=(M_n)_{n\in\mathbb{Z}}$ is the `past maximum' of $S$, namely
\begin{equation}\label{originalM}
M_n:=\sup_{m \le n}S_m.
\end{equation}
As long as $M_0<\infty$, (which ensures $M_n<\infty$ for all $n\in\mathbb{Z}$,) this procedure yields a solution to \eqref{bbsudkdv2}. Indeed, it is not difficult to check that $W$ satisfies $W_n =\eta_n+W_{n-1}-\min\{1-\eta_n,W_{n-1}\}$, and one can then use the first equation at \eqref{bbsudkdv2} to determine $T\eta$. Moreover, it is easily seen that the updated configuration is encoded by the path
\begin{equation}\label{originalpitman}
TS:=2M-S
\end{equation}
in that, as at \eqref{incs}, the increments of $TS$ yield $T\eta$, i.e.\ $(TS)_n-(TS)_{n-1}=1-2(T\eta)_n$ for all $n\in\mathbb{Z}$. See Figure \ref{taufig} below
for a depiction of the BBS dynamics at the configuration and path encoding levels. The transformation at \eqref{originalpitman} of reflection in the past maximum is well-known in the probability literature as Pitman's transformation after its introduction in the fundamental work on stochastic processes \cite{Pitman}. See \cite{Bertoin, HMOC, Jeulin, MY0, MY, MY1, MY2, Rog, RogPit, Saisho} for examples of further studies concerning Pitman's transformation and its generalizations, \cite{DMO,HW,OCY2} for connections to queuing theory, and \cite{BBOC1,BBOC2, OC0,OC,OCY} for research relating Pitman's transformation and stochastic integrable systems. (We note that in \cite{CKST}, the operator $T$ was given by setting $TS=2M-S-2M_0$. Clearly dropping the constant shift does not affect the dynamics, and it will be convenient for our later arguments to do without this.) For one time step of the dynamics, the carrier is not uniquely defined by \eqref{bbsudkdv2}, see \cite[Remark 2.11]{CKST} for discussion of this point. However, in \cite{CKST}, a complete description was given of configurations for which the BBS dynamics can be extended and are time reversible (i.e.\ invertible) for all times (forwards and backwards), and apart from on a certain `critical' set,
for such configurations the choice of $W$ at \eqref{originalW} yields the unique solution to \eqref{bbsudkdv2} that stays within the domain of the dynamics. Thus the above construction of $W$ and $T\eta$ is the only viable choice if one seeks to define the global (i.e.\ all-time) dynamics of the BBS. We note that a similar description of bi-infinite BBS dynamics in terms of path encodings is provided in \cite{FNRW}, as is a description in terms of a so-called `slot decomposition', which keeps track of solitons.

\begin{figure}
\begin{center}
\BA\BA\BA\BA\BB\BA\BA\BA\BA\BA\BA\BB\BA\BB\BA\BA\BA\BA\BA\BA\BB\BB\BB\BB\BA\BB\BB\BB\BA\BA\BB\BA\BB\BB\BA\BB\BB\BA\BA\BA\BA\BA\BA\BA\BB\BB\BA\BB\BB\BB\BA\BB\BA\BB\BA\BA\BB\BA\BA\BA\BA\BA\BA\BB\BA\BB\BA\BB\BB\BB\BA\BB\BB\BA\BA\BA\BA\BA\BB\BA\BA\BA\BA\BA\BB\BA\BA\BA\BB\BA\BA\BB\BA\BA\BB\BB\BA\BB\BA\\
\vspace{-5pt}\BA\BA\BA\BB\BA\BA\BA\BA\BA\BB\BB\BA\BB\BA\BB\BB\BB\BB\BB\BB\BA\BA\BA\BA\BB\BA\BA\BA\BB\BB\BA\BB\BA\BA\BB\BA\BA\BA\BA\BA\BB\BB\BB\BB\BA\BA\BB\BA\BA\BA\BB\BA\BB\BA\BA\BB\BA\BA\BA\BB\BB\BB\BB\BA\BB\BA\BB\BA\BA\BA\BB\BA\BA\BA\BA\BA\BA\BB\BA\BA\BA\BA\BA\BB\BA\BA\BA\BB\BA\BB\BB\BA\BB\BB\BA\BA\BB\BA\BB\\
\vspace{-5pt}\BB\BB\BB\BA\BB\BB\BB\BB\BB\BA\BA\BB\BA\BB\BA\BA\BA\BA\BA\BA\BA\BA\BA\BB\BA\BA\BB\BB\BA\BA\BB\BA\BA\BB\BA\BA\BB\BB\BB\BB\BA\BA\BA\BA\BA\BB\BA\BA\BA\BB\BA\BB\BA\BB\BB\BA\BB\BB\BB\BA\BA\BA\BA\BB\BA\BB\BA\BA\BA\BB\BA\BA\BA\BA\BA\BA\BB\BA\BA\BA\BA\BB\BB\BA\BB\BB\BB\BA\BB\BA\BA\BB\BA\BA\BB\BB\BA\BB\BA\\
\vspace{-5pt}\BA\BA\BA\BB\BA\BA\BA\BA\BA\BA\BB\BA\BB\BA\BA\BA\BA\BA\BA\BA\BA\BA\BB\BA\BB\BB\BA\BA\BB\BB\BA\BB\BB\BA\BB\BB\BA\BA\BA\BA\BA\BA\BA\BB\BB\BA\BB\BB\BB\BA\BB\BA\BB\BA\BA\BB\BA\BA\BA\BA\BA\BA\BB\BA\BB\BA\BA\BA\BB\BA\BA\BA\BA\BA\BA\BB\BA\BB\BB\BB\BB\BA\BA\BB\BA\BA\BA\BB\BA\BB\BB\BA\BB\BB\BA\BA\BB\BA\BB\\
\vspace{-5pt}\BA\BA\BB\BA\BA\BA\BA\BA\BA\BB\BA\BB\BA\BA\BA\BA\BA\BA\BB\BB\BB\BB\BA\BB\BA\BA\BB\BB\BA\BA\BB\BA\BA\BB\BA\BA\BA\BA\BA\BB\BB\BB\BB\BA\BA\BB\BA\BA\BA\BB\BA\BB\BA\BA\BB\BA\BA\BA\BA\BA\BA\BB\BA\BB\BA\BA\BA\BB\BA\BA\BB\BB\BB\BB\BB\BA\BB\BA\BA\BA\BA\BB\BB\BA\BB\BB\BB\BA\BB\BA\BA\BB\BA\BA\BB\BB\BA\BB\BA\\
\vspace{-5pt}\BA\BB\BA\BA\BA\BA\BA\BA\BB\BA\BB\BA\BA\BA\BB\BB\BB\BB\BA\BA\BA\BA\BB\BA\BB\BB\BA\BA\BA\BB\BA\BA\BB\BA\BA\BB\BB\BB\BB\BA\BA\BA\BA\BA\BB\BA\BA\BA\BB\BA\BB\BA\BA\BB\BA\BA\BA\BA\BA\BB\BB\BA\BB\BA\BB\BB\BB\BA\BB\BB\BA\BA\BA\BA\BA\BB\BA\BB\BB\BB\BB\BA\BA\BB\BA\BA\BA\BB\BA\BA\BB\BA\BB\BB\BA\BA\BB\BA\BB\\
\vspace{-5pt}\BB\BA\BA\BA\BA\BA\BB\BB\BA\BB\BA\BB\BB\BB\BA\BA\BA\BA\BA\BA\BB\BB\BA\BB\BA\BA\BA\BB\BB\BA\BB\BB\BA\BB\BB\BA\BA\BA\BA\BA\BA\BA\BA\BB\BA\BB\BB\BB\BA\BB\BA\BB\BB\BA\BB\BB\BB\BB\BB\BA\BA\BB\BA\BB\BA\BA\BA\BB\BA\BA\BB\BB\BB\BB\BB\BA\BB\BA\BA\BA\BA\BB\BB\BA\BB\BB\BB\BA\BB\BB\BA\BB\BA\BA\BB\BB\BA\BB\BA\\
\vspace{-5pt}\BA\BA\BB\BB\BB\BB\BA\BA\BB\BA\BB\BA\BA\BA\BA\BA\BB\BB\BB\BB\BA\BA\BB\BA\BB\BB\BB\BA\BA\BB\BA\BA\BB\BA\BA\BB\BB\BB\BB\BB\BB\BB\BB\BA\BB\BA\BA\BA\BB\BA\BB\BA\BA\BB\BA\BA\BA\BA\BA\BB\BB\BA\BB\BA\BB\BB\BB\BA\BB\BB\BA\BA\BA\BA\BA\BB\BA\BB\BB\BB\BB\BA\BA\BB\BA\BA\BA\BB\BA\BA\BB\BA\BB\BB\BA\BA\BB\BA\BA\\
\vspace{-5pt}\BB\BB\BA\BA\BA\BA\BB\BB\BA\BB\BA\BB\BB\BB\BB\BB\BA\BA\BA\BA\BB\BB\BA\BB\BA\BA\BA\BB\BB\BA\BB\BB\BA\BB\BB\BA\BA\BA\BA\BA\BA\BA\BA\BB\BA\BA\BA\BB\BA\BB\BA\BA\BB\BA\BB\BB\BB\BB\BB\BA\BA\BB\BA\BB\BA\BA\BA\BB\BA\BA\BA\BB\BB\BB\BB\BA\BB\BA\BA\BA\BA\BA\BB\BA\BA\BA\BB\BA\BB\BB\BA\BB\BA\BA\BA\BB\BA\BA\BA\\
\vspace{-5pt}\BA\BA\BB\BB\BB\BB\BA\BA\BB\BA\BB\BA\BA\BA\BA\BA\BA\BB\BB\BB\BA\BA\BB\BA\BB\BB\BB\BA\BA\BB\BA\BA\BB\BA\BA\BA\BA\BA\BA\BA\BA\BB\BB\BA\BB\BB\BB\BA\BB\BA\BB\BB\BA\BB\BA\BA\BA\BA\BA\BA\BB\BA\BB\BA\BA\BB\BB\BA\BB\BB\BB\BA\BA\BA\BA\BB\BA\BA\BA\BA\BA\BB\BA\BA\BB\BB\BA\BB\BA\BA\BB\BA\BA\BA\BB\BA\BA\BA\BB\\

\vspace{-20pt}
\includegraphics[width=0.95\textwidth]{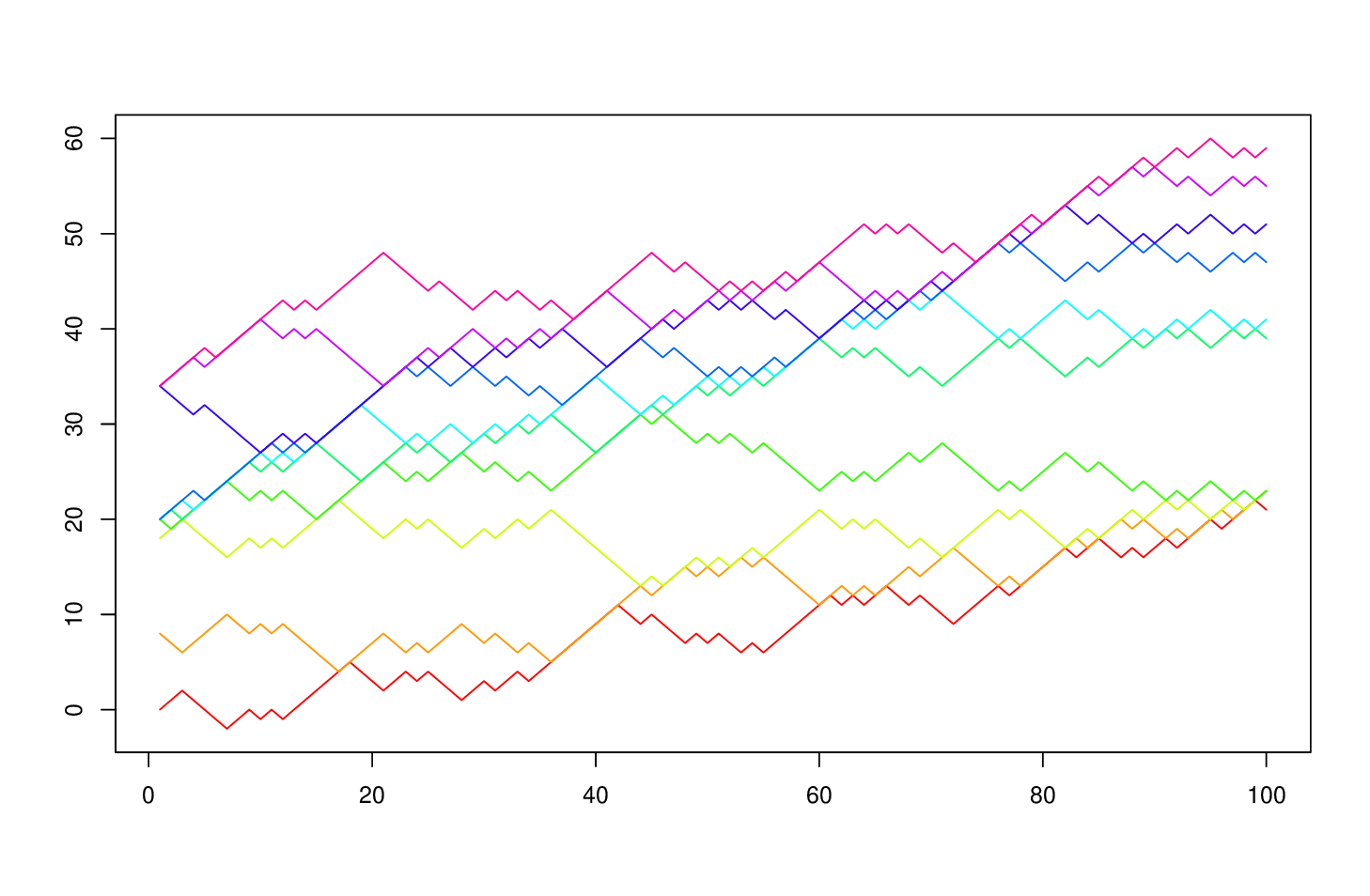}
\rput(-6.7,.7){$n$}
\rput(-13.7,5.2){$S_n^t$}
\end{center}
\caption{Configuration $(\eta_n^t)_{n,t\in\mathbb{Z}}$ and path encoding $(S_n^t)_{n,t\in\mathbb{Z}}$ for the same realization of the BBS, shown over the spatial interval $n=1,\dots,100$ and time interval $t=1,\dots,10$. In the configuration picture, time runs from bottom to top. In the path encoding picture, colours run from red to pink with increasing time.}\label{taufig}
\end{figure}

The principal contribution of the present article is to show that the Pitman transformation strategy of constructing solutions is also valid for the equations defining the discrete integrable systems of interest here, i.e.\ \eqref{DKDV}, \eqref{UDKDV}, \eqref{DTODA}, and \eqref{UDTODA}. That is, for each system, we have that the existence of a single time step solution is equivalent to existence of a certain carrier process (which can be seen as the analogue of the notion introduced for the BBS above), and we demonstrate that such a carrier can be constructed in terms of a functional of a certain path encoding, which also yields via a Pitman-type transformation the dynamics on the configuration space. Moreover, we check that our choice of carrier is the unique one for which we are to be able to iterate this process; see Section \ref{dissec} for our main results in this direction. Specifically, in each case, path encodings will be elements of the space
\begin{equation}\label{sdef}
\mathcal{S} :=\left\{ S : \Z \to \R\right\},
\end{equation}
and the dynamics are described similarly to \eqref{originalpitman}, but the past maximum of \eqref{originalM} is replaced by the path functionals given in Table \ref{Mtable}, and the Toda-type systems incorporate a spatial shift. Note that, in Table \ref{Mtable} and subsequently, we use $\theta$ to denote the usual left-shift, i.e.\ $\theta((x_n)_{n\in\mathbb{Z}}) = (x_{n+1})_{n\in\mathbb{Z}}$. See Figure \ref{pitmanpic} for illustrative examples of the path transformations that we introduce, and Section \ref{dissec} for details of how these relate to the relevant integrable systems. We note that, whilst $T^\vee$ and $T^{\vee^*}$ closely correspond to the original Pitman transformation defined at \eqref{originalpitman}, the operators $T^{\sum}$ and $T^{\sum^*}$ are related to the exponential version of Pitman's transformation, which is also familiar to probabilists (see \cite{MY1, MY2}, for example). It is not difficult to verify that the path operations described in Table \ref{Mtable} are well-defined on the set of asymptotically linear functions, as given by
\begin{equation}\label{slindef}
\mathcal{S}^{lin} :=\left\{ S  \in\mathcal{S} \::\: \lim_{n \to -\infty} \frac{S_n}{n} > 0\mbox{ and }\lim_{n \to +\infty} \frac{S_n}{n} > 0\right\}.
\end{equation}
Moreover, whilst we do not attempt to replicate the detailed study of \cite{CKST} and determine the full set of configurations upon which the dynamics can be extended and checked to be time reversible for all time, we will check that the latter properties hold for configurations with a path encoding that lies in $\mathcal{S}^{lin}$. (See Section \ref{dissec} for details.) We remark that the restriction to configurations with a path encoding in $\mathcal{S}^{lin}$ is relatively mild, and in particular yields that for initial configurations that have a shift ergodic distribution and satisfy a certain density condition (specific to each case), the dynamics of the discrete integrable system are uniquely defined for all times. Moreover, the class of configurations we deal with includes the previously studied examples of semi-infinite, periodic and rapidly decaying configurations, see Subsection \ref{knownsec}.

\begin{table}[t]
\begin{center}
\begin{tabular}{r|c|l}
  \emph{Model} & \emph{`Past maximum'} & \emph{Path encoding dynamics}\\
  \hline
 {udKdV}&  \raisebox{20pt}{\phantom{M}}\(\displaystyle M^{\vee}(S)_n:=\sup_{m \le n} \left(\frac{S_m+S_{m-1}}{2}\right)\) \raisebox{-20pt}{\phantom{M}}&  $T^{\vee}(S)=2M^{\vee}(S)-S$\\
             \hline
  {dKdV} & \raisebox{20pt}{\phantom{M}}\(\displaystyle M^{\sum}(S)_n:=\log \left(\sum_{m \le n} \exp\left(\frac{S_m+S_{m-1}}{2}\right)\right)\)\raisebox{-20pt}{\phantom{M}}& $T^{\sum}(S)=2M^{\sum}(S)-S$\\
  \hline
  {udToda}& \raisebox{20pt}{\phantom{M}}\(\displaystyle M^{\vee^*}(S)_{n}:=\left\{\begin{array}{ll}
 \sup_{m \le \frac{n-1}{2}}S_{2m}, & n\mbox{ odd}, \\
  \frac{M^{\vee^*}(S)_{n+1}+M^{\vee^*}(S)_{n-1}}{2}, &  n\mbox{ even},
  \end{array}
\right.\)\raisebox{-20pt}{\phantom{M}} &\parbox{3.6cm}{$\mathcal{T}^{\vee^*}(S)=\theta\circ{T}^{\vee^*}(S),$\\where\\
${T}^{\vee^*}(S):=2M^{\vee^*}(S)-S$}\\
  \hline
  {dToda} &\raisebox{20pt}{\phantom{M}}\(\displaystyle M^{\sum^*}(S)_n:=\left\{\begin{array}{ll}
 \log \left(\sum_{m \le \frac{n-1}{2}} \exp\left(S_{2m}\right) \right), & n\mbox{ odd}, \\
  \frac{M^{\sum^*}(S)_{n+1}+M^{\sum^*}(S)_{n-1}}{2}, &  n\mbox{ even}
  \end{array}
\right.\)&\parbox{3.6cm}{$\mathcal{T}^{\sum^*}(S)=\theta\circ{T}^{\sum^*}(S),$\\where\\
${T}^{\sum^*}(S):=2M^{\sum^*}(S)-S$}
\end{tabular}
\bigskip
\end{center}
\caption{`Past maximum' operators and description of path encoding dynamics for the four discrete integrable systems studied in this article. Note that the operator $M^\vee$ that naturally arises in our study of \eqref{UDKDV} is not the same as the operator $M$ appearing in previous studies of the original BBS. However, as we discuss in Remark \ref{origbbsrem}(a), when considered for the latter model, $M^\vee$ and $M$ differ only by a constant shift and lead to the same system dynamics.}\label{Mtable}
\end{table}

\begin{figure}
\begin{center}
\includegraphics[width=0.43\textwidth]{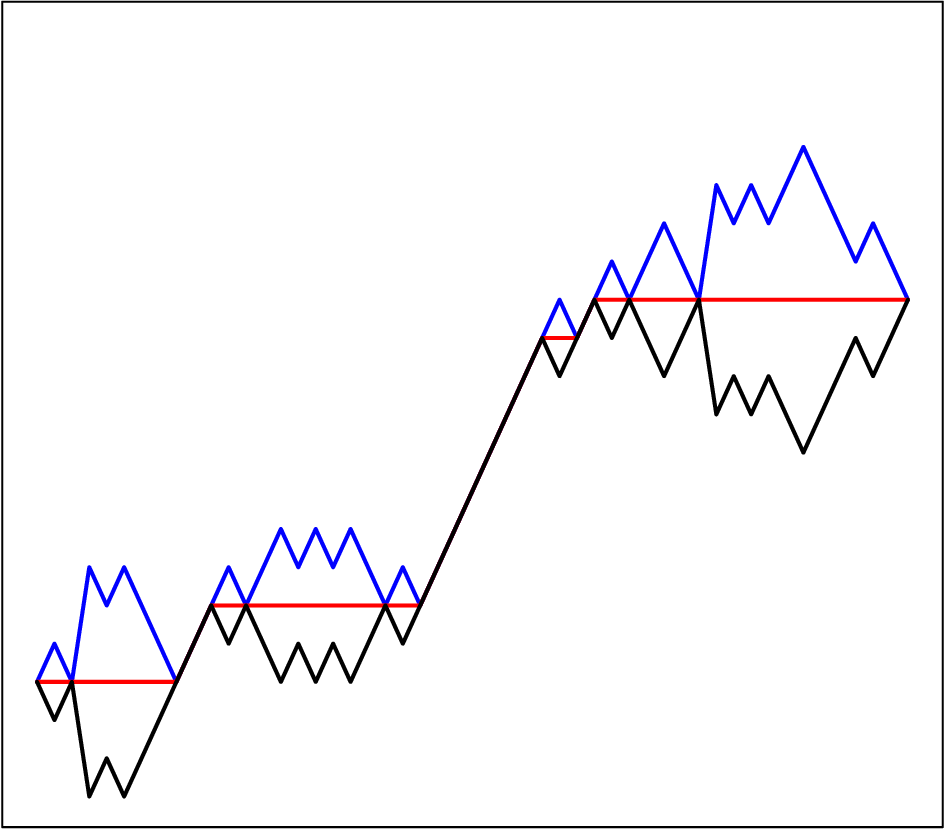}
\rput(-6.05,5.15){$\boxed{T}$}
\bigskip

\includegraphics[width=0.43\textwidth]{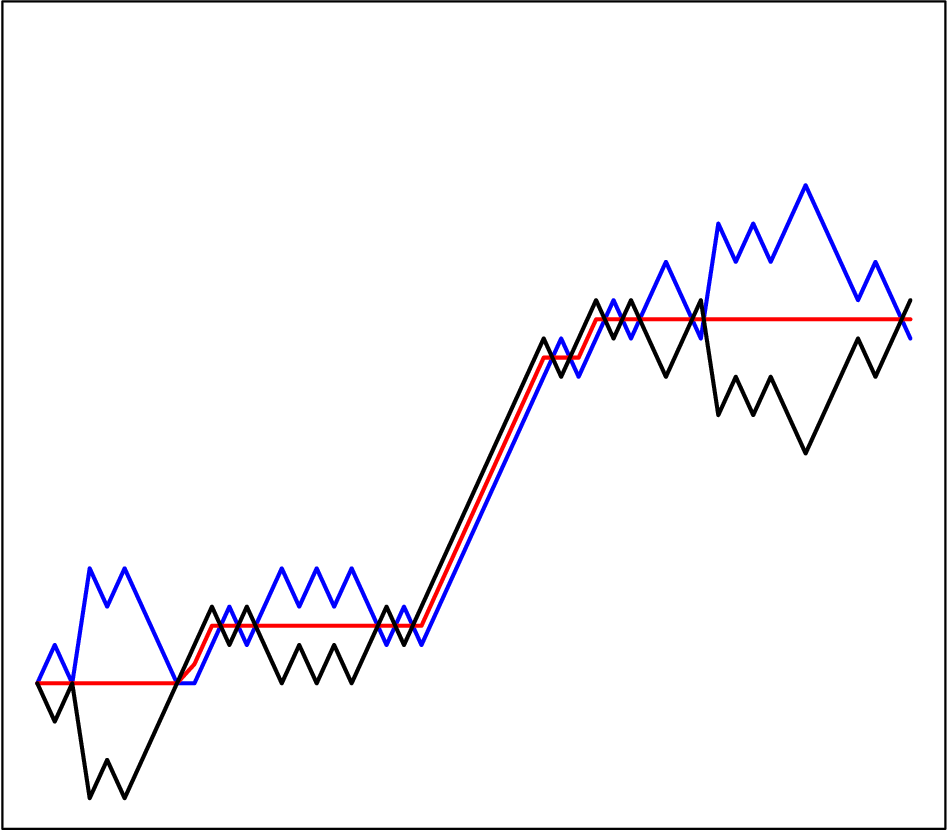}\:\:\:\includegraphics[width=0.43\textwidth]{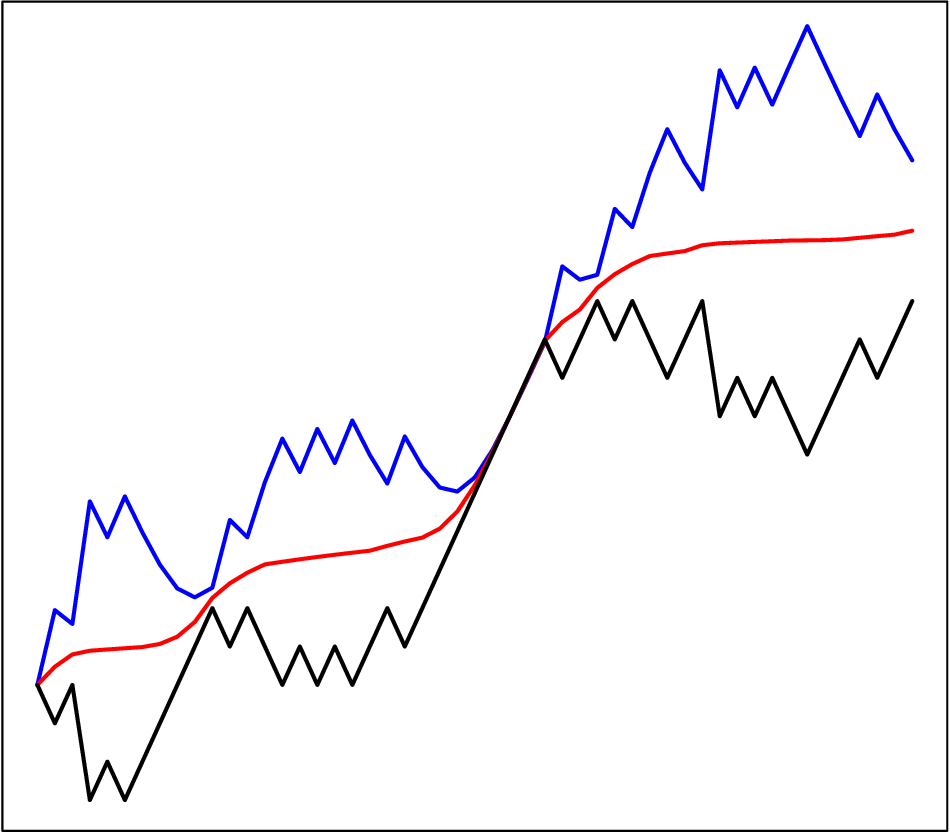}
\rput(-12.4,5.1){$\boxed{T^\vee}$}
\rput(-5.9,5.1){$\boxed{T^{\sum}}$}
\bigskip

\includegraphics[width=0.43\textwidth]{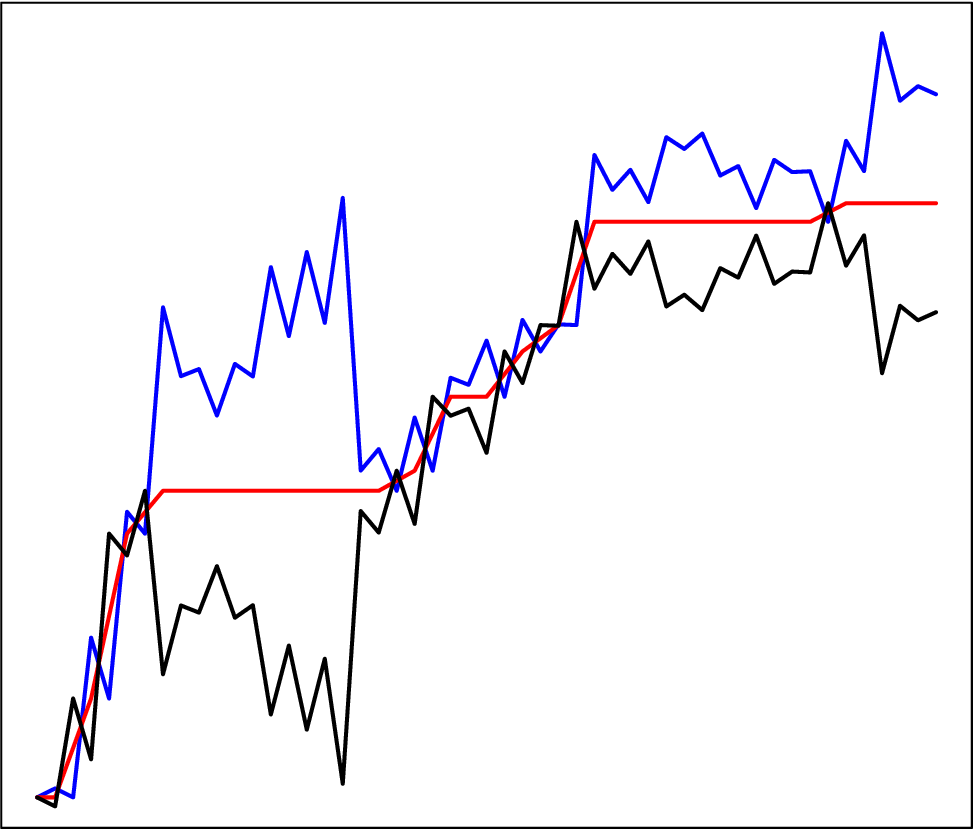}\:\:\:\includegraphics[width=0.43\textwidth]{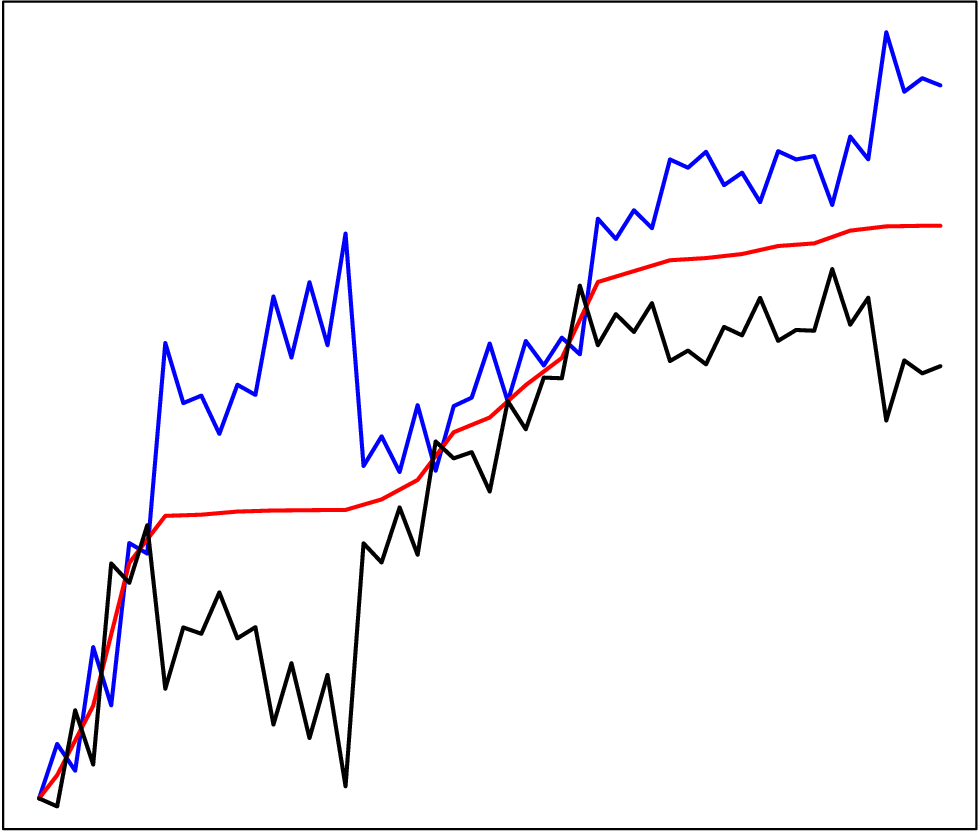}
\rput(-12.3,4.95){$\boxed{T^{\vee^*}}$}
\rput(-5.8,4.95){$\boxed{T^{\sum^*}}$}
\end{center}
\caption{Examples of Pitman-type transformations. In each picture, $S$ is shown in black, $M(S)$ is shown in red, and $T(S)$ is shown in blue. The initial paths are the same in the figures for the KdV-type operators $T$, $T^\vee$ and $T^{\sum}$, as they are in the figures for the Toda-type operators $T^{\vee^*}$ and $T^{\sum^*}$.}\label{pitmanpic}
\end{figure}

Once the path encodings and Pitman-type transformations have been identified for each of the discrete integrable systems, it is not a hugely challenging exercise to check the existence and uniqueness of solutions to the equations \eqref{DKDV}, \eqref{UDKDV}, \eqref{DTODA}, and \eqref{UDTODA}. However, rather than do this individually for each model, we prefer to present a unified approach; see Section \ref{general} for details of how, within a general framework of locally-defined dynamics, we define notions of path encodings, carriers, past maxima and Pitman-type transformations. Although the abstraction does make the proofs slightly more technical, it is perhaps more informative as to why we see the Pitman-type transformations that we do in the first place, since it gives the perspective that such transformations can be understood as encoding conservation laws for the underlying lattice equations. Moreover, our framework gives a means to developing similar theory in other applications. (Indeed, it also covers the case of the BBS with finite box and/or ball capacity, as studied in \cite{CS}.) Other advantages are that it allows us to treat issues such as time reversibility for each of the models simultaneously, and makes clear various symmetries of and links between the models. For instance, we check that the ultra-discretization procedure extends to the general class of configurations that we consider (see Subsection \ref{udsec}). Whilst reading Section \ref{general}, we recommend readers keep their favourite example in mind, as most of the statements become quite transparent when converted into the notation of an explicit system. To this end, it might be helpful to note that the relevant translations for the \eqref{DKDV}, \eqref{UDKDV}, \eqref{DTODA} and \eqref{UDTODA} systems are described in Section \ref{proofsec}.

The remainder of the article is organized as follows. In Section \ref{dissec}, we introduce the equations \eqref{DKDV}, \eqref{UDKDV}, \eqref{DTODA} and \eqref{UDTODA}, and present our main results concerning their solution via path encodings. Towards motivating our general approach, in Section \ref{lmsec} we describe how each of the systems in question can be represented in terms of lattice maps, and recall some relevant properties of these. Our unified framework for studying discrete integrable systems is then developed in what can be considered the heart of the article -- Section \ref{general}, before being applied to prove our principal conclusions for the four specific models of interest in Section \ref{proofsec}. Finally, we collect together various points of discussion concerning our main results and modifications of them in Section \ref{discussionsec}. Before starting to present this content, we highlight that it is our notational convention to distinguish between $\mathbb{Z}_+:=\{0,1,\dots\}$ and $\mathbb{N}:=\{1,2,\dots\}$. We will also define $\mathbb{Z}_-:=-\mathbb{N}=\{\dots,-2,-1\}$, and write:
\begin{itemize}
  \item $\lim_{n \to \pm\infty}a_n<A$, $\lim_{n \to \pm\infty}a_n>A$, etc.\ to mean that both $\lim_{n \to -\infty}a_n$ and $\lim_{n \to +\infty}a_n$ exist and satisfy the relevant inequality (i.e.\ we do not insist that the limits at $-\infty$ and $+\infty$ are equal, cf.\ \eqref{slindef});
  \item for $n_1=n_2+1$, $\sum_{m=n_1}^{n_2}:=0$, and for $n_1>n_2+1$, $\sum_{m=n_1}^{n_2}:=-\sum_{m=n_2+1}^{n_1-1}$.
\end{itemize}

\section{Initial value problems for KdV- and Toda-type discrete integrable systems}\label{dissec}

In this section, we present our main results for the four discrete integrable systems focussed upon in this article. In particular, in each case, we: introduce precisely the model; describe a path encoding for the configuration; give a Pitman-type transformation for the path encoding; and then explain how the latter operation yields the dynamics of the original model and enables us to solve an initial value problem (see Theorems \ref{udkdvthm}, \ref{dkdvthm}, \ref{udtodathm} and \ref{dtodathm}). As noted in the introduction, path encodings will be elements of the space $\mathcal{S}$ defined at \eqref{sdef}, and, as well as the subset $\mathcal{S}^{lin}$ defined at \eqref{slindef}, the following subset will also arise in our discussion:
\begin{equation}\label{s0def}
\mathcal{S}^0 :=\left\{ S \in\mathcal{S}\::\:S_0=0\right\}.
\end{equation}
We note that, although there is a parallel between the descriptions of the different integrable systems considered and we give a common argument for the proofs of the theorems, the four examples can be read independently.

\subsection{Ultra-discrete KdV equation}\label{udsec1} For a fixed $L\in\mathbb{R}$, the ultra-discrete KdV equation is described as follows, where the variables $(\eta_n^t,U_n^t)_{n,t\in\mathbb{Z}}$ can in general take values in $\mathbb{R}$,
though one can also consider specializations, some of which we discuss below (see Subsections \ref{specialsec}, \ref{knownsec}). For example, the original BBS of \cite{takahashi1990} corresponds to setting $L=1$ and restricting to configurations taking values in $\{0,1\}^\mathbb{Z}$.

\eqbox{Ultra-discrete KdV equation}{\begin{equation}\tag{udKdV}\label{UDKDV}
\begin{cases}
\eta_n^{t+1}& = \min\{L-\eta_n^t,U_{n-1}^t\},\\
U_n^{t} & =\eta_n^t+U_{n-1}^t-\eta_n^{t+1},
\end{cases}
\end{equation}
for all $n,t\in\mathbb{Z}$.}

By analogy with the box-ball model, we think of $\eta^t=(\eta^t_n)_{n\in\mathbb{Z}}$ as the configuration of the system at time $t$, and $U^t=(U^t_n)_{n\in\mathbb{Z}}$ as an auxiliary variable representing a carrier-type process, which captures how mass is shifted from left to right at time $t$. For certain boundary conditions, it is clear that the two-dimensional lattice equation \eqref{UDKDV} has a unique solution (see Subsection \ref{boundarysec} for discussion of this point).
However, for the initial value problem with boundary condition given by $\eta^0=\eta$ for some $\eta\in\mathbb{R}^\mathbb{Z}$, both the existence and the uniqueness of the solution is in question. Indeed, as per the comment in the introduction regarding the BBS, it is not true in general if one considers the evolution over a single time step. In this subsection, we tackle the problem for configurations $\eta$ that fall within the class
\begin{equation}\label{cudkdef}
\mathcal{C}_{udK}^{(L)}:=\left\{\eta \in \R^{\Z}\::\: \lim_{n \to \pm\infty}\frac{\sum_{m=1}^{n}\eta_m}{n}<\frac{L}{2}\right\},
\end{equation}
where we note that, if one is viewing $\eta_m$ as the mass of the configuration at site $m$, then the limit in the definition of $\mathcal{C}_{udK}^{(L)}$ can be interpreted as a density condition. In particular, the main result in this subsection, Theorem \ref{udkdvthm}, (partially) generalizes the construction of bi-infinite dynamics for the original BBS from \cite{CKST} to the more general \eqref{UDKDV} setting. As commented below Table \ref{Mtable}, and discussed in more detail in Subsection \ref{specialsec} below, although the past maximum operator that appears in the discussion below at \eqref{mveedef} is different to that at \eqref{originalM}, it gives the same dynamics for the original BBS.

We continue to introduce the path encoding for a configuration $\eta\in\mathbb{R}^\mathbb{Z}$. In particular, we define $S_{udK}^{(L)}:\mathbb{R}^\mathbb{Z}\rightarrow \mathcal{S}^0$ by taking $S=S_{udK}^{(L)}(\eta)$ to be the element of $\mathcal{S}^0$ given by
\begin{equation}\label{1pe}
S_{n}-S_{n-1}:=L-2\eta_n,\qquad \forall n \in\mathbb{Z}.
\end{equation}
Clearly this is a bijection, though it will also be convenient to extend the definition of inverse to the whole of $\mathcal{S}$. We do this by setting
\begin{equation}\label{extendinv}
\left(S_{udK}^{(L)}\right)^{-1}(S):=\left(S_{udK}^{(L)}\right)^{-1}(S-S_0),\qquad \forall S\in\mathcal{S}.
\end{equation}
We further note that $S_{udK}^{(L)}:\mathcal{C}_{udK}^{(L)}\rightarrow\mathcal{S}^0\cap\mathcal{S}^{lin}$ is a bijection. Moreover, $(S_{udK}^{(L)})^{-1}(\mathcal{S}^{lin})=\mathcal{C}_{udK}^{(L)}$.

As for the path dynamics, we set: for $S\in\mathcal{S}^{lin}$,
\begin{equation}\label{tveedef}
T^{\vee}(S):=2M^{\vee}(S)-S,
\end{equation}
where
\begin{equation}\label{mveedef}
M^{\vee}(S)_n:=\sup_{m \le n} \left(\frac{S_m+S_{m-1}}{2}\right).
\end{equation}
We will check later that $T^{\vee}:\mathcal{S}^{lin}\rightarrow\mathcal{S}^{lin}$ is a bijection (cf.\ Corollary \ref{c25}), and thus we further have that $(T^{\vee})^{-1}$ is well-defined. To describe the carrier, similarly to \eqref{originalW}, we write
\begin{equation}\label{wveedef}
W^{\vee}(S):=M^{\vee}(S)-S.
\end{equation}
However, to place the variables $(U_n^t)_{n,t\in\mathbb{Z}}$ on the same scale as the variables generated by such functions, we also introduce a simple bijection on $\mathbb{R}^\mathbb{Z}$ given by setting
\[W_{udK}^{(L)}(u)=\left(u_n-\frac{L}{2}\right)_{n\in\Z},\qquad \forall u\in\mathbb{R}^\Z.\]
With these preparations in place, we are now ready to state the main result of this subsection, which links the preceding operations to \eqref{UDKDV}, showing that the time evolution of the configuration is given by applying $T^{\vee}$ to its path encoding.

\begin{thm}\label{udkdvthm} If $\eta=(\eta_n)_{n\in\mathbb{Z}}\in\mathcal{C}_{udK}^{(L)}$, then there is a unique solution $(\eta_n^t,U_n^t)_{n,t\in\mathbb{Z}}$ to \eqref{UDKDV} that satisfies the initial condition $\eta^0=\eta$. This solution is given by setting, for $t\in\mathbb{Z}$,
\begin{align*}
\eta^t&:=\left(S_{udK}^{(L)}\right)^{-1}\circ \left(T^\vee\right)^t\circ S_{udK}^{(L)}(\eta),\\
U^t&:=\left(W_{udK}^{(L)}\right)^{-1}\circ W^{\vee}\circ \left(T^\vee\right)^t\circ S_{udK}^{(L)}(\eta).
\end{align*}
Or, more explicitly,
\begin{equation}\nonumber
\eta_n^t:=\frac{L-S_n^t+S_{n-1}^{t}}{2},\qquad U_n^t:=\frac{S_n^{t+1}-S_n^t+L}{2},\qquad \forall n,t\in\mathbb{Z},
\end{equation}
where $S^0:=S_{udK}^{(L)}(\eta)$, and $S^t:=(T^\vee)^t(S^0)$ for all $t\in\mathbb{Z}$. Moreover, $\eta^t\in\mathcal{C}_{udK}^{(L)}$ for all $t\in\mathbb{Z}$.
\end{thm}

\subsection{Discrete KdV equation}\label{dsec1} For a fixed $\delta\in(0,\infty)$, the discrete KdV equation has the following form, where the variables $(\omega_n^t,U_n^t)_{n,t\in\mathbb{Z}}$ take values in $(0,\infty)$. For a discussion of a related closed form equation involving only the variables $(\omega_n^t)_{n,t\in\mathbb{Z}}$, see Subsection \ref{cfsec} below.

\eqbox{Discrete KdV equation}{\begin{equation}\tag{dKdV}\label{DKDV}
\begin{cases}
\omega_n^{t+1} & =\left(\delta \omega_n^t +(U_{n-1}^t)^{-1}\right)^{-1},\\
U_n^t & = U_{n-1}^t\omega_n^t (\omega_n^{t+1})^{-1},
\end{cases}
\end{equation}
for all $n,t\in\mathbb{Z}$.}

As for the ultra-discrete model, for each time $t$, we will view the variables in terms of a configuration, which in this case is denoted $\omega^t$, and a carrier, which is denoted $U^t$. Moreover, the basic initial value problem is as in the previous subsection, namely to solve \eqref{DKDV} given a boundary condition of the form $\omega^0=\omega$ for some $\omega\in(0,\infty)^\mathbb{Z}$. We will do this for $\omega$ taking values in the following set:
\[\mathcal{C}_{dK}^{(\delta)}:=\left\{\omega \in (0,\infty)^{\Z}\::\: \lim_{n \to \pm\infty}\frac{\sum_{m=1}^{n}\log\omega_m}{n}<\frac{-\log\delta}{2}\right\},\]
where the limit can again be viewed as a density condition.

Regarding the path encoding, we define $S_{dK}^{(\delta)}:(0,\infty)^\mathbb{Z}\rightarrow \mathcal{S}^0$ by taking $S=S_{dK}^{(\delta)}(\omega)$ to be the element of $\mathcal{S}^0$ given by
\begin{equation}\label{2pe}
S_{n}-S_{n-1}:=-\log\delta-2\log\omega_n,\qquad \forall n \in\mathbb{Z}.
\end{equation}
Paralleling the properties of $S_{udK}^{(L)}$, we have that $S_{dK}^{(\delta)}$ is a bijection, and we extend its inverse to the whole of $\mathcal{S}$ similarly to \eqref{extendinv}. Moreover, $S_{dK}^{(\delta)}:\mathcal{C}_{dK}^{(\delta)}\rightarrow\mathcal{S}^0\cap\mathcal{S}^{lin}$ is a bijection, and $(S_{dK}^{(\delta)})^{-1}(\mathcal{S}^{lin})=\mathcal{C}_{dK}^{(\delta)}$.

Similarly to the path dynamics of the previous subsection, we set: for $S\in\mathcal{S}^{lin}$,
\[T^{\sum}(S):=2M^{\sum}(S)-S,\]
where
\[M^{\sum}(S)_n:=\log \left(\sum_{m \le n} \exp\left(\frac{S_m+S_{m-1}}{2}\right)\right).\]
We note that, just as $T^{\vee}$ is a discrete, two-step average version of Pitman's transformation, $T^{\sum}$ is a discrete, two-step average of the exponential version of Pitman's transformation, which was studied in \cite{MY1}, for example. Moreover, as for $T^\vee$, we will check later that $T^{\sum}:\mathcal{S}^{lin}\rightarrow\mathcal{S}^{lin}$ is bijection (cf.\ Corollary \ref{c25}), a result which implies that $(T^{\sum})^{-1}$ is well-defined. For describing the carrier, we write
\begin{equation}\label{wsumdef}
W^{\sum}(S):=M^{\sum}(S)-S,
\end{equation}
(see Figure \ref{dkdvcarrierfig} for a graphical example,) and introduce the following bijection from $(0,\infty)^\mathbb{Z}$ to $\mathbb{R}^\Z$:
\[W_{dK}^{(\delta)}(u)=\left(\log u_n+\frac{\log \delta}{2}\right)_{n\in\Z},\qquad \forall u\in(0,\infty)^\Z.\]
Our main result concerning \eqref{DKDV} is as follows.

\begin{figure}
\begin{center}
\includegraphics[width = 0.9\textwidth]{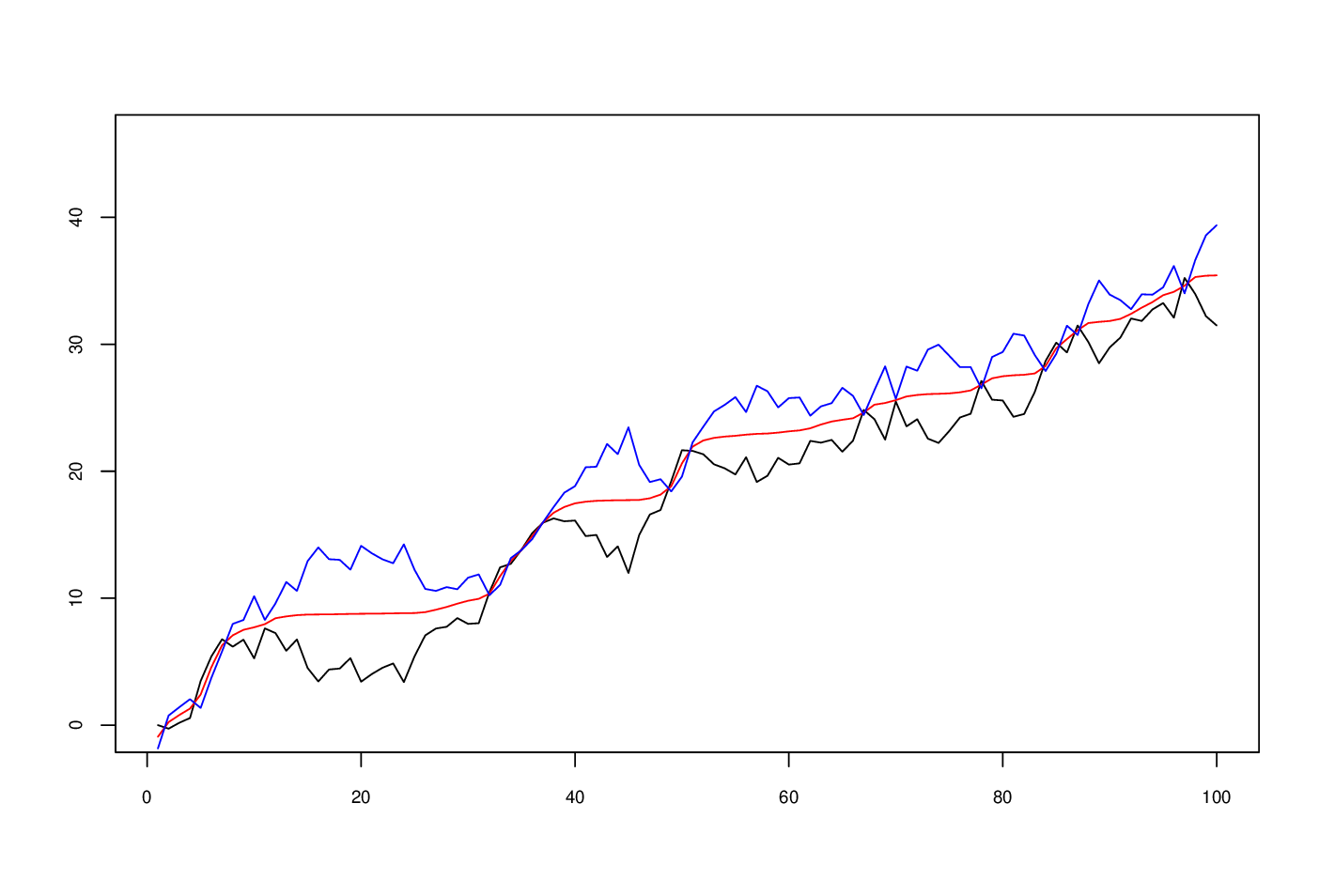}
\end{center}
\vspace{-40pt}
\begin{center}
\includegraphics[width = 0.9\textwidth]{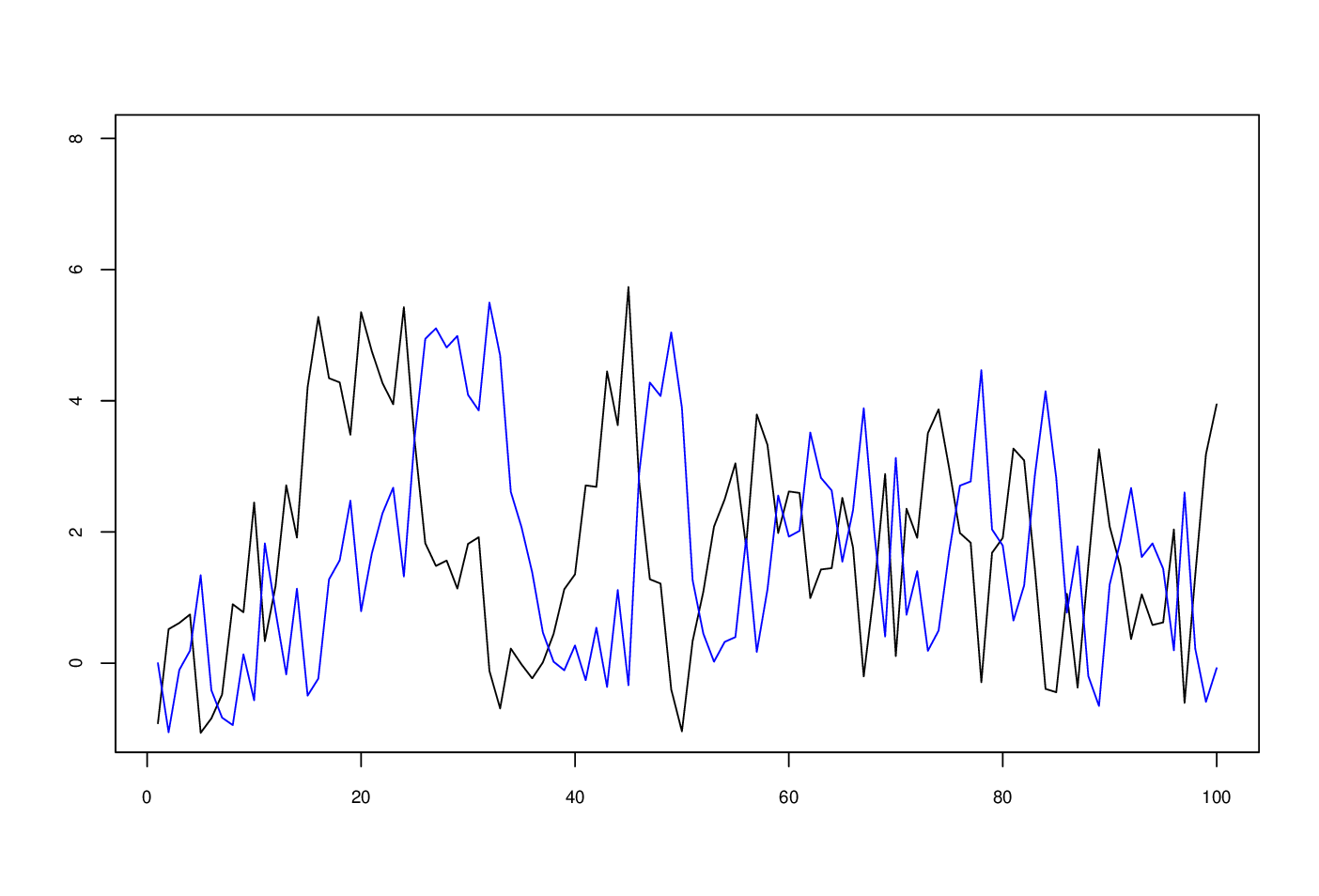}
\end{center}
\caption{The top figure shows an example sample path of $S$ (black), $M^{\sum}(S)$ (red) and $T^{\sum}(S)$ (blue). The bottom figure shows the corresponding carriers $W^{\sum}(S)$ (black) and $W^{\sum}(T^{\sum}(S))$.}\label{dkdvcarrierfig}
\end{figure}

\begin{thm}\label{dkdvthm} If $\omega=(\omega_n)_{n\in\mathbb{Z}}\in\mathcal{C}_{dK}^{(\delta)}$, then there is a unique solution $(\omega_n^t,U_n^t)_{n,t\in\mathbb{Z}}$ to \eqref{DKDV} that satisfies the initial condition $\omega^0=\omega$. This solution is given by setting, for $t\in\mathbb{Z}$,
\begin{align}
\omega^t&:=\left(S_{dK}^{(\delta)}\right)^{-1}\circ \left(T^{\sum}\right)^t\circ S_{dK}^{(\delta)}(\omega),\label{omegatdef}\\
U^t&:=\left(W_{dK}^{(\delta)}\right)^{-1}\circ W^{\sum}\circ \left(T^{\sum}\right)^t\circ S_{dK}^{(\delta)}(\omega).\nonumber
\end{align}
Or, more explicitly,
\begin{equation}\nonumber
\omega_n^t:=\exp\left(\frac{-\log\delta-S_n^t+S_{n-1}^{t}}{2}\right),\qquad U_n^t:=\exp\left(\frac{S_n^{t+1}-S_n^t-\log\delta}{2}\right),\qquad \forall n,t\in\mathbb{Z},
\end{equation}
where $S^0:=S_{dK}^{(\delta)}(\eta)$, and $S^t:=(T^{\sum})^t(S^0)$ for all $t\in\mathbb{Z}$. Moreover, $\omega^t\in\mathcal{C}_{dK}^{(\delta)}$ for all $t\in\mathbb{Z}$.
\end{thm}

\subsection{Ultra-discrete Toda equation}

We now turn to Toda-type equations, and the ultra-discrete Toda equation in particular, which has variables $(Q_n^t,E_n^t,U_n^t)_{n,t\in\mathbb{Z}}$ that take values in $\mathbb{R}$.

\eqbox{Ultra-discrete Toda (lattice) equation}{\begin{equation}\tag{udToda}\label{UDTODA}
\begin{cases}
Q_{n}^{t+1}=\min \{U_{n}^t,E_n^t\}, \\
E_{n}^{t+1}=Q_{n+1}^t+E_{n}^t-Q_{n}^{t+1},\\
U_{n+1}^t=U_{n}^t+Q_{n+1}^t-Q_{n}^{t+1},
\end{cases}
\end{equation}
for all $n,t\in\mathbb{Z}$.}

For the ultra-discrete Toda model, the configuration at time $t$ is given by the variables $Q^t=(Q^t_n)_{n\in\mathbb{Z}}$ and $E^t=(E^t_n)_{n\in\mathbb{Z}}$, with $Q^t_n$ representing the length of the $n$th interval containing mass, and $E^t_n$ representing the length of the $n$th empty interval. Of course, the terminology `length' only really makes physical sense if we restrict the variables to strictly positive values, see Theorem \ref{specialthm} below for treatment of this case. Other specific choices of configuration will also be discussed later (see Subsections \ref{specialsec}, \ref{knownsec}). As in previous cases, $U^t$ is viewed as an auxiliary variable representing a carrier process.
In particular, the initial value problem is to solve \eqref{UDTODA} given a boundary condition of the form $(Q^0,E^0)=(Q,E)$ for some $(Q,E)\in(\mathbb{R}^2)^\mathbb{Z}$. We will do this for $(Q,E)$ taking values in the following set:
\[\mathcal{C}_{udT}:=\left\{(Q,E)\in(\mathbb{R}^2)^\mathbb{Z}\::\:
 \begin{array}{l}
   \lim_{n \to \infty}\frac{\sum_{m=1}^{n}(Q_m-E_m)}{n} =  \lim_{n \to \infty}\frac{\sum_{m=1}^{n}(Q_m-E_m)+Q_{n+1}}{n} <0,\\
\lim_{n \to -\infty}\frac{\sum_{m=1}^{n}(Q_m-E_m)}{n}=  \lim_{n \to -\infty}\frac{\sum_{m=1}^{n}(Q_m-E_m)+E_{n}}{n} <0
 \end{array}\right\},\]
where the limit criteria can once again be viewed as a density condition.

Regarding the path encoding, we define $S_{udT}:(\mathbb{R}^2)^\mathbb{Z}\rightarrow \mathcal{S}^0$ by taking $S=S_{udT}(Q,E)$ to be the element of $\mathcal{S}^0$ given by
\begin{equation}\label{3pe}
S_{n}-S_{n-1}:=\left\{\begin{array}{ll}
                  -Q_{(n+1)/2}, & \mbox{$n$ odd},\\
E_{n/2}, & \mbox{$n$ even}.\\
                 \end{array}\right.
\end{equation}
Similarly to the previous cases, we have that $S_{udT}$ is a bijection, and we extend its inverse to the whole of $\mathcal{S}$ similarly to \eqref{extendinv}. Moreover, $S_{udT}:\mathcal{C}_{udT}\rightarrow\mathcal{S}^0\cap\mathcal{S}^{lin}$ is a bijection, and $(S_{udT})^{-1}(\mathcal{S}^{lin})=\mathcal{C}_{udT}$. We note that this path encoding of ultra-discrete Toda configurations is different to that used in \cite{CST}, which involved concatenating line segments of gradient alternating between $-1$ or $+1$, with lengths given by the elements of $Q$ and $E$, respectively (and is thus only applicable when $Q_n$ and $E_n$ take strictly positive values).

Reflecting the alternating nature of the definition of the path encoding, the dynamics involve the operator given by setting: for $S\in\mathcal{S}^{lin}$,
\[T^{\vee^*}(S):=2M^{\vee^*}(S)-S,\]
where
\[M^{\vee^*}(S)_{n}:=\left\{\begin{array}{ll}
 \sup_{m \le \frac{n-1}{2}}S_{2m}, & n\mbox{ odd}, \\
  \frac{M^{\vee^*}(S)_{n+1}+M^{\vee^*}(S)_{n-1}}{2}, &  n\mbox{ even}.
  \end{array}
\right.\]
It will further be convenient to introduce a shifted version of $T^{\vee^*}$, and so we define: for $S\in\mathcal{S}^{lin}$,
\[\mathcal{T}^{\vee^*}(S):=\theta\circ{T}^{\vee^*}(S),\]
where, as noted in the introduction, $\theta$ is the usual left-shift. We show later that $T^{\vee^*}:\mathcal{S}^{lin}\rightarrow\mathcal{S}^{lin}$ is bijection (cf.\ Corollary \ref{c25}), a result which implies that $(T^{\vee^*})^{-1}$ is well-defined, and clearly the same conclusions hold for $\mathcal{T}^{\vee^*}$. We also write
\begin{equation}\label{wveestardef}
W^{\vee^*}(S):=M^{\vee^*}(S)-S;
\end{equation}
in this case, we do not need a scale change to connect these to the carrier variables. Our main result concerning \eqref{UDTODA} is as follows, demonstrating that the dynamics of the configuration correspond to those of the path encoding under $\mathcal{T}^{\vee^*}$.

\begin{thm}\label{udtodathm} If $(Q,E)=(Q_n,E_n)_{n\in\mathbb{Z}}\in\mathcal{C}_{udT}$, then there is a unique solution $(Q_n^t,E_n^t,U_n^t)_{n,t\in\mathbb{Z}}$ to \eqref{UDTODA} that satisfies the initial condition $(Q^0,E^0)=(Q,E)$. This solution is given by setting, for $n,t\in\mathbb{Z}$,
\begin{align*}
(Q^t,E^t)&:=\left(S_{udT}\right)^{-1}\circ \left(\mathcal{T}^{\vee^*}\right)^t\circ S_{udT}(Q,E),\\
U^t_n&:=W^{\vee^*}\circ \left(\mathcal{T}^{\vee^*}\right)^t\circ S_{udT}(Q,E)_{2n-1}.
\end{align*}
Or, more explicitly,
\begin{equation}\nonumber
Q^t_n:=S^t_{2n-2}-S^t_{2n-1},\qquad E^t_n=S^t_{2n}-S^t_{2n-1},\qquad U^t_n:=\frac{S^{t+1}_{2n-2}-S^t_{2n-1}}{2},\qquad \forall n,t\in\mathbb{Z},
\end{equation}
where $S^0:=S_{udT}(Q,E)$, and $S^t:=(\mathcal{T}^{\vee^*})^t(S^0)$ for all $t\in\mathbb{Z}$. Moreover, $(Q^t,E^t)\in\mathcal{C}_{udT}$ for all $t\in\mathbb{Z}$.
\end{thm}

\begin{rem}
We highlight that, due to the nature of the path encoding in this case, whereby the pair $(E_n^t,Q_{n+1}^t)$ yields two steps of the path encoding, the carrier variable $U_n^t$ corresponds to the $n$th odd-indexed variable of the carrier process. The even-indexed variables of the carrier process appear as intermediate values in an enriched lattice system, see Corollary \ref{c57} (and the proof of the above result) for details. A similar comment applies in the case of the discrete Toda system, see Corollary \ref{c58} (and the proof of Theorem \ref{dtodathm}).
\end{rem}

\subsection{Discrete Toda equation}\label{dsec1toda}

As our fourth and final discrete integrable system, we come to the discrete Toda equation, which involves variables $(I_n^t,J_n^t,U_n^t)_{n,t\in\mathbb{Z}}$ that take values in $(0,\infty)$. For a discussion of a related closed form equation involving only the variables $(I_n^t,J_n^t)_{n,t\in\mathbb{Z}}$, see Subsection \ref{cfsec} below.

\eqbox{Discrete Toda (lattice) equation}{\begin{equation}\tag{dToda}\label{DTODA}
\begin{cases}
I_n^{t+1}=J_n^t+U_n^t,\\
J_n^{t+1}={I_{n+1}^{t}J_n^{t}}(I_n^{t+1})^{-1},\\
U_{n+1}^t={I_{n+1}^{t}U_n^{t}}(I_n^{t+1})^{-1},
\end{cases}
\end{equation}
for all $n,t\in\mathbb{Z}$.}

Similarly to the ultra-discrete Toda model, the configuration at time $t$ is given by the variables $I^t=(I^t_n)_{n\in\mathbb{Z}}$ and $J^t=(J^t_n)_{n\in\mathbb{Z}}$, and $U^t$ is an auxiliary variable. The initial value problem is to solve \eqref{DTODA} given a boundary condition of the form $(I^0,J^0)=(I,J)$ for some $(I,J)\in((0,\infty)^2)^\mathbb{Z}$. We will do this for $(I,J)$ taking values in the following set:
\begin{align*}
\lefteqn{\mathcal{C}_{dT}:=}\\
&\left\{(I,J)\in((0,\infty)^2)^\mathbb{Z}\::\:
 \begin{array}{l}
   \lim_{n \to \infty}\frac{\sum_{m=1}^{n}(\log J_m-\log I_m)}{n} =  \lim_{n \to \infty}\frac{\sum_{m=1}^{n}(\log J_m-\log I_m)-\log I_{n+1}}{n} <0,\\
\lim_{n \to -\infty}\frac{\sum_{m=1}^{n}(\log J_m-\log I_m)}{n}=  \lim_{n \to -\infty}\frac{\sum_{m=1}^{n}(\log J_m-\log I_m)-\log J_{n}}{n} <0
 \end{array}\right\},
 \end{align*}
where the limit criteria can once again be viewed as a density condition.

Regarding the path encoding, we define $S_{dT}:((0,\infty)^2)^\mathbb{Z}\rightarrow \mathcal{S}^0$ by taking $S=S_{dT}(I,J)$ to be the element of $\mathcal{S}^0$ given by
\begin{equation}\label{4pe}
S_{n}-S_{n-1}:=\left\{\begin{array}{ll}
                  \log I_{(n+1)/2}, & \mbox{$n$ odd},\\
-\log J_{n/2}, & \mbox{$n$ even}.\\
                 \end{array}\right.
\end{equation}
As before, we have that $S_{dT}$ is a bijection, and we extend its inverse to the whole of $\mathcal{S}$ similarly to \eqref{extendinv}. Moreover, $S_{dT}:\mathcal{C}_{dT}\rightarrow\mathcal{S}^0\cap\mathcal{S}^{lin}$ is a bijection, and $(S_{dT})^{-1}(\mathcal{S}^{lin})=\mathcal{C}_{dT}$.

For the discrete Toda system, the dynamics involve the operator given by setting: for $S\in\mathcal{S}^{lin}$,
\[T^{\sum^*}(S):=2M^{\sum^*}(S)-S,\]
where
\[M^{\sum^*}(S)_n:=\left\{\begin{array}{ll}
 \log \left(\sum_{m \le \frac{n-1}{2}} \exp\left(S_{2m}\right) \right), & n\mbox{ odd}, \\
  \frac{M^{\sum^*}(S)_{n+1}+M^{\sum^*}(S)_{n-1}}{2}, &  n\mbox{ even}.
  \end{array}
\right.\]
As in the ultra-discrete Toda case, we further introduce a shifted version of $T^{\sum^*}$, and so we define: for $S\in\mathcal{S}^{lin}$,
\[\mathcal{T}^{\sum^*}(S):=\theta\circ{T}^{\sum^*}(S),\]
where $\theta$ is again the left-shift. We show in later that $T^{\sum^*}:\mathcal{S}^{lin}\rightarrow\mathcal{S}^{lin}$ is bijection (cf.\ Corollary \ref{c25}), a result which implies that $(T^{\sum^*})^{-1}$ is well-defined, and clearly the same conclusions hold for $\mathcal{T}^{\sum^*}$. For the carrier, we also write
\begin{equation}\label{wsumstardef}
W^{\sum^*}(S):=M^{\sum^*}(S)-S,
\end{equation}
and define a scale change from $(0,\infty)^\Z$ to $\mathbb{R}^\Z$ by setting
\[W_{dT}(u)=\left(-\log u_n\right)_{n\in\Z},\qquad \forall u\in(0,\infty)^{\Z}.\]
Completing our collection of results concerning the solution of initial value problems for discrete integrable systems, we have the following.

\begin{thm}\label{dtodathm} If $(I,J)=(I_n,J_n)_{n\in\mathbb{Z}}\in\mathcal{C}_{dT}$, then there is a unique solution $(I_n^t,J_n^t,U_n^t)_{n,t\in\mathbb{Z}}$ to \eqref{DTODA} that satisfies the initial condition $(I^0,J^0)=(I,J)$. This solution is given by setting, for $n,t\in\mathbb{Z}$,
\begin{align}
(I^t,J^t)&:=\left(S_{dT}\right)^{-1}\circ \left(\mathcal{T}^{\sum^*}\right)^t\circ S_{dT}(I,J),\label{ijij}\\
U^t_n&:=\left(W_{dT}\right)^{-1}\circ W^{\sum^*}\circ \left(\mathcal{T}^{\sum^*}\right)^t\circ S_{dT}(I,J)_{2n-1}.\nonumber
\end{align}
Or, more explicitly, for $n,t\in\mathbb{Z}$,
\begin{equation}\nonumber
I^t_n:=\exp\left(S^t_{2n-1}-S^t_{2n-2}\right),\qquad
 J^t_n=\exp\left(-\left(S^t_{2n}-S^t_{2n-1}\right)\right),\qquad
U^t_n:=\exp\left(\frac{S^t_{2n-1}-S^{t+1}_{2n-2}}{2}\right),
\end{equation}
where $S^0:=S_{dT}(I,J)$, and $S^t:=(\mathcal{T}^{\sum^*})^t(S^0)$ for all $t\in\mathbb{Z}$. Moreover, $(I^t,J^t)\in\mathcal{C}_{dT}$ for all $t\in\mathbb{Z}$.
\end{thm}

\section{Lattice maps and their symmetries}\label{lmsec}

We have at several points above used the terminology `lattice equation' in describing the \eqref{UDKDV}, \eqref{DKDV}, \eqref{UDTODA} and \eqref{DTODA} systems. Towards formulating our general framework, as we do in the next section, here we make the lattice picture more precise for the systems of interest, with Table \ref{lmtable} (see below)
presenting details of the lattice structure and corresponding dynamics. In particular, using the notation of Table \ref{lmtable}, and ordering components of the maps as follows:
\begin{center}
\begin{tabular}{cc}
 \(\boxed{F(\cdot,\cdot)}\hspace{-15pt}\xymatrix@C-15pt@R-15pt{ & F^{(1)}(a,b) & \\
            b \ar[rr]& & F^{(2)}(a,b),\\
             & a \ar[uu]&}\) & \(\boxed{F(\cdot,\cdot,\cdot)}\hspace{-15pt}\xymatrix@C-15pt@R-15pt{ &F^{(1)}(a,b,c) &F^{(2)}(a,b,c)  &\\
            c \ar[rrr]& & & F^{(3)}(a,b,c), \\
& b \ar[uu]&a\ar[uu]&}\) \\
\end{tabular}
\end{center}
\noindent
we can reexpress \eqref{UDKDV}, \eqref{DKDV}, \eqref{UDTODA} and \eqref{DTODA} in the following compact way:
\[ \left.\begin{array}{ll}
   \hfill \eqref{UDKDV}:  & {(\eta_n^{t+1},U_n^t)=F_{udK}^{(L)}(\eta_n^t,U_{n-1}^t);} \medskip\\
   \hfill \eqref{DKDV}:   & {(\omega_n^{t+1},U_n^t)=F_{dK}^{(\delta)}(\omega_n^t,U_{n-1}^t);} \medskip\\
    \hfill\eqref{UDTODA}: & {(Q_{n}^{t+1},E_n^{t+1},U_{n+1}^t)=F_{udT}(Q_{n+1}^t,E_n^t,U_{n}^t);} \medskip\\
    \hfill\eqref{DTODA}:  & {(I_{n}^{t+1},J_n^{t+1},U_{n+1}^t)=F_{dT}(I_{n+1}^t,J_n^t,U_{n}^t).}
  \end{array}\right.\]
In the remainder of this section, we highlight some of the natural symmetries of these four maps, which will be relevant to solving initial value problems. Before doing this, however, we further remark that it is possible to decompose the Toda lattice structure from a single map with three inputs and three outputs to two maps, each with two inputs and two outputs (we postpone the details until Section \ref{proofsec}, see Lemmas \ref{ls1} and \ref{ls2}, and Remarks \ref{rs1} and \ref{rs2}, in particular). The latter picture will allow us to place both the KdV- and Toda-type systems into the same general framework that is developed in Section \ref{general}.

\begin{table}
\begin{center}
\begin{tabular}{r|c|c l}
  \emph{Model} & \emph{Lattice structure} & \multicolumn{2}{l}{\emph{Local dynamics}} \\
  \hline
  {udKdV} & \(\displaystyle\xymatrix@C-15pt@R-15pt{ & \eta_n^{t+1} & \\
            U_{n-1}^t \ar[rr]& & U_n^t\\
             & \eta_n^t \ar[uu]&}\) & \framebox(25,22){$F_{udK}^{(L)}$}&\hspace{-5pt}\(\displaystyle\xymatrix@C-15pt@R-15pt{ & \min\{L-a,b\} & \\
            b \ar[rr]& & ^{a+b}_{-\min\{L-a,b\}}\\
             & a \ar[uu]&}\)  \\
             \hline
  {dKdV} & \(\displaystyle\xymatrix@C-15pt@R-15pt{ & \omega_n^{t+1}  & \\
            U_{n-1}^t \ar[rr]& & U_n^t\\
             & \omega_n^t\ar[uu]&}\) & \framebox(25,22){$F_{dK}^{(\delta)}$}&\hspace{-5pt}\(\displaystyle\xymatrix@C-15pt@R-15pt{ & \frac{1}{\delta a+ b^{-1}} & \\
            b \ar[rr]& & ab\left(\delta a+b^{-1}\right)\\
             & a \ar[uu]&}\) \\
             \hline
  {udToda} & \(\displaystyle\xymatrix@C-15pt@R-15pt{ & Q_n^{t+1} & E_n^{t+1}& \\
            U_n^{t} \ar[rrr] & && U_{n+1}^t\\
             & E_n^{t}\ar[uu]& Q_{n+1}^{t}\ar[uu]&}\) &  \framebox(25,22){$F_{udT}$}&\hspace{-5pt}\(\displaystyle\xymatrix@C-15pt@R-15pt{ & \min\{b,c\} & ^{a+b}_{-\min\{b,c\}}& \\
            c \ar[rrr] & && ^{a+c}_{-\min\{b,c\}}\\
             & b \ar[uu]&a\ar[uu]&}\)\\
             \hline
  {dToda} & \(\displaystyle\xymatrix@C-15pt@R-15pt{ &I_n^{t+1}&J_n^{t+1} &\\
            U_n^{t} \ar[rrr]& & & U_{n+1}^{t}\\
& J_n^t \ar[uu]&I_{n+1}^t\ar[uu]&}\)  &  \framebox(25,22){$F_{dT}$}&\hspace{-5pt}\(\displaystyle\xymatrix@C-15pt@R-15pt{ &b+c&\frac{ab}{b+c} &\\
            c \ar[rrr]& & & \frac{ac}{b+c}\\
& b \ar[uu]&a\ar[uu]&}\)
\end{tabular}
\end{center}
\caption{Lattice structures for the \eqref{UDKDV}, \eqref{DKDV}, \eqref{UDTODA} and \eqref{DTODA} systems, and the corresponding local dynamics.}\label{lmtable}
\end{table}

\subsection{Self-inverse maps}\label{sisec} The first symmetry we discuss is that the maps are all involutions, i.e.\ they are all bijections, and it holds that
\[\left(F_{udK}^{(L)}\right)^{-1}=F_{udK}^{(L)},\qquad \left(F_{dK}^{(\delta)}\right)^{-1}=F_{dK}^{(\delta)},\qquad \left(F_{udT}\right)^{-1}=F_{udT},\qquad\left(F_{dT}\right)^{-1}=F_{dT};\]
checking this involves elementary computations, which are omitted. Consequently, if we reverse the directions of the arrows in the lattice structures shown in Table \ref{lmtable}, then, in each case, the corresponding local dynamics are described by the same map as in the original system. It follows that the backwards in time evolution can be understood in exactly the same way as the forward in time evolution. In Subsection \ref{fbeq} below, we explain how this observation can be applied to reduce the all-time initial value problems considered in Section \ref{dissec} to forward and backward versions. Moreover, a general treatment of the time-reversal of a system of locally-defined dynamics and the relevance of a system being self-reverse to solving corresponding initial value problems is given in Subsection \ref{s42}.

\subsection{Conserved quantities}\label{conssec} Being integrable systems, each of the models \eqref{UDKDV}, \eqref{DKDV}, \eqref{UDTODA} and \eqref{DTODA} admits many globally conserved quantities. Additionally, for each of the maps $F_{udK}^{(L)}$, $F_{dK}^{(\delta)}$, $F_{udT}$ and $F_{dT}$, one can identify a number of locally-defined quantities that are preserved, and we now introduce those that will be important for this article.
\begin{description}
  \item[\eqref{UDKDV}] For $F_{udK}^{(L)}$, one has that $a+b$ is conserved, i.e.
  \[\left(F_{udK}^{(L)}\right)^{(1)}(a,b)+\left(F_{udK}^{(L)}\right)^{(2)}(a,b)=a+b,\]
  where we write $F_{udK}^{(L)}(a,b)=((F_{udK}^{(L)})^{(1)}(a,b),(F_{udK}^{(L)})^{(2)}(a,b))$. Using the viewpoint that $a$ represents the original mass at the relevant lattice location, $b$ represents the mass brought to the site by the carrier, $(F_{udK}^{(L)})^{(1)}(a,b)$ represents the mass left after the carrier has passed, and $(F_{udK}^{(L)})^{(2)}(a,b)$ represents the mass moved onwards by the carrier, we see that the preservation of $a+b$ can be interpreted as a conservation of mass property. (Of course, we allow the variables $a$ and $b$ to take negative values, so this is only a heuristic in general.)
  \item[\eqref{DKDV}] For $F_{dK}^{(\delta)}$, one readily sees that $ab$ is conserved. Equivalently, we have that $\log a +\log b$ is conserved, and so if we view the $\log$-transformed variables as representing mass, then we can again view this as a conservation of mass.
  \item[\eqref{UDTODA}] For $F_{udT}$, one has that both $a+b$ and $a+c$ are conserved. Both of these can be interpreted physically, with $a+b$ representing the `length' of the spatial interval to which the local dynamics applies, and $a+c$ representing the mass (assuming mass is placed at a unit density on non-empty intervals). The conserved quantity that will arise naturally in our study is a combination of these, specifically being given by $(a+b)-2(a+c)=b-a-2c$.
  \item[\eqref{DTODA}] For $F_{dT}$, one similarly has that both $ab$ and $ac$ are conserved, and again taking $\log$-transformations allows a physical interpretation in terms of `length' and `mass' corresponding to that of the ultra-discrete Toda lattice. Moreover, the conserved quantity that will be most relevant to us is given by $\log b-\log a-2\log c$.
\end{description}

As we will set out in Subsection \ref{s44}, for bijections $F:\mathbb{R}^2\rightarrow\mathbb{R}^2$ given by $(a,b)\mapsto F(a,b)$ for which $a-2b$ is a conserved quantity, we have a natural way to relate the associated lattice dynamics to a Pitman-type transformation of a certain path encoding of the configuration. The path-encoding picture turns out to be extremely useful for identifying a `canonical carrier' for configurations within a certain class, and this allows us to identify solutions to the corresponding initial value problems. By making appropriate changes of variables, we can map each of the systems \eqref{UDKDV}, \eqref{DKDV}, \eqref{UDTODA} and \eqref{DTODA} into the latter setting; the conserved quantities that are relevant for the four discrete integrable systems in question are discussed further in Remark \ref{consrem}. As a result, we obtain the connection between the models of interest and the Pitman-type transformations of Table \ref{Mtable}, and a means for establishing our main results (see Subsection \ref{s44} and Section \ref{proofsec} for details).

\subsection{Duality}

A further symmetry that we will not explore here, but which relates to our results, is the duality of the KdV systems under the exchange of the role of the configuration and carrier. Indeed, reflecting the lattices in the diagonal gives dual lattices with structure and dynamics given by
\begin{equation}\label{duala}
\xymatrix@C-15pt@R-15pt{ & U_n^t& \\
           \eta_n^t \ar[rr]& & \eta_n^{t+1}, \\
             &  U_{n-1}^t \ar[uu]&}\qquad\xymatrix@C-15pt@R-15pt{ &^{a+b}_{-\min\{L-b,a\}} & \\
            b \ar[rr]& &  \min\{L-b,a\},\\
             & a \ar[uu]&}
             \end{equation}
in the ultra-discrete case, and
\begin{equation}\label{dualb}
\xymatrix@C-15pt@R-15pt{ & U_n^t & \\
            \omega_n^t\ar[rr]& & \omega_n^{t+1}, \\
             &  U_{n-1}^t\ar[uu]&}\qquad\xymatrix@C-15pt@R-15pt{ & ab\left(a^{-1}+\delta b\right) & \\
            b \ar[rr]& &\frac{1}{a^{-1}+\delta b},\\
             & a\ar[uu]&}
             \end{equation}
in the discrete case.

In contrast to the \eqref{UDKDV} system that we study here, which has configuration capacity $L$ and carrier capacity $\infty$, the system at \eqref{duala} represents the corresponding model with configuration capacity $\infty$ and carrier capacity $L$. Given this relationship, it is clear that the solutions we construct to \eqref{UDKDV} in Theorem \ref{udkdvthm} also give solutions to the dual model with the initial condition being the state of the carrier at spatial location 0 (rather than the state of the configuration at time 0, as in the original model). The property of duality was applied in \cite{CS} to explore the invariant measures of BBS($J$,$K$) for $J,K\in\mathbb{N}\cup\{\infty\}$, that is, box-ball systems of box capacity $J$ and carrier capacity $K$, for which BBS($J$,$K$) and BBS($K$,$J$) are dual. Interestingly, the path encoding of \cite{CKST} (and this article) is only appropriate for BBS($J$,$K$) with $J<K$, and so \cite{CS} applied duality to understand the dual systems, i.e.\ those with $J>K$. Since the  BBS($J$,$K$) is a special case of a two-parameter version of the ultra-discrete KdV equation, it is natural to extend the results of \cite{CS} to the more general model, and such a study forms part of the follow-up work \cite{CSirf}.

For the \eqref{DKDV} system, again one has a two-parameter version \eqref{DKDV}$(\alpha,\beta)$, with our original model corresponding to the \eqref{DKDV}$(\delta,0)$ model, and the system at \eqref{dualb} corresponding to \eqref{DKDV}$(0,\delta)$. Hence, similarly to the ultra-discrete case, Theorem \ref{dkdvthm} yields solutions to \eqref{DKDV}$(0,\delta)$ with initial condition being the state of the carrier at spatial location 0. In \cite{CSirf}, duality between the invariant measures of  \eqref{DKDV}$(\delta,0)$ and  \eqref{DKDV}$(0,\delta)$ are explored. For the discrete systems, we are only able to give a useful path encoding for the model \eqref{DKDV}$(\delta,0)$, and it remains an interesting question as to whether one can say anything in this direction for \eqref{DKDV}$(\alpha,\beta)$ with $\alpha>\beta>0$.

Finally, given the specific structure of the Toda lattice, the relevance of duality is less clear. Whilst one might conceive of techniques for constructing two auxiliary carrier variables for each configuration variable in the dual model, it is not clear to us how the spatial shift that the dynamics incorporates should be handled.

\section{Pitman-type transformation maps and path encodings}\label{general}

The aim of this section is to develop the framework that will allow us to prove the main results of the article, as presented in Section \ref{dissec}. In particular, we introduce a general definition of locally-defined dynamics on a configuration space, and relate solutions of a corresponding initial value problem to the existence of what we term a canonical carrier process -- Subsection \ref{s41} deals with the forward problem (see Theorem \ref{thm:uni-ex-f-gp}), and in Subsection \ref{s42} we introduce the time-reversal of the system, which allows us to extend the techniques to cover both forward and backward time evolution (see Theorem \ref{thm:uni-ex-f-gp-2}). Moreover, in Subsection \ref{s43}, we go on to describe how solutions to initial value problems can be transferred to related systems via a change of coordinates (see Theorem \ref{thm:uni-ex-f-gp-3}). The principal weakness of the aforementioned results is that they depend on the identification of a canonical carrier. This issue is addressed in Subsection \ref{s44}, where a criteria is provided for a special class of locally-defined dynamics in terms of path encodings of configurations (see Assumption \ref{a1}), and it is further explained how the dynamics of the system can be seen as a Pitman-type transformations on path space in these cases (see Theorems \ref{t16} and \ref{t17}). Finally, in Subsection \ref{s45}, we apply the results to four important examples of locally-defined dynamics (see Corollary \ref{c25}), which correspond to unparameterized versions of the \eqref{UDKDV}, \eqref{DKDV}, \eqref{UDTODA} and \eqref{DTODA} systems, as will be detailed in the subsequent section.

\subsection{Initial value problem for locally-defined dynamics}\label{s41}

Let us start by introducing the initial value problem that is the focus of this subsection. Firstly, we consider a configuration space $\mathcal{X}:= \Pi_{ n \in \mathbb{Z}} \mathcal{X}_n$ given as an infinite product of sets; configurations will typically be written as $x=(x_n)_{n\in\mathbb{Z}} \in \mathcal{X}$. Secondly, the dynamics will be locally-defined, according to a collection of maps that satisfy the following definition. Note that this definition involves the factors of a space $\mathcal{U}:= \Pi_{ n \in \mathbb{Z}} \mathcal{U}_n$, in which realizations of a carrier process will exist (for a precise description of a carrier and the associated dynamics, see Definition \ref{carrierdef} below).

\begin{df}\label{ldd}
We say that we have dynamics on $\mathcal{X}$ that are \emph{locally-defined} if they are given by a shift parameter $m\in\mathbb{Z}$ and collection of maps $(F_n)_{n\in\mathbb{Z}}$ such that, for each $n$,
\[F_n: \mathcal{X}_n \times \mathcal{U}_{n-1} \to \mathcal{X}_{n-m} \times \mathcal{U}_{n}\]
is a bijection.
\end{df}

Throughout this subsection, we will suppose that we have locally-defined dynamics given by $m$ and $(F_n)_{n\in\mathbb{Z}}$, as per the preceding definition, and study the existence and uniqueness of solutions to the following forward problem: given $x\in \mathcal{X}$, find $((x^t,u^t))_{t\in\mathbb{Z}_+}\in(\mathcal{X}\times\mathcal{U})^{\mathbb{Z}_+}$ such that
\begin{equation}\label{IVP}
\begin{cases}
x^0=x, \\
\left(x^{t+1}_{n-m},u^t_n\right)=F_n\left(x^{t}_{n},u^t_{n-1}\right),\qquad \forall n \in \Z,\:\ t \in \Z_+.
\end{cases}
\end{equation}
We note that in our applications to KdV-type discrete integrable systems, we will take $m=0$, and $\mathcal{X}_n$, $\mathcal{U}_n$ and $F_n$ will not depend on $n$. (In particular, for the various versions of BBS covered by our framework, $F_n$ is given by the combinatorial $R$, as described in \cite[Section 2.2]{KLO}, for example, with the order of the output components exchanged.) The principal motivation for presenting the generalized setting is to handle Toda-type systems, for which we will take $m=1$, and $\mathcal{X}_n$, $\mathcal{U}_n$ and $F_n$ will alternate between odd and even $n$.
Our solution to \eqref{IVP} is presented as Theorem \ref{thm:uni-ex-f-gp} below, and will be expressed in terms of `canonical carrier functions', which we now define. In the following definition, we also describe the dynamics associated with a particular carrier.

\begin{df}\label{carrierdef}
(a) We say $u\in \mathcal{U}$ is a \emph{carrier} for $x\in\mathcal{X}$ if
\[u_n=F_n^{(2)}\left(x_n,u_{n-1}\right),\qquad\forall n\in\mathbb{Z},\]
where we write $F_n(x_n,u_{n-1})=(F_n^{(1)}(x_n,u_{n-1}),F_n^{(2)}(x_n,u_{n-1}))$. The associated dynamics $x\mapsto \mathcal{T}^ux$ are given by setting
\[\left(\mathcal{T}^ux\right)_{n-m}=F_n^{(1)}\left(x_n,u_{n-1}\right).\]
(b) We say $U:\mathcal{X}^U\rightarrow\mathcal{U}$, where $\mathcal{X}^U\subseteq\mathcal{X}$, is a \emph{carrier function} if, for all $x\in \mathcal{X}^U$, it holds that $U(x)$ is a carrier for $x$.\\
(c) We say $U:\mathcal{X}^U\rightarrow\mathcal{U}$, where $\mathcal{X}^U\subseteq\mathcal{X}$, is a \emph{canonical carrier function} if it is a carrier function and moreover the following two properties are satisfied:\\
(i) it holds that $\mathcal{T}(\mathcal{X}^U)\subseteq \mathcal{X}^U$, where
\begin{equation}\label{curlyt}
\mathcal{T}(x):=\mathcal{T}^{U(x)}x,\qquad \forall x\in \mathcal{X}^U;
\end{equation}
(ii) for any $x \in \mathcal{X}^U$ and $u\in\mathcal{U}$ such that $u$ is a carrier for $x$, if $u \neq U(x)$, then there is no carrier for $\mathcal{T}^ux$.
\end{df}

\begin{rem} The definition of `canonical' above differs from the corresponding definitions in \cite{CKST, CS}. Indeed, in the latter papers (which concern the BBS and its finite capacity versions), more transparent conditions were given, and justified on the basis of their physical relevance. Nonetheless, for the systems where multiple definitions apply, the resulting carrier is the same under each definition.
\end{rem}

Clearly, property (i) in the definition of a canonical carrier function yields that for $x\in\mathcal{X}^U$, it is possible to iterate the dynamics given $\mathcal{T}$ repeatedly to obtain $\mathcal{T}^t(x)$ for any $t\in\mathbb{Z}_+$. Moreover, property (ii) ensures that the choice of carrier $U(x)$ for $x\in\mathcal{X}^U$ is literally `canonical', in that for any other choice of carrier, we can not continue to define the dynamics beyond a single time step. As a consequence, we find that in the case a canonical carrier function $U:\mathcal{X}^U\rightarrow\mathcal{U}$ exists, the only solution to \eqref{IVP} with initial data $x\in\mathcal{X}^U$ is given by the dynamics associated with $U$. This is the content of the following theorem, which also explains why we choose not to add a superscript $U$ in the definition of $\mathcal{T}$ at \eqref{curlyt}.

\begin{thm}\label{thm:uni-ex-f-gp} Suppose that $U:\mathcal{X}^U\rightarrow\mathcal{U}$, where $\mathcal{X}^U\subseteq\mathcal{X}$, is a canonical carrier function for the locally-defined dynamics given by $m\in\mathbb{Z}$ and maps $(F_n)_{n\in\mathbb{Z}}$. It then holds that, for each $x\in\mathcal{X}^U$, the initial value problem at \eqref{IVP} has a unique solution, which is given by setting
\begin{equation}\label{xudef}
x^t:=\mathcal{T}^t(x),\qquad u^t:=U(x^t),\qquad \forall t\in\mathbb{Z}_+,
\end{equation}
where $\mathcal{T}$ is defined as at \eqref{curlyt}.
\end{thm}
\begin{proof} By Definition \ref{carrierdef}, it is clear that $((x^t,u^t))_{t\in\mathbb{Z}_+}$, as defined at \eqref{xudef}, is a solution to \eqref{IVP}. We will establish the uniqueness claim recursively. If $((\tilde{x}^t,\tilde{u}^t))_{t\in\mathbb{Z}_+}$ is a solution of \eqref{IVP}, then $\tilde{u}^0$ is a carrier for $\tilde{x}^0=x\in\mathcal{X}^{U}$. Moreover, $\tilde{x}^1=\mathcal{T}^{\tilde{u}^0}x$ has a carrier $\tilde{u}^1$. Hence, by property (ii) of a canonical carrier, we must have that $\tilde{u}^0=U(x)=u^0$. It follows that $\tilde{x}^1=x^1$, and repeating the same argument gives the conclusion.
\end{proof}

\subsection{Reversibility and all-time solutions to the initial value problem}\label{s42}

The aim of this subsection is to extend the discussion of the previous subsection to negative times $t$, and thereby identify solutions to the initial value problem that are valid for all times (forward and backward), see Theorem \ref{thm:uni-ex-f-gp-2} in particular. Throughout we will suppose we have spaces $\mathcal{X}$ and $\mathcal{U}$, and locally-defined dynamics given by a shift parameter $m\in\mathbb{Z}$ and collection of maps $(F_n)_{n\in\mathbb{Z}}$, as in Definition \ref{ldd}. We start by considering the following backward problem: given $x\in \mathcal{X}$, set $x^0=x$ and find $((x^{t},u^{t}))_{t\in\mathbb{Z}_-}\in(\mathcal{X}\times\mathcal{U})^{\mathbb{Z}_-}$ such that
\begin{equation}\label{bIVP}
\left(x^{t+1}_{n-m},u^t_n\right)=F_n\left(x^{t}_{n},u^t_{n-1}\right),\qquad \forall n \in \Z,\:\ t \in \Z_-,
\end{equation}
where we recall the convention that $\mathbb{Z}_-:=-\mathbb{N}=\{\dots,-2,-1\}$. We will approach this via a corresponding reversed problem, as introduced in the subsequent definition.

\begin{df} (a) The \emph{reversed configuration space} is given by $\mathcal{X}^R:=\prod_{n\in\mathbb{Z}}\mathcal{X}^R_n$, where $\mathcal{X}_n^R:=\mathcal{X}_{m+1-n}$. We define a corresponding \emph{configuration reversal operator} $R^\mathcal{X}:\mathcal{X}\rightarrow\mathcal{X}^R$ by setting
\[R^\mathcal{X}(x)_n:=x_{m+1-n},\qquad \forall n\in\mathbb{Z}.\]
(b) The \emph{reversed carrier space} is given by $\mathcal{U}^R:=\prod_{n\in\mathbb{Z}}\mathcal{U}^R_n$, where $\mathcal{U}_n^R:=\mathcal{U}_{2m-n}$. We define a corresponding \emph{carrier reversal operator} $R^\mathcal{U}:\mathcal{U}\rightarrow\mathcal{U}^R$ by setting
\[R^\mathcal{U}(u)_n:=u_{2m-n},\qquad \forall n\in\mathbb{Z}.\]
(c) The \emph{reversed locally-defined dynamics} are given by shift parameter $m$ and maps $(F_n^R)_{n\in\mathbb{Z}}$, where
\[F_n^R: \mathcal{X}^R_n \times \mathcal{U}^R_{n-1} \to \mathcal{X}^R_{n-m} \times \mathcal{U}^R_{n}\]
is defined by $F_n^R:=F_{2m+1-n}^{-1}$.\\
(d) The \emph{forward problem for the reversed locally-defined dynamics} is as follows: given $x\in \mathcal{X}^R$, find $((x^t,u^t))_{t\in\mathbb{Z}_+}\in(\mathcal{X}^R\times\mathcal{U}^R)^{\mathbb{Z}_+}$ such that
\begin{equation}\label{rIVP}
\begin{cases}
x^0=x, \\
\left(x^{t+1}_{n-m},u^t_n\right)=F^R_n\left(x^{t}_{n},u^t_{n-1}\right),\qquad \forall n \in \Z,\:\ t \in \Z_+.
\end{cases}
\end{equation}
\end{df}

Specifically, given initial data $x^0$, the backward problem at \eqref{bIVP} is to determine the variables in the lower-half plane of the following lattice:
\[\xymatrix@C-15pt@R-15pt{     & \boxed{F_m}            & x_0^0           &         &      &  \boxed{F_{n+m}}  & x_n^0  & &      \\
                          \dots&  u_{m-1}^{-1}  \ar[rr] &                 &u_m^{-1} &\dots &  u_{n+m-1}^{-1}\ar[rr] &        & u_{n+m}^{-1}&\dots.  \\
                               &                        & x_m^{-1} \ar[uu]&         &      &   & x_{n+m}^{-1}\ar[uu]& &\\
                               &&\vdots&&&&\vdots&&}\]
Reversing the directions of the arrows, rotating by $180^\circ$ about $x_0^0$, and shifting by $m+1$ positions horizontally to the left, this problem is equivalent to determining the variables in the upper-half plane of the picture:
\[\xymatrix@C-15pt@R-15pt{  &&\vdots&&&&\vdots&&\\
   & \boxed{F_{2m+1}^{-1}}            & x_{2m+1}^{-1}           &         &      &  \boxed{F_{2m+1-n}^{-1}}  & x_{2m+1-n}^{-1}  & &      \\
                          \dots&  u_{2m+1}^{-1}  \ar[rr] &                 &u_{2m}^{-1} &\dots &  u_{2m+1-n}^{-1}\ar[rr] &        & u_{2m-n}^{-1}&\dots.  \\
                               &                        & x_{m+1}^{0} \ar[uu]&         &      &   & x_{m+1-n}^{0}\ar[uu]& &}\]
It is then an elementary exercise in relabelling to rewrite the variables and maps as:
\[\xymatrix@C-15pt@R-15pt{  &&\vdots&&&&\vdots&&\\
   & \boxed{F_0^R}            & R^\mathcal{X}(x^{-1})_{-m}           &         &      &  \boxed{F_{n}^R}  & R^\mathcal{X}(x^{-1})_{n-m}  & &      \\
                          \dots\hspace{-15pt}&  R^\mathcal{U}(u^{-1})_{-1} \hspace{-2pt} \ar[rr] &                 & R^\mathcal{U}(u^{-1})_{0} \hspace{-9pt}&\hspace{-22pt}\dots\hspace{-22pt} &\hspace{-9pt}  R^\mathcal{U}(u^{-1})_{n-1}\hspace{-2pt}\ar[rr] &        & R^\mathcal{U}(u^{-1})_{n}&\hspace{-15pt}\dots.  \\
                               &                        & R^\mathcal{X}(x^0)_{0} \ar[uu]&         &      &   & R^\mathcal{X}(x^0)_{n}\ar[uu]& &}\]
In this manner, one arrives at the following proposition.

\begin{prop}\label{fbprop} The sequence $((x^{t},u^{t}))_{t\in\mathbb{Z}_-}\in(\mathcal{X}\times\mathcal{U})^{\mathbb{Z}_-}$ is a solution to the backward problem \eqref{bIVP} with initial condition $x^0\in\mathcal{X}$ if and only if the sequence $((x^{R,t},u^{R,t}))_{t\in\mathbb{Z}_+}\in(\mathcal{X}\times\mathcal{U})^{\mathbb{Z}_+}$ is a solution to the forward problem for the reversed locally-defined dynamics \eqref{rIVP} with initial condition $R^\mathcal{X}(x^0)\in\mathcal{X}^R$, where
\[\left(x^{R,t},u^{R,t}\right):=\left(R^{\mathcal{X}}(x^{-t}),R^{\mathcal{U}}(u^{-t-1})\right),\qquad \forall t\in\mathbb{Z}_+.\]
\end{prop}

From this result, we obtain the following corollary, which demonstrates that if a canonical carrier function for the reversed dynamics $U^R:\mathcal{X}^{R,U^R}\rightarrow\mathcal{U}^R$ exists, the only solution to \eqref{bIVP} with initial data $x\in (R^\mathcal{X})^{-1}(\mathcal{X}^{R,U^R})$ is given by the dynamics associated with $U^R$.

\begin{cor}\label{bivpun} Suppose that $U^R:\mathcal{X}^{R,U^R}\rightarrow\mathcal{U}^R$, where $\mathcal{X}^{R,U^R}\subseteq\mathcal{X}^R$, is a canonical carrier function for the reversed locally-defined dynamics given by $m\in\mathbb{Z}$ and maps $(F^R_n)_{n\in\mathbb{Z}}$. It then holds that, for each $x\in (R^\mathcal{X})^{-1}(\mathcal{X}^{R,U^R})$, the backward problem at \eqref{bIVP} has a unique solution, which is given by setting
\[x^t:=(R^\mathcal{X})^{-1}\circ(\mathcal{T}^R)^{-t}\circ R^\mathcal{X}(x),\qquad  u^t:=(R^\mathcal{U})^{-1}\circ U^R\circ(\mathcal{T}^R)^{-t-1}\circ R^\mathcal{X}(x),\qquad \forall t\in\mathbb{Z}_-,\]
where $\mathcal{T}^R(\tilde{x}):=\mathcal{T}^{U^R(\tilde{x})}(\tilde{x})$ for $\tilde{x}\in\mathcal{X}^{R,U^R}$.
\end{cor}
\begin{proof} By Theorem \ref{thm:uni-ex-f-gp}, we have that the unique solution to the forward problem for the reversed locally-defined dynamics \eqref{rIVP} with initial condition $R^\mathcal{X}(x)\in \mathcal{X}^{R,U^R}$ is given by
\[x^{R,t}:=(\mathcal{T}^R)^t\circ R^\mathcal{X} (x),\qquad u^{R,t}:=U^R(x^{R,t}),\qquad \forall t\in\mathbb{Z}_+.\]
Hence, by Proposition \ref{fbprop}, we obtain that, for each $x\in (R^\mathcal{X})^{-1}(\mathcal{X}^{R,U^R})$, the backward problem at \eqref{bIVP} has a unique solution, which is given by setting
\[x^t:=(R^\mathcal{X})^{-1}(x^{R,-t}),\qquad u^t:=(R^\mathcal{U})^{-1}(u^{R,-t-1}),\qquad \forall t\in\mathbb{Z}_-,\]
and the result follows.
\end{proof}

In the next result, we relate the original and reversed dynamics when we have canonical carriers for both that operate on corresponding parts of the original and reversed configuration spaces (in the sense made precise by \eqref{compatible}). In particular, we show that the original and reversed dynamics are essentially inverses, and obtain identities linking the original and reversed carriers. See Remark \ref{graphical} below for a graphical presentation of the result.

\begin{prop}\label{ppp} Suppose that $U:\mathcal{X}^{U}\rightarrow\mathcal{U}$, where $\mathcal{X}^{U}\subseteq\mathcal{X}$, is a canonical carrier function for the original dynamics, and $U^R:\mathcal{X}^{R,U^R}\rightarrow\mathcal{U}^R$, where $\mathcal{X}^{R,U^R}\subseteq\mathcal{X}^R$, is a canonical carrier function for the reversed dynamics, such that
\begin{equation}\label{compatible}
R^{\mathcal{X}}\left(\mathcal{X}^U\right)=\mathcal{X}^{R,U^R}.
\end{equation}
It then holds that $\mathcal{T}$ is a bijection from $\mathcal{X}^U$ to itself, and
\begin{equation}\label{trel}
\mathcal{T}^{-1}=(R^{\mathcal{X}})^{-1}\circ\mathcal{T}^R\circ R^{\mathcal{X}}.
\end{equation}
Similarly, $\mathcal{T}^R$ is a bijection from $\mathcal{X}^{R,U^R}$ to itself, and
\[(\mathcal{T}^R)^{-1}=R^{\mathcal{X}}\circ\mathcal{T}\circ (R^{\mathcal{X}})^{-1}.\]
Moreover,
\begin{equation}\label{urel}
U(x)=(R^{\mathcal{X}})^{-1}\circ U^R \circ R^{\mathcal{X}}\circ\mathcal{T}(x),\qquad \forall x\in\mathcal{X}^U,
\end{equation}
and similarly
\[U^R(x)=R^{\mathcal{X}}\circ U \circ (R^{\mathcal{X}})^{-1}\circ\mathcal{T}^R(x),\qquad \forall x\in\mathcal{X}^{R,U^R}.\]
\end{prop}
\begin{proof} For $x\in\mathcal{X}^U$, we have that $((x^t,u^t))_{t\in\mathbb{Z}_+}$, as defined at \eqref{xudef}, is the unique solution to the initial value problem at \eqref{IVP} with initial condition $x$. Moreover, $x^1=\mathcal{T}(x)\in (R^{\mathcal{X}})^{-1}(\mathcal{X}^{R,U^R})$, and so by the uniqueness of the solution to the backward problem at \eqref{bIVP} given by Corollary \ref{bivpun}, we find that
\begin{equation}\label{star1}
x=(R^{\mathcal{X}})^{-1}\circ\mathcal{T}^R\circ R^{\mathcal{X}}\circ\mathcal{T}(x),\qquad \forall x\in\mathcal{X}^U.
\end{equation}
Similarly,
\begin{equation}\label{star2}
x=R^{\mathcal{X}}\circ\mathcal{T}\circ (R^{\mathcal{X}})^{-1}\circ\mathcal{T}^R(x),\qquad \forall x\in\mathcal{X}^{R,U^R}.
\end{equation}
From \eqref{star1}, the injectivity of $\mathcal{T}:\mathcal{X}^U\rightarrow\mathcal{X}^U$ is clear. As for the surjectivity, let $x\in\mathcal{X}^U$. Then $R^\mathcal{X}(x)\in \mathcal{X}^{R,U^R}$, and so \eqref{star2} implies
\[R^{\mathcal{X}}\circ\mathcal{T}\circ (R^{\mathcal{X}})^{-1}\circ\mathcal{T}^R\circ{R}^\mathcal{X}(x)={R}^\mathcal{X}(x),\]
which is equivalent to
\[\mathcal{T}\circ (R^{\mathcal{X}})^{-1}\circ\mathcal{T}^R\circ{R}^\mathcal{X}(x)=x.\]
Since $(R^{\mathcal{X}})^{-1}\circ\mathcal{T}^R\circ{R}^\mathcal{X}(x)\in \mathcal{X}^U$, this confirms that $\mathcal{T}$ is indeed surjective. The bijectivity of $\mathcal{T}^R$ is established in a similar manner.

For the carrier identities, we again compare the results of Theorem \ref{thm:uni-ex-f-gp} and Corollary \ref{bivpun}. In particular, these results yield that, for $x\in \mathcal{X}^U$,
\[U(x)=u^0=(R^{\mathcal{X}})^{-1}\circ U^R \circ (\mathcal{T}^R)^0 \circ R^{\mathcal{X}}(x^1)=(R^{\mathcal{X}})^{-1}\circ U^R \circ R^{\mathcal{X}}\circ\mathcal{T}(x),\]
and the final part of the proposition is obtained similarly.
\end{proof}

\begin{rem}\label{graphical} The original dynamics started from $x$ have the following form:
\[\xymatrix@C-15pt@R-15pt{&\boxed{F_0}&\mathcal{T}(x)_{-m}&\boxed{F_{1}}&\mathcal{T}(x)_{1-m}& & &  \\
                          \dots& U(x)_{-1} \ar[rr] &&U(x)_0\ar[rr]  & & U(x)_1&\dots.  \\
                               &                   & x_0 \ar[uu]&   & x_{1}\ar[uu] &&}\]
Rotating the picture and reversing the direction of the arrows gives the reversed dynamics, starting from $R^\mathcal{X}\circ\mathcal{T}(x)$:
\[\xymatrix@C-15pt@R-15pt{&\boxed{F^R_0}&R^\mathcal{X}(x)_{-m}&\boxed{F^R_{1}}&R^\mathcal{X}(x)_{1-m}& &\\
                          \dots& R^\mathcal{U}\circ U(x)_{-1} \ar[rr] &&R^\mathcal{U}\circ U(x)_0\ar[rr]  & & R^\mathcal{U}\circ U(x)_1&\dots.  \\
                               &                   &R^\mathcal{X}\circ\mathcal{T}( x)_0 \ar[uu]&   & R^\mathcal{X}\circ\mathcal{T}(x)_{1}\ar[uu]&&}\]
In particular, we see that $\mathcal{T}^R\circ R^\mathcal{X}\circ\mathcal{T}(x)=R^\mathcal{X}(x)$, which yields \eqref{trel}. Moreover, one can read off that $U^R\circ R^\mathcal{X}\circ\mathcal{T}(x)=R^\mathcal{X}\circ U(x)$, from which we readily obtain \eqref{urel}. The other relationships given in Proposition \ref{ppp} can be explained similarly.
\end{rem}

To complete this subsection, we consider the special case when the original and reversed locally-defined dynamics are the same, which will cover each of the four integrable systems \eqref{DKDV}, \eqref{UDKDV}, \eqref{DTODA} and \eqref{UDTODA}. In particular, for finite configurations in the original BBS, the result recovers the observation
that the inverse of the dynamics is given by running the carrier from right to left, rather than from left to right, which was essentially made in \cite{takahashi1990}.

\begin{thm}\label{thm:uni-ex-f-gp-2} Suppose that $m$ and $(F_n)_{n\in\mathbb{Z}}$ describe self-reverse locally-defined dynamics, i.e.\ $\mathcal{X}^R=\mathcal{X}$, $\mathcal{U}^R=\mathcal{U}$, and $F_n^R=F_n$ for all $n\in\mathbb{Z}$. Moreover, suppose that $U:\mathcal{X}^U\rightarrow\mathcal{U}$, where $\mathcal{X}^U\subseteq\mathcal{X}$, is a canonical carrier function for the locally-defined dynamics that satisfies
\[R^{\mathcal{X}}(\mathcal{X}^U)=\mathcal{X}^U.\]
The following statements then hold.\\
(a) The operator $\mathcal{T}$ is a bijection from $\mathcal{X}^U$ to itself, and
\[\mathcal{T}^{-1}=R^{\mathcal{X}}\circ\mathcal{T}\circ R^{\mathcal{X}}.\]
(b) On $\mathcal{X}^U$, it holds that
\[U=R^{\mathcal{X}}\circ U \circ R^{\mathcal{X}}\circ\mathcal{T}.\]
(c) For each $x\in\mathcal{X}^U$, the initial value problem
\[\begin{cases}
x^0=x, \\
\left(x^{t+1}_{n-m},u^t_n\right)=F_n\left(x^{t}_{n},u^t_{n-1}\right),\qquad \forall n \in \Z,\:\ t \in \Z.
\end{cases}\]
has unique solution $((x^t,u^t))_{t\in\Z}$ given by
\[x^t:=\mathcal{T}^t(x),\qquad u^t:=U(x^t),\qquad \forall t\in\mathbb{Z}.\]
\end{thm}
\begin{proof} Setting $U^R=U$, parts (a) and (b) are easy consequences of Proposition \ref{ppp}. As for part (c), the forward part of the claim, i.e.\ for $t\in\mathbb{Z}_+$, is given by Theorem \ref{thm:uni-ex-f-gp}. For the backward part, we apply Corollary \ref{bivpun} and parts (a) and (b) to obtain, for $t\in\mathbb{Z}_-$,
\[x^t=R^\mathcal{X}\mathcal{T}^{-t}R^\mathcal{X}(x)=\left(R^\mathcal{X}\mathcal{T}R^\mathcal{X}\right)^{-t}(x)=\left(\mathcal{T}^{-1}\right)^{-t}(x)=\mathcal{T}^{t}(x),\]
and
\[u^t=R^\mathcal{X}U\mathcal{T}^{-t-1}R^{\mathcal{X}}(x)=R^\mathcal{X}UR^\mathcal{X}\mathcal{T}^{t+1}(x)=U(x^t),\]
which completes the proof.
\end{proof}

\subsection{Change of coordinates}\label{s43} In our study of the dynamics of the systems \eqref{DKDV}, \eqref{UDKDV}, \eqref{DTODA} and \eqref{UDTODA}, it will be convenient to first consider parameter-free versions of the models obtained via a change of coordinates. To apply the results we obtain for the latter models to the original setting, we will appeal to the main result of this subsection (see Theorem \ref{thm:uni-ex-f-gp-3} below), which describes how solutions of initial value problems of the kind discussed hitherto can be transferred under such a change of coordinates.

To present our conclusion, we now suppose that, in addition to the objects introduced so far, we have spaces $\tilde{\mathcal{X}}=\prod_{n\in\mathbb{Z}}\tilde{\mathcal{X}}_n$ and $\tilde{\mathcal{U}}=\prod_{n\in\mathbb{Z}}\tilde{\mathcal{U}}_n$, and also maps
\begin{eqnarray*}
\mathcal{A}:\tilde{\mathcal{X}}&\rightarrow&\mathcal{X}\\
\tilde{x}&\mapsto&\left(\mathcal{A}_n(\tilde{x}_n)\right)_{n\in\mathbb{Z}},
\end{eqnarray*}
\begin{eqnarray*}
\mathcal{B}:\tilde{\mathcal{U}}&\rightarrow&{\mathcal{U}}\\
\tilde{u}&\mapsto&\left(\mathcal{B}_n(\tilde{u}_n)\right)_{n\in\mathbb{Z}},
\end{eqnarray*}
where $\mathcal{A}_n:\tilde{\mathcal{X}}_n\rightarrow{\mathcal{X}}_n$, $\mathcal{B}_n:\tilde{\mathcal{U}}_n\rightarrow{\mathcal{U}}_n$, $n\in\mathbb{Z}$ are injections. For the statement of the result, we note that the restrictions $\mathcal{A}:\tilde{\mathcal{X}}\rightarrow\mathcal{A}(\tilde{\mathcal{X}})$  and $\mathcal{B}:\tilde{\mathcal{U}}\rightarrow\mathcal{B}(\tilde{\mathcal{U}})$ are in fact bijections with well-defined inverses $\mathcal{A}^{-1}$ and $\mathcal{B}^{-1}$, respectively. We moreover suppose we have locally-defined dynamics on $\tilde{\mathcal{X}}$ given by $m$ and $(\tilde{F}_n)_{n\in\mathbb{Z}}$ such that
\[\left(\mathcal{A}_{n-m}\times\mathcal{B}_n\right)\circ \tilde{F}_n={F}_n\circ\left(\mathcal{A}_n\times\mathcal{B}_{n-1}\right),\qquad \forall n\in\mathbb{Z},\]
where $(\mathcal{A}_k\times\mathcal{B}_l)(\tilde{x}_k,\tilde{u}_l):=(\mathcal{A}_k(\tilde{x}_k),\mathcal{B}_l(\tilde{u}_l))$ for all $\tilde{x}_k\in\tilde{\mathcal{X}}_k$, $\tilde{u}_l\in\tilde{\mathcal{U}}_l$.

\begin{thm}\label{thm:uni-ex-f-gp-3} (a) Suppose that ${U}:{\mathcal{X}}^{{U}}\rightarrow{\mathcal{U}}$, where ${\mathcal{X}}^{{U}}\subseteq{\mathcal{X}}$, is a canonical carrier function for the locally-defined dynamics on ${\mathcal{X}}$ given by $m$ and $({F}_n)_{n\in\mathbb{Z}}$ such that
\begin{equation}\label{inc}
{\mathcal{X}}^{{U}}\subseteq\mathcal{A}\left(\tilde{\mathcal{X}}\right),\qquad U\left({\mathcal{X}}^{{U}}\right)\subseteq\mathcal{B}\left(\tilde{\mathcal{U}}\right).
\end{equation}
For each $\tilde{x}\in\mathcal{A}^{-1}({\mathcal{X}}^{{U}})$, the initial value problem
\begin{equation}\label{anotherivp}
\begin{cases}
\tilde{x}^0=\tilde{x}, \\
\left(\tilde{x}^{t+1}_{n-m},\tilde{u}^t_n\right)=\tilde{F}_n\left(\tilde{x}^{t}_{n},\tilde{u}^t_{n-1}\right),\qquad \forall n \in \Z,\:\ t \in \Z_+.
\end{cases}
\end{equation}
has unique solution $((\tilde{x}^t,\tilde{u}^t))_{t\in\Z_+}$ given by
\begin{equation}\label{eee}
\tilde{x}^t:=(\mathcal{A}^{-1}\circ{\mathcal{T}}\circ\mathcal{A})^t(\tilde{x}),\qquad \tilde{u}^t:=\mathcal{B}^{-1}\circ{U}\circ\mathcal{A}(\tilde{x}^t),\qquad \forall t\in\mathbb{Z}_+,
\end{equation}
where ${\mathcal{T}}$ is the operator on $\mathcal{X}^U$ giving the dynamics associated with $U$.\\
(b) Suppose that the locally-defined dynamics on $\mathcal{X}$ given by $m$ and $(F_n)_{n\in\mathbb{Z}}$ are self-reverse (in the sense described in the statement of Theorem \ref{thm:uni-ex-f-gp-2}), and $U:\mathcal{X}^U\rightarrow\mathcal{U}$, where $\mathcal{X}^U\subseteq\mathcal{X}$, is an associated canonical carrier function that satisfies $R^{\mathcal{X}}(\mathcal{X}^U)=\mathcal{X}^U$. If \eqref{inc} holds, then for each $\tilde{x}\in\mathcal{A}^{-1}({\mathcal{X}}^{{U}})$, the all-time initial value problem given by \eqref{anotherivp}, but replacing $t\in\mathbb{Z}_+$ by $t\in\mathbb{Z}$, has unique solution $((\tilde{x}^t,\tilde{u}^t))_{t\in\Z}$ given by \eqref{eee}, again replacing $t\in\mathbb{Z}_+$ by $t\in\mathbb{Z}$.
\end{thm}
\begin{proof} Towards proving part (a), let $(({x}^t,{u}^t))_{t\in\mathbb{Z}_+}$ be the unique solution of the initial value forward problem for the for the locally-defined dynamics on ${\mathcal{X}}$ given by $m$ and $({F}_n)_{n\in\mathbb{Z}}$ with initial condition ${x}:=\mathcal{A}(\tilde{x})$, as described in Theorem \ref{thm:uni-ex-f-gp}. Defining $((\tilde{x}^t,\tilde{u}^t))_{t\in\Z_+}$, by \eqref{eee} we have that $(\tilde{x}^t,\tilde{u}^t)=(\mathcal{A}^{-1}({x}^t),\mathcal{B}^{-1}({u}^t))$, and so
\begin{eqnarray*}
\left(\mathcal{A}_{n-m}\times\mathcal{B}_n\right)\circ \tilde{F}_n\left(\tilde{x}_n^t,\tilde{u}_{n-1}^t\right)&=&
\left(\mathcal{A}_{n-m}\times\mathcal{B}_n\right)\circ \tilde{F}_n\left(\mathcal{A}^{-1}({x}^t)_n,\mathcal{B}^{-1}({u}^t)_{n-1}\right)\\
&=&{F}_n\left({x}^t_n,{u}^t_{n-1}\right)\\
&=&\left({x}^{t+1}_{n-m},{u}^t_{n}\right)\\
&=&\left(\mathcal{A}_{n-m}\times\mathcal{B}_n\right)\left(\tilde{x}^{t+1}_{n-m},\tilde{u}^t_{n}\right).
\end{eqnarray*}
Since $\mathcal{A}_{n-m}$ and $\mathcal{B}_n$ are injective, we thus obtain that $\tilde{F}_n(\tilde{x}_n^t,\tilde{u}_{n-1}^t)=(\tilde{x}^{t+1}_{n-m},\tilde{u}^t_{n})$. Moreover, it clearly holds that $\tilde{x}^0=\tilde{x}$, and so we have confirmed that $((\tilde{x}^t,\tilde{u}^t))_{t\in\Z_+}$ is indeed a solution to \eqref{anotherivp}. Similar manipulations allow us to check that if $((\bar{x}^t,\bar{u}^t))_{t\in\Z_+}$ is an arbitrary solution to \eqref{anotherivp}, then $((\mathcal{A}(\bar{x}^t),\mathcal{B}(\bar{u}^t)))_{t\in\Z_+}$ solves the corresponding problem on ${\mathcal{X}}$. We must therefore have that $((\mathcal{A}(\bar{x}^t),\mathcal{B}(\bar{u}^t)))_{t\in\Z_+}$ is equal to $(({x}^t,{u}^t))_{t\in\mathbb{Z}_+}$, and this implies in turn that $((\bar{x}^t,\bar{u}^t))_{t\in\Z_+}$ is equal to $((\tilde{x}^t,\tilde{u}^t))_{t\in\Z_+}$. This completes the proof of part (a), and the proof of part (b) is essentially the same, only appealing to Theorem \ref{thm:uni-ex-f-gp-2} in place of Theorem \ref{thm:uni-ex-f-gp}.
\end{proof}

\subsection{Pitman-type locally-defined dynamics and path encodings}\label{s44}

In this subsection, we specialize to the case when $\mathcal{X}=\mathcal{U}=\mathbb{R}^\mathbb{Z}$, and the locally-defined dynamics are given by Pitman-type transformation maps, as introduced in the following definition. For this setting, we present a criteria for identifying a canonical carrier function in terms of a path encoding for the configuration (see Assumption \ref{a1} and Proposition \ref{linearsp}), and show that the associated dynamics can be described in terms of Pitman-type transformations (see Theorems \ref{t16} and \ref{t17}).

\begin{df}\label{P-mapdef}
(a) We say $F: \mathbb{R}^2 \to \mathbb{R}^2$ is a \emph{Pitman-type transformation map (P-map)} if the following statements hold:\\
(i) $F$ is a bijection;\\
(ii) $P \circ F =P$, where
\[P(a,b):=a-2b,\qquad \forall (a,b) \in \mathbb{R}^2.\]
(b) We say locally-defined dynamics on $\mathbb{R}^\mathbb{Z}$ given by shift parameter $m$ and maps $(F_n)_{n\in\mathbb{Z}}$ are of \emph{Pitman-type} if $F_n$ is a P-map for each $n\in\mathbb{Z}$.
\end{df}

It will be useful to note for later that if $F$ is a P-map, then so is $F^{-1}$. Moreover, although the importance of the condition in (a)(ii) above may not be clear at this point, it represents a natural conservation property in the integrable systems that we are interested in (recall the discussion from Subsection \ref{conssec} and see also Remark \ref{consrem} below). For example, for the original BBS of \cite{takahashi1990}, $a$ will represent $1-2\eta^t_n$ and $b$ will represent $U^t_{n-1}-\frac{1}{2}$, where we recall $\eta^t_n$ is the configuration and $U^t_{n-1}$ is the carrier load arriving at the relevant space-time point (see Subsection \ref{515151} for details), meaning that the conservation of $a-2b$ under the relevant $F$ is equivalent to having the conservation of mass: $\eta^t_n+U^t_{n-1}=\eta^{t+1}_n+U_n^t$. Moreover, the conservation property of condition (a)(ii) will be the key to the connection with Pitman-type transformations. Towards introducing these, we next define the path encoding of a configuration, which is a certain anti-derivative, and the associated dynamics. We recall the space $\mathcal{S}^0$ from \eqref{s0def}.

\begin{df}\label{pidef} (a) The \emph{path encoding} of a configuration $x \in \R^{\Z}$ is a function $S^x\in\mathcal{S}^0$ defined by
\[S^x_n-S^x_{n-1}=x_n.\]
We let $\pi:\R^{\Z}\rightarrow\mathcal{S}^0$ be the map given by $\pi(x)=S^x$, and note that this is clearly a bijection.\\
(b) If $u\in\R^{\Z}$ is a carrier for $x\in\R^{\Z}$, then define the associated path-encoding dynamics $S^x\mapsto\mathcal{T}^u_SS^x$ by setting
\[\mathcal{T}_S^uS^x:=S^{\mathcal{T}^{u}(x)},\]
i.e.\ $\mathcal{T}_S^u:=\pi \circ\mathcal{T}^{u}\circ\pi^{-1}$.
\end{df}

In the following definition, we introduce the Pitman-type transformation of the path encoding to which a particular carrier gives rise. The function $M^u$ plays the role that the past maximum did in the original version of Pitman's transformation, as recalled at \eqref{originalpitman}.

\begin{df} For $u\in\R^{\Z}$, let $M^u:\mathcal{S}\rightarrow\mathcal{S}$ be the operator given by
\[M^u(S)_n=S_n+u_n,\qquad \forall n\in\mathbb{Z}.\]
Given this, we introduce the associated \emph{Pitman-type transformation} $T^u:\mathcal{S}\rightarrow\mathcal{S}$ by setting
\[T^u(S)=2M^u(S)-S.\]
\end{df}

Part (b) of the next lemma provides the connection between the conservation law that is assumed to hold for Pitman-type locally-defined dynamics and the Pitman-type transformation on path encodings. Part (a) will be useful for obtaining an explicit expression for carrier functions when we study concrete examples in Subsection \ref{s45}.

\begin{lem}\label{lem:pitman}
(a) It holds that $u\in\mathbb{R}^\mathbb{Z}$ is a carrier for $x\in\mathbb{R}^\mathbb{Z}$ if and only if
\begin{equation}\label{eqeqeq}
M^u(S^x)_n=S^x_n+F_n^{(2)}\left(S^x_n-S^x_{n-1}, M^u(S^x)_{n-1}-S^x_{n-1}\right),\qquad \forall n\in\mathbb{Z}.
\end{equation}
(b) For Pitman-type locally-defined dynamics on $\mathbb{R}^\mathbb{Z}$, if $u\in\mathbb{R}^\mathbb{Z}$ is a carrier for $x\in\mathbb{R}^\mathbb{Z}$, then
\begin{equation}\label{pitu}
\mathcal{T}^u_S\left(S^x\right)=\theta^m\circ T^u(S^x)-\theta^m\circ T^u(S^x)_0.
\end{equation}
\end{lem}
\begin{proof} Applying the definitions of $S^x$ and $M^u$, it is straightforward to check that the collection of equations $u_n=F_n^{(2)}(x_n,u_{n-1})$, $n\in\mathbb{Z}$, is equivalent to \eqref{eqeqeq}, which establishes (a). For part (b), we start by noting that since $\mathcal{T}_S^u(S^x)_0=0=\theta^m\circ T^u(S^x)_0-\theta^m\circ T^u(S^x)_0$, we only need to prove
\[\mathcal{T}_S^u\left(S^x\right)_n-\mathcal{T}_S^u\left(S^x\right)_{n-1}=\theta^m\circ T^u(S^x)_n-\theta^m\circ T^u(S^x)_{n-1},\qquad \forall n\in\Z.\]
By the definition of $T^u$, we have that
\begin{eqnarray*}
\theta^m\circ T^u(S^x)_n-\theta^m\circ T^u(S^x)_{n-1}&=&2M^u(S^x)_{n+m}-S^x_{n+m}-2M^u(S^x)_{n+m-1}+S^x_{n+m-1}\\
&=&S^x_{n+m}+2u_{n+m}-S^x_{n+m-1}-2u_{n+m-1}\\
&=&x_{n+m}+2u_{n+m}-2u_{n+m-1}.
\end{eqnarray*}
Now, since we are assuming Pitman-type dynamics, we have that
\[x_{n+m}-2u_{n+m-1}=\mathcal{T}^u(x)_{n}-2u_{n+m},\]
and so we obtain
\[\theta^m\circ T^u(S^x)_n-\theta^m\circ T^u(S^x)_{n-1}=\mathcal{T}^u(x)_{n}=\mathcal{T}_S^u\left(S^x\right)_n-\mathcal{T}_S^u\left(S^x\right)_{n-1},\]
as desired.
\end{proof}

For a given system of Pitman-type locally-defined dynamics on $\mathbb{R}^\mathbb{Z}$, we now introduce the assumption that, as we will show in the following proposition, ensures the existence of a canonical carrier function on a certain subset of the configuration space. In order to describe the latter, we present some notation. Specifically, recall $\mathcal{S}$ from \eqref{sdef}, and, for $c>0$, let
\[\mathcal{S}^{lin(c)}_-:=\left\{S\in \mathcal{S}\::\:\lim_{n \to - \infty}\frac{S_n}{n}=c\right\},\qquad \mathcal{S}^{lin(c)}_+:=\left\{S\in \mathcal{S}\::\:\lim_{n \to+ \infty}\frac{S_n}{n}=c\right\},\]
be the set of asymptotically linear path encodings with gradient $c$ at $-\infty$ or $+\infty$, respectively, and
\[\mathcal{S}^{lin}_-:=\bigcup_{c>0}\mathcal{S}^{lin(c)}_-,\qquad\mathcal{S}^{lin}_+:=\bigcup_{c>0}\mathcal{S}^{lin(c)}_+,\qquad \mathcal{S}^{lin}:=\mathcal{S}^{lin}_-\cap\mathcal{S}^{lin}_+;\]
where we note that this definition of $\mathcal{S}^{lin}$ matches that given at \eqref{slindef}. Moreover, we define the corresponding parts of the configuration space by setting, for $c>0$,
\[\mathcal{X}^{lin(c)}_-:=\left\{x\in \mathbb{R}^\Z\::\:S^x\in \mathcal{S}^{lin(c)}_-\right\},\qquad\mathcal{X}^{lin(c)}_+:=\left\{x\in \mathbb{R}^\Z\::\:S^x\in \mathcal{S}^{lin(c)}_+\right\},\]
and also
\[\mathcal{X}^{lin}_-:=\bigcup_{c>0}\mathcal{X}^{lin(c)}_-,\qquad\mathcal{X}^{lin}_+:=\bigcup_{c>0}\mathcal{X}^{lin(c)}_+,\qquad \mathcal{X}^{lin}:=\mathcal{X}^{lin}_-\cap\mathcal{X}^{lin}_+.\]

\begin{assu}\label{a1}
(a) For any $x\in\mathbb{R}^\Z$ satisfying
\[\lim_{n \to -\infty}S^x_n=\infty,\]
there does not exist a carrier for $x$.\\
(b) For any $x\in\mathcal{X}^{lin}_{-}$, there exists a carrier $w$ such that
\[\liminf_{n \to -\infty}M^w(S^x)_n=-\infty.\]
Moreover, if $u$ is any other carrier for $x$, then
\[\liminf_{n \to -\infty}M^u(S^x)_n >-\infty.\]
In particular, setting $W(x)=w$ yields a well-defined carrier function on $\mathcal{X}^{lin}_{-}$.\\
(c) For any $c>0$ and $x\in\mathcal{X}^{lin(c)}_{-}$, it holds that
\[\lim_{n \to -\infty}\frac{M^{W(x)}(S^x)_n}{n} =c.\]
(d) For any $c>0$ and $x\in\mathcal{X}^{lin}_{-}\cap\mathcal{X}^{lin(c)}_{+}$, it holds that
\[\lim_{n \to \infty}\frac{M^{W(x)}(S^x)_n}{n} =c.\]
\end{assu}

\begin{rem}
  In Assumption \ref{a1}, condition (a) gives a criteria for the degeneracy of a configuration (with respect to the existence of a carrier). Condition (b) ensures the existence of at least one carrier, and a means through which to select one of these uniquely. Condition (c) will allow us to check that the forward dynamics of the system are well-defined, and condition (d) similarly for the backward dynamics.
\end{rem}

\begin{prop}\label{linearsp} Suppose that we have a system of Pitman-type locally-defined dynamics on $\mathbb{R}^\mathbb{Z}$ such that Assumption \ref{a1}(a)-(c) holds. It is then the case that $W:\mathcal{X}^{lin}_{-}\rightarrow \R^{\Z}$ is a canonical carrier function. If it is moreover the case that  Assumption \ref{a1}(d) holds, then $W:\mathcal{X}^{lin}\rightarrow \R^{\Z}$ is a canonical carrier function.
\end{prop}
\begin{proof} Since $W:\mathcal{X}^{lin}_{-}\rightarrow \R^{\Z}$ is a carrier function by definition, we only need to check that properties (i) and (ii) of Definition \ref{carrierdef}(c) hold in each of the cases. Suppose $c>0$ and $x\in \mathcal{X}^{lin(c)}_{-}$. By Lemma \ref{lem:pitman}, we have that $S^{\mathcal{T}^{W(x)}(x)}=\mathcal{T}^{W(x)}_S(S^x)=\theta^m\circ T^{W(x)}(S^x)-\theta^m\circ T^{W(x)}(S^x)_0$, and so
\[\lim_{n \to - \infty}\frac{S^{\mathcal{T}^{W(x)}(x)}_n}{n}=\lim_{n\to-\infty}\frac{T^{W(x)}(S^x)_n}{n}=\lim_{n \to - \infty}\frac{2M^{W(x)}_n(S^x)-S^x_n}{n}=c,\]
where we have applied Assumption \ref{a1}(c) to obtain the final equality. In particular, this establishes that $\mathcal{T}^{W(x)}(x)\in \mathcal{X}^{lin(c)}_{-}$, and thus Definition \ref{carrierdef}(c)(i) holds for $W:\mathcal{X}^{lin}_{-}\rightarrow \R^{\Z}$. Moreover, if Assumption \ref{a1}(d) holds, then one similarly obtains that for $c>0$ and $x\in\mathcal{X}^{lin}_{-}\cap\mathcal{X}^{lin(c)}_{+}$, $\mathcal{T}^{W(x)}(x)\in \mathcal{X}^{lin(c)}_{+}$, and thus Definition \ref{carrierdef}(c)(i) holds for $W:\mathcal{X}^{lin}\rightarrow \R^{\Z}$ as well. Finally, we show Definition \ref{carrierdef}(c)(ii). Suppose $u$ is a carrier for $x\in \mathcal{X}^{lin}_{-}$ and $u\neq W(x)$. By Assumption \ref{a1}(b), we therefore have that $\liminf_{n \to -\infty}M^u(S^x)_n >-\infty$, and consequently
\begin{align*}
\liminf_{n \to -\infty}S^{\mathcal{T}^{u}(x)}&=\liminf_{n \to -\infty}\left(T^u(S^x)_{n+m}-T^u(S^x)_m\right)\\
&=\liminf_{n \to -\infty}  (2M^u(S^x)_n-S^x_n)-(2M^u(S^x)_m-S^x_m)\\
& =\infty,
\end{align*}
where we have again applied Lemma \ref{lem:pitman}. Hence $T^u(x)$ does not have a carrier by Assumption \ref{a1}(a).
\end{proof}

In view of the preceding result (and Theorem \ref{thm:uni-ex-f-gp}), it makes sense to write $\mathcal{T}=\mathcal{T}^{W(x)}$ whenever Assumption \ref{a1}(a)-(c) holds, and we will do this henceforth. We note that this operator on configurations also gives rise to an operator $\mathcal{T}_S=\pi\circ\mathcal{T}\circ\pi^{-1}$ on path encodings, which can be viewed as a Pitman-type transformation, as per the description at \eqref{pitu} with $u=W(x)$. Moreover, combining Theorem \ref{thm:uni-ex-f-gp} with Proposition \ref{linearsp} yields the following theorem, which provides a means of identifying solutions to the forward problem. Apart from the fact they include a path encoding description of the dynamics, a key difference of this and the subsequent theorem to our earlier results is that they make explicit the set of initial conditions for which we can solve the initial value problem.

\begin{thm}\label{t16} Suppose $m$ and $(F_n)_{n\in\Z}$ describe a system of Pitman-type locally-defined dynamics on $\mathbb{R}^\mathbb{Z}$ such that Assumption \ref{a1}(a)-(c) holds. It is then the case that, for each $x\in\mathcal{X}^{lin}_-$, the initial value problem
\[\begin{cases}
x^0=x, \\
\left(x^{t+1}_{n-m},u^t_n\right)=F_n\left(x^{t}_{n},u^t_{n-1}\right),\qquad \forall n \in \Z,\:\ t \in \Z_+,
\end{cases}\]
has a unique solution $((x^t,u^t))_{t\in\mathbb{Z}_+}$. This is given by setting $x^t:=\pi^{-1}(S^t)$ and $u^t:=W(x^t)$, where
\[S^t:=\mathcal{T}_S^t(S^x),\qquad \forall t\in\mathbb{Z}_+,\]
i.e.\ $(S^t)_{t\in\mathbb{Z}_+}$ is obtained from $S^x$ by applying iteratively the Pitman-type transformation associated with $W$ .
\end{thm}

We also have the following all-time version of the result, which follows from Theorem \ref{thm:uni-ex-f-gp-2} and Proposition \ref{linearsp}. For the statement, we recall the definition of self-reverse locally-defined dynamics from the former of the aforementioned results, and also introduce a path encoding reversal operator $R^\mathcal{S}:\mathcal{S}\rightarrow\mathcal{S}$ by setting
\[R^\mathcal{S}(S)_n:=-S_{m-n},\]
which we note satisfies $R^\mathcal{S}(S^x)=S^{R^\mathcal{X}(x)}$ for all $x\in\mathbb{R}^\mathbb{Z}$.

\begin{thm}\label{t17} Suppose $m$ and $(F_n)_{n\in\Z}$ describe a system of Pitman-type self-reverse locally-defined dynamics on $\mathbb{R}^\mathbb{Z}$ such that Assumption \ref{a1} holds. The following statements then hold.\\
(a) The operator $\mathcal{T}$ is a bijection from $\mathcal{X}^{lin}$ to itself, and
\[\mathcal{T}^{-1}=R^{\mathcal{X}}\circ\mathcal{T}\circ R^{\mathcal{X}}.\]
Or, in terms of path encodings, $\mathcal{T}_S$ is a bijection from $\mathcal{S}^0\cap\mathcal{S}^{lin}$ to itself, and
\begin{equation}\label{tsinv}
\mathcal{T}_S^{-1}=R^{\mathcal{S}}\circ\mathcal{T}_S\circ R^{\mathcal{S}}.
\end{equation}
(b) On $\mathcal{X}^{lin}$, it holds that
\[W=R^{\mathcal{X}}\circ W \circ R^{\mathcal{X}}\circ\mathcal{T}.\]
(c) For each $x\in\mathcal{X}^{lin}$, the initial value problem
\[\begin{cases}
x^0=x, \\
\left(x^{t+1}_{n-m},u^t_n\right)=F_n\left(x^{t}_{n},u^t_{n-1}\right),\qquad \forall n \in \Z,\:\ t \in \Z,
\end{cases}\]
has a unique solution $((x^t,u^t))_{t\in\mathbb{Z}}$. This is given by setting $x^t:=\pi^{-1}(S^t)$ and $u^t:=W(x^t)$, where
\[S^t:=\mathcal{T}_S^t(S^x),\qquad \forall t\in\mathbb{Z}.\]
\end{thm}

\begin{rem}
As we will show subsequently, each of the integrable systems \eqref{UDKDV}, \eqref{DKDV}, \eqref{UDTODA}, \eqref{DTODA} can be transformed so that the locally-defined dynamics are given by P-maps, and in particular satisfy a conservation property as in Definition \ref{P-mapdef}(a)(ii). It would be an interesting problem to explore which other integrable systems can be handled in a similar fashion, and which other Pitman-type transformations on path encodings arise in this way.
\end{rem}

\subsection{Key examples of Pitman-type locally-defined dynamics}\label{s45}

To complete this section, we consider four explicit examples of locally-defined dynamics on $\R^{\Z}$. Precisely, these are given as follows:
\begin{itemize}
\item $\mathbb{K}^{\vee}=(m,(K_n)_{n \in \Z})$, where $m=0$, and $K_n=K^{\vee}$ for all $n$, with
\[K^{\vee}(a,b)  :=\left(-\min\{a,2b\},b-\frac{a}{2}-\frac{\min\{a,2b\}}{2}\right),\qquad\forall(a,b)\in\mathbb{R}^2;\]
\item $\mathbb{K}^{\sum}=(m,(K_n)_{n \in \Z})$, where $m=0$ and $K_n=K^{\sum}$ for all $n$, with
\[K^{\sum}(a,b)  :=\left(2\log(e^{-\frac{a}{2}}+e^{-b}),b-\frac{a}{2}+\log(e^{-\frac{a}{2}}+e^{-b})\right),\qquad\forall(a,b)\in\mathbb{R}^2;\]
\item $\mathbb{K}^{\vee^*}=(m,(K_n)_{n \in \Z})$, where $m=1$, and $K_{n}=K^{\vee^*}$ for $n$ even, and $K_{n}=(K^{\vee^*})^{-1}$ for $n$ odd, with
\[K^{\vee^*}(a,b):= \left(-\min\{a,b\},b-\frac{a}{2}-\frac{\min\{a,b\}}{2}\right),\qquad\forall(a,b)\in\mathbb{R}^2;\]
\item $\mathbb{K}^{\sum^*}=(m,(K_n)_{n \in \Z})$, where $m=1$ and $K_{n}=K^{\sum^*}$ for $n$ even, and $K_{n}=(K^{\sum^*})^{-1}$ for $n$ odd, with
\[K^{\sum^*}(a,b) := \left(\log(e^{-a}+e^{-b}),b-\frac{a}{2}+\frac{\log(e^{-a}+e^{-b})}{2}\right),\qquad\forall(a,b)\in\mathbb{R}^2.\]
\end{itemize}
In later discussion, we will show that $\mathbb{K}^{\vee}$, $\mathbb{K}^{\sum}$, $\mathbb{K}^{\vee^*}$ and $\mathbb{K}^{\sum^*}$ correspond to \eqref{UDKDV}, \eqref{DKDV}, \eqref{UDTODA} and \eqref{DTODA}, respectively (see Section \ref{proofsec} in particular). By applying Proposition \ref{linearsp}, we will describe explicitly the canonical carrier functions for these systems on $\mathcal{X}^{lin}_-$ and $\mathcal{X}^{lin}$ (see Corollary \ref{c20}), and also describe the solutions to the corresponding initial value problems (see Corollary \ref{c25}). We start by checking that the maps $K^{\vee}$, $K^{\sum}$, $K^{\vee^*}$ and $K^{\sum^*}$ are P-maps, which ensures the inverses that appear in the above definition are well-defined, and that the systems are invariant under reversal.

\begin{lem}\label{lem:explicitmaps} (a) It holds that $K^{\vee}$, $K^{\sum}$, $K^{\vee^*}$ and $K^{\sum^*}$ are bijections, with inverses given by
\begin{eqnarray*}
(K^{\vee})^{-1}(a,b)&=&K^{\vee}(a,b),\\
(K^{\sum})^{-1}(a,b)&=&K^{\sum}(a,b),\\
(K^{\vee^*})^{-1} (a,b) &=&\left(-a-\min\{0,2b\}, -a-\min\{-b,0\}\right),\\
(K^{\sum^*})^{-1}(a,b)&=&\left(-a-2\log\left(\frac{\sqrt{1+4e^{2b}}-1}{2e^b}\right),-a-\log\left(\frac{\sqrt{1+4e^{2b}}-1}{2e^{2b}}\right)\right).
\end{eqnarray*}
(b) Each of the maps $K^{\vee}$, $K^{\sum}$, $K^{\vee^*}$ and $K^{\sum^*}$ are P-maps.\\
(c) Each of the systems of locally-defined dynamics $\mathbb{K}^{\vee}$, $\mathbb{K}^{\sum}$, $\mathbb{K}^{\vee^*}$ and $\mathbb{K}^{\sum^*}$ is self-reverse (as per the description in Theorem \ref{thm:uni-ex-f-gp-2}).
\end{lem}
\begin{proof} Part (a) can be checked by elementary computation, the details of which we omit.
Given part (a), to establish part (b) we simply need to check that property (ii) of Definition \ref{P-mapdef}(a) holds in each case, but this is obvious from the definitions of the maps in question. Finally, part (c) is immediate from the previous parts of the lemma.
\end{proof}

Our next step is to prove that the models $\mathbb{K}^{\vee}$, $\mathbb{K}^{\sum}$, $\mathbb{K}^{\vee^*}$ and $\mathbb{K}^{\sum^*}$  satisfy Assumption \ref{a1}, with the carrier functions mapping $\mathcal{X}^{lin}_-$ into $\mathbb{R}^\Z$ being given by, respectively:
\begin{eqnarray}
W_\pi^{\vee}(x)_n&=&M^{\vee}(S^x)_n-S^x_n,\label{wpiv}\\
W_\pi^{\sum}(x)_n&=&M^{\sum}(S^x)_n-S^x_n,\nonumber\\
W_\pi^{\vee^*}(x)_n&=&M^{\vee^*}(S^x)_n-S^x_n,\nonumber\\
W_\pi^{\sum^*}(x)_n&=&M^{\sum^*}(S^x)_n-S^x_n,\nonumber
\end{eqnarray}
where $M^{\vee}$, $M^{\vee^*}$, $M^{\sum}$ and $M^{\sum^*}$ are defined as in Table \ref{Mtable}; the descriptions of the latter functions explain our choice of superscripts in the model definitions. For future reference, we note that the above definitions correspond to the configuration versions of the carrier functions introduced in Section \ref{dissec}, which are defined on $\mathcal{S}^{lin}_-$. More specifically, we have $W_\pi^{\vee}=W^{\vee}\circ\pi$, $W_\pi^{\sum}=W^{\sum}\circ\pi$, $W_\pi^{\vee^*}=W^{\vee^*}\circ\pi$ and $W_\pi^{\sum^*}=W^{\sum^*}\circ\pi$, where $W^{\vee}$, $W^{\sum}$, $W^{\vee^*}$ and $W^{\sum^*}$ were defined at \eqref{wveedef}, \eqref{wsumdef}, \eqref{wveestardef} and \eqref{wsumstardef}, respectively.

\begin{thm}\label{thm:examples} The models $\mathbb{K}^{\vee}$, $\mathbb{K}^{\sum}$, $\mathbb{K}^{\vee^*}$ and $\mathbb{K}^{\sum^*}$  satisfy Assumption \ref{a1} with carrier functions being given by $W_\pi^{\vee}$, $W_\pi^{\sum}$, $W_\pi^{\vee^*}$ and $W_\pi^{\sum^*}$, respectively.
\end{thm}

Combining the previous result with Proposition \ref{linearsp}, yields the following corollary.

\begin{cor}\label{c20} For the models  $\mathbb{K}^{\vee}$, $\mathbb{K}^{\sum}$, $\mathbb{K}^{\vee^*}$ and $\mathbb{K}^{\sum^*}$, it holds that $W_\pi^{\vee}, W_\pi^{\sum}, W_\pi^{\vee^*}$ and $W_\pi^{\sum^*}$, respectively, are canonical carrier functions on $\mathcal{X}^{lin}_-$. Moreover, they are also canonical carrier functions on $\mathcal{X}^{lin}$.
\end{cor}

The claims of Theorem \ref{thm:examples} will be proved for each model separately across the subsequent four lemmas.

\begin{lem} If the locally-defined dynamics are given by $\mathbb{K}^{\vee}$, then:\\
(a) for $x\in\mathbb{R}^\Z$, $u$ is a carrier for $x$ if and only if
\begin{equation}\label{m-kdv}
M^u(S^x)_n=\max\left\{ M^u(S^x)_{n-1}, \frac{S^x_n+S^x_{n-1}}{2}\right\},\qquad \forall n\in\mathbb{Z};
\end{equation}
(b) Assumption \ref{a1} holds.
\end{lem}
\begin{proof} By Lemma \ref{lem:pitman}(a) and the definition of $K^{\vee}$, $u$ is a carrier for $x$ if and only if
\begin{align*}
M^u(S^x)_n & =S^x_n+M^u(S^x)_{n-1}-S^x_{n-1}-\frac{S^x_n-S^x_{n-1}}{2}-\frac{\min\left\{S^x_n-S^x_{n-1},2(M^u(S^x)_n-S^x_{n-1})\right\}}{2}\\
& = \frac{S^x_n-S^x_{n-1}}{2}+M^u(S^x)_{n-1} -\min \left\{ \frac{S^x_n-S^x_{n-1}}{2}, M^u(S^x)_n-S^x_{n-1}\right\} \\
& = \max\left\{ M^u(S^x)_{n-1}, \frac{S^x_n+S^x_{n-1}}{2}\right\},
\end{align*}
which establishes part (a) of the lemma. For part (b), we check the various requirements of Assumption \ref{a1} as follows:\\
(a) Suppose $x\in\mathbb{R}^{\Z}$ satisfies $\lim_{n \to -\infty}S^x_n=\infty$. From \eqref{m-kdv}, if $u$ is a carrier, then $M^u(S^x)_n \ge \sup_{m \le n}  \frac{S_m+S_{m-1}}{2}$ for all $n$. Hence $M^u(S^x)_n = \infty$ for all $n$, but this contradicts $M^u(S^x)_n=S^x_n+u_n<\infty$. Thus there does not exist a carrier for $x$.\\
(b) Let $x\in\mathcal{X}^{lin}_-$. Since $M^{\vee}(S^x)=M^{W_\pi^{\vee}(x)}(S^x)$ satisfies $M^{\vee}(S^x)_n=\max\{M^{\vee}(S^x)_{n-1},\frac{S^x_n+S^x_{n-1}}{2}\}$, $W_\pi^{\vee}(x)$ is a carrier for $x$. Also, that $\liminf_{n \to -\infty}M^{\vee}(S^x)=-\infty$ is obvious from the definition. Next, suppose that $u$ is a carrier for $x$, but $u \neq W_\pi^{\vee}(x)$. Since $M^u$ must satisfy \eqref{m-kdv}, $M^u(S^x)_n \ge M^{\vee}(S^x)_n$ for all $n$. In conjunction with the assumption that $u \neq W_\pi^{\vee}(x)$, this implies the existence of an $n_0$ such that $M^u(S^x)_{n_0} > M^\vee(S^x)_{n_0}$. Hence, by applying \eqref{m-kdv} again, we find that $M^u(S^x)_{n}=M^u(S^x)_{n_0}$ for all $n \le n_0$, and so $\liminf_{n \to -\infty}M^u(S^x)_n > -\infty$.\\
(c),(d) These properties are straightforward to check from the definition of $M^{\vee}$.
\end{proof}

\begin{lem} If the locally-defined dynamics are given by $\mathbb{K}^{\sum}$, then:\\
(a) for $x\in\mathbb{R}^\Z$, $u$ is a carrier for $x$ if and only if
\begin{equation}\label{m-kdv-d}
\exp\left( M^u(S^x)_n\right)= \exp\left( M^u(S^x)_{n-1}\right)+ \exp\left(\frac{S^x_n+S^x_{n-1}}{2}\right),\qquad \forall n\in\mathbb{Z};
\end{equation}
(b) Assumption \ref{a1} holds.
\end{lem}
\begin{proof}
By Lemma \ref{lem:pitman}(a) and the definition of $K^{\sum}$, $u$ is a carrier for $x$ if and only if
\begin{align*}
M^u(S^x)_n & =S^x_n+M^u(S^x)_{n-1}-S^x_{n-1}-\frac{S^x_n-S^x_{n-1}}{2}\\
&\qquad+\log\left(\exp\left(-\frac{S^x_n-S^x_{n-1}}{2}\right)+\exp\left(-M^u(S^x)_{n-1}+S^x_{n-1}\right)\right)\\
&=\frac{S^x_n-S^x_{n-1}}{2}+M^u(S^x)_{n-1}+\log\left(\exp\left(-\frac{S^x_n-S^x_{n-1}}{2}\right)+\exp\left(-M^u(S^x)_{n-1}+S^x_{n-1}\right)\right)\\
& = \log \left(\exp\left(M^u(S^x)_{n-1}\right)+\exp\left(\frac{S^x_n+S^x_{n-1}}{2}\right)\right),
\end{align*}
which establishes part (a) of the lemma. For part (b), we check the various requirements of Assumption \ref{a1} as follows:\\
(a) Suppose $x\in\mathbb{R}^{\Z}$ satisfies $\lim_{n \to -\infty}S^x_n=\infty$. From \eqref{m-kdv-d}, if $u$ is a carrier, then $\exp(M^u(S^x)_n) \ge \sum_{m \le n}  \exp(\frac{S^x_m+S^x_{m-1}}{2})$ for all $n$. Hence $M^u(S^x)_n = \infty$ for all $n$, but this contradicts $M^u(S^x)_n=S^x_n+u_n<\infty$. Thus there does not exist a carrier for $x$.\\
(b) Let $x\in\mathcal{X}^{lin}_-$. Since $M^{\sum}(S^x)=M^{W_\pi^{\sum}(x)}(S^x)$ satisfies $\exp( M^{\sum}(S^x)_n)= \exp( M^{\sum}(S^x)_{n-1})+ \exp(\frac{S^x_n+S^x_{n-1}}{2})$, $W_\pi^{\sum}(x)$ is a carrier for $x$. Also, that $\liminf_{n \to -\infty}M^{\sum}(S^x)=-\infty$ is obvious from the definition. Next, suppose that $u$ is a carrier for $x$, but $u \neq W_\pi^{\sum}(x)$. Since $M^u$ must satisfy \eqref{m-kdv-d}, $M^u(S^x)_n \ge M^{\sum}(S^x)_n$ for all $n$. In conjunction with the assumption that $u \neq W_\pi^{\sum}(x)$, this implies the existence of an $n_0$ such that $M^u(S^x)_{n_0} > M^{\sum}(S^x)_{n_0}$. Hence, by applying \eqref{m-kdv-d} again, we find that $\exp(M^u(S^x)_{n})-\exp(M^{\sum}(S^x)_{n})= \exp(M^u(S^x)_{n_0})-\exp(M^{\sum}(S^x)_{n_0}) >0$ for all $n \le n_0$, and so $\liminf_{n \to -\infty}\exp(M^u(S^x)_{n})>0$, which implies in turn that $\liminf_{n \to -\infty}M^u(S^x)_n > -\infty$.\\
(c),(d) These properties are straightforward to check from the definition of $M^{\sum}$.
\end{proof}

\begin{lem} If the locally-defined dynamics are given by $\mathbb{K}^{\vee^*}$, then:\\
(a) for $x\in\mathbb{R}^\Z$, $u$ is a carrier for $x$ if and only if
\begin{equation}\label{m-toda}
M^u(S^x)_{n}:=\left\{\begin{array}{ll}
 \max\{ M^u(S^x)_{n-2}, S^x_{n-1}\}, & n\mbox{ odd}, \\
  \frac{M^u(S^x)_{n-1}+M^u(S^x)_{n+1}}{2}, &  n\mbox{ even};
  \end{array}
\right.
\end{equation}
(b) Assumption \ref{a1} holds.
\end{lem}
\begin{proof}
By Lemma \ref{lem:pitman}(a), the definition of $K^{\vee^*}$, and the description of $(K^{\vee^*})^{-1}$ from Lemma \ref{lem:explicitmaps}, $u$ is a carrier for $x$ if and only if
\begin{align}
M^u(S^x)_{n} & =S^x_{n}+M^u(S^x)_{n-1}-S^x_{n-1}-\frac{S^x_{n}-S^x_{n-1}}{2}-\frac{\min\left\{S^x_{n}-S^x_{n-1},M^u(S^x)_{n-1}-S^x_{n-1}\right\}}{2}\nonumber\\
& = \frac{S^x_{n}-S^x_{n-1}}{2}+M^u(S^x)_{n-1}-\min\left\{\frac{S^x_{n}-S^x_{n-1}}{2}, \frac{M^u(S^x)_{n-1}-S^x_{n-1}}{2}\right\}\nonumber\\
&=\max\left\{M^u(S^x)_{n-1},\frac{M^u(S^x)_{n-1}+S^x_{n}}{2}\right\}\nonumber\\
&=\frac{M^u(S^x)_{n-1}}{2}+\frac{\max\left\{M^u(S^x)_{n-1},S^x_{n}\right\}}{2}\label{even}
\end{align}
for $n$ even, and
\begin{align}
M^u(S^x)_{n} & =S^x_{n}-(S^x_{n}-S^x_{n-1})-\min\left\{-M^u(S^x)_{n-1}+S^x_{n-1},0 \right\}\nonumber\\
& =S^x_{n-1} -\min\left\{-M^u(S^x)_{n-1}+S^x_{n-1},0 \right\}\nonumber\\
& = \max\left\{ M^u(S^x)_{n-1}, S^x_{n-1}\right\}\label{odd}
\end{align}
for $n$ odd. Suppose that both \eqref{even} and \eqref{odd} hold. Substituting \eqref{even} into \eqref{odd}, we obtain
\[M^u(S^x)_{n}=\max \left\{ \max\left\{ M^u(S^x)_{n-2}, \frac{M^u(S^x)_{n-2} +S^x_{n-1}}{2}\right\}, S^x_{n-1}\right\}=\max\left\{M^u(S^x)_{n-2} , S^x_{n-1}\right\}\]
for $n$ odd. Moreover, substituting this relation into \eqref{even}, we find
\[M^u(S^x)_{n}=\frac{M^u(S^x)_{n+1}+M^u(S^x)_{n-1}}{2}\]
for $n$ even. Hence \eqref{even} and \eqref{odd} together imply \eqref{m-toda}. The converse is clear by the explicit expression for $M^u(S^x)_{n}$ given by \eqref{m-toda}, and thus we have established part (a) of the lemma. For part (b), we check the various requirements of Assumption \ref{a1} as follows:\\
(a) Suppose $x\in\mathbb{R}^{\Z}$ satisfies $\lim_{n \to -\infty}S^x_n=\infty$. From \eqref{m-toda}, if $u$ is a carrier, then $M^u(S^x)_n \ge \max_{m \le \frac{n-1}{2}} S^x_{2m}$ for all $n$ odd. Hence $M^u(S^x)_n = \infty$ for all $n$, but this contradicts $M^u(S^x)_n=S^x_n+u_n<\infty$. Thus there does not exist a carrier for $x$.\\
(b) Let $x\in\mathcal{X}^{lin}_-$. Since $M^{\vee^*}(S^x)=M^{W_\pi^{\vee^*}(x)}(S^x)$ satisfies \eqref{m-toda}, $W_\pi^{\vee^*}(x)$ is a carrier for $x$. Also, that $\liminf_{n \to -\infty}M^{\vee^*}(S^x)=-\infty$ is obvious from the definition. Next, suppose that $u$ is a carrier for $x$, but $u \neq W_\pi^{\vee^*}(x)$. Since $M^u$ must satisfy \eqref{m-toda}, $M^u(S^x)_{n} \ge M^{\vee^*}(S^x)_{n}$ for all $n$. In conjunction with the assumption that $u \neq W_\pi^{\vee^*}(x)$, this implies the existence of an odd $n_0$ such that $M^u(S^x)_{n_0} > M^{\vee^*}(S^x)_{n_0}$. Hence, by applying \eqref{m-kdv-d} again, we find that $M^u(S^x)_{n}=M^u(S^x)_{n_0}$ for all $n \le n_0$, and so $\liminf_{n \to -\infty}M^u(S^x)_n > -\infty$.\\
(c),(d) These properties are straightforward to check from the definition of $M^{\vee^*}$.
\end{proof}

\begin{lem} If the locally-defined dynamics are given by $\mathbb{K}^{\sum^*}$, then:\\
(a) for $x\in\mathbb{R}^\Z$, $u$ is a carrier for $x$ if and only if
\begin{equation}\label{m-toda-d}
M^u(S^x)_{n}:=\left\{\begin{array}{ll}
 \exp(M^u(S^x)_{n-2})+\exp(S^x_{n-1}), & n\mbox{ odd}, \\
  \frac{M^u(S^x)_{n-1}+M^u(S^x)_{n+1}}{2}, &  n\mbox{ even};
  \end{array}
\right.
\end{equation}
(b) Assumption \ref{a1} holds.
\end{lem}
\begin{proof}
By Lemma \ref{lem:pitman}(a), the definition of $K^{\sum^*}$, and the description of $(K^{\sum^*})^{-1}$ from Lemma \ref{lem:explicitmaps}, $u$ is a carrier for $x$ if and only if
\begin{align}
M^u(S^x)_{n} & =S^x_{n}+M^u(S^x)_{n-1}-S^x_{n-1}-\frac{S^x_{n}-S^x_{n-1}}{2}\nonumber\\
&\qquad+\frac{\log\left(\exp\left(-S^x_{n}+S^x_{n-1}\right)+\exp\left(-M^u(S^x)_{n-1}+S^x_{n-1}\right)\right)}{2}\nonumber\\
& = \frac{S^x_{n}-S^x_{n-1}}{2}+M^u(S^x)_{n-1}\nonumber\\
&\qquad+\frac{\log\left(\exp\left(-S^x_{n}+S^x_{n-1}\right)+\exp\left(-M^u(S^x)_{n-1}+S^x_{n-1}\right)\right)}{2}\nonumber\\
& = \frac{1}{2}\log\left( \exp\left(2M^u(S^x)_{n-1}\right)+\exp\left(M^u(S^x)_{n-1}+S^x_{n}\right)\right)\nonumber\\
&=\frac{M^u(S^x)_{n-1}}{2}+\frac{\log\left( \exp\left(M^u(S^x)_{n-1}\right)+\exp\left(S^x_{n}\right)\right)}{2}\label{bbb3}
\end{align}
for $n$ even, and
\begin{align}
M^u(S^x)_{n}& =S^x_{n}-\left(S^x_{n}-S^x_{n-1}\right)-\log\left(\frac{\sqrt{1+4\exp\left(2M^u(S^x)_{n-1}-2S^x_{n-1}\right)}-1}{2\exp\left(2M^u(S^x)_{n-1}-2S^x_{n-1}\right)}\right)\nonumber\\
& =S^x_{n-1}-\log\left(\frac{\sqrt{1+4\exp\left(2M^u(S^x)_{n-1}-2S^x_{n-1}\right)}-1}{2\exp\left(2M^u(S^x)_{n-1}-2S^x_{n-1}\right)}\right)\label{bbb4}
\end{align}
for $n$ odd. Note that \eqref{bbb3} can be rewritten as
\begin{equation}\label{cal1}
\exp\left(2M^u(S^x)_{n}\right)=\exp\left(M^u(S^x)_{n-1}\right)\left(\exp\left(M^u(S^x)_{n-1}\right)+\exp\left(S_{n}\right)\right)
\end{equation}
for $n$ even. Suppose that \eqref{bbb4} and \eqref{cal1} both hold. From \eqref{cal1},
\begin{align*}
\sqrt{1+4\exp\left(2M^u(S^x)_{n}-2S^x_{n}\right)} & = \sqrt{1+ 4\exp\left(2M^u(S^x)_{n-1}-2S^x_{n}\right)+4\exp\left(M^u(S^x)_{n-1}-S^x_{n}\right)}\\
&=2\exp\left(M^u(S^x)_{n-1}-S^x_{n}\right)+1
\end{align*}
for $n$ even. Hence
\[\log\left(\frac{\sqrt{1+4\exp\left(2M^u(S^x)_{n}-2S^x_{n}\right)}-1}{2\exp\left(2M^u(S^x)_{n}-2S^x_{n}\right)}\right)=M^u(S^x)_{n-1}+S^x_{n}-2M^u(S^x)_n\]
for $n$ even. In conjunction with \eqref{bbb4}, this implies that
\begin{equation}\label{ccc1}
M^u(S^x)_{n}+M^u(S^x)_{n-2}=2M^u(S^x)_{n-1}
\end{equation}
for $n$ odd. Moreover, plugging this back into \eqref{cal1} yields
\begin{equation}\label{ccc2}
\exp\left(M^u(S^x)_{n}\right)=\exp\left(M^u(S^x)_{n-2}\right)+\exp\left(S^x_{n-1}\right)
\end{equation}
for $n$ odd, and clearly \eqref{ccc1} and \eqref{ccc2} together give \eqref{m-toda-d}. The converse is clear by the explicit expression for $M^u(S^x)_{n}$ given by \eqref{m-toda-d}, and thus we have established part (a) of the lemma. For part (b), we check the various requirements of Assumption \ref{a1} as follows:\\
(a) Suppose $x\in\mathbb{R}^{\Z}$ satisfies $\lim_{n \to -\infty}S^x_n=\infty$. From \eqref{m-toda-d}, if $u$ is a carrier, then $\exp(M^u(S^x)_n) \ge \sum_{m \le \frac{n-1}{2}}\exp(S^x_{2m})$ for all $n$ odd. Hence $M^u(S^x)_n = \infty$ for all $n$, but this contradicts $M^u(S^x)_n=S^x_n+u_n<\infty$. Thus there does not exist a carrier for $x$.\\
(b) Let $x\in\mathcal{X}^{lin}_-$. Since $M^{\sum^*}(S^x)=M^{W_\pi^{\sum^*}(x)}(S^x)$ satisfies \eqref{m-toda-d}, $W_\pi^{\sum^*}(x)$ is a carrier for $x$. Also, that $\liminf_{n \to -\infty}M^{\sum^*}(S^x)=-\infty$ is obvious from the definition. Next, suppose that $u$ is a carrier for $x$, but $u \neq W_\pi^{\sum^*}(x)$. Since $M^u$ must satisfy \eqref{m-toda-d}, $M^u(S^x)_{n} \ge M^{\sum^*}(S^x)_{n}$ for all $n$. In conjunction with the assumption that $u \neq W_\pi^{\sum^*}(x)$, this implies the existence of an odd $n_0$ such that $M^u(S^x)_{n_0} > M^{\sum^*}(S^x)_{n_0}$. Hence, by applying \eqref{m-toda-d} again, we find that $\exp(M^u(S^x)_{n})-\exp(M^{\sum^*}(S^x)_{n})= \exp(M^u(S^x)_{n_0})-\exp(M^{\sum^*}(S^x)_{n_0})>0$ for all $n \le n_0$, and so it must hold that $\liminf_{n \to -\infty}M^u(S^x)_n > -\infty$.\\
(c),(d) These properties are straightforward to check from the definition of $M^{\sum^*}$.
\end{proof}

To finish this section, we present a corollary that is the culmination of our preceding results, specifically Theorems \ref{t16}, \ref{t17} and \ref{thm:examples}

\begin{cor}\label{c25} The conclusions of Theorems \ref{t16} and Theorems \ref{t17} apply to each of the models $\mathbb{K}^{\vee}$, $\mathbb{K}^{\sum}$, $\mathbb{K}^{\vee^*}$ and $\mathbb{K}^{\sum^*}$, with dynamics on the path space $\mathcal{S}^0\cap\mathcal{S}^{lin}_-$, i.e.\ the map $S\mapsto \mathcal{T}_S(S)$, being given by the operators
\begin{align*}
T^{\vee}_0(S)&:=T^{\vee}(S)-T^{\vee}(S)_0,\\
T^{\sum}_0(S)&:=T^{\sum}(S)-T^{\sum}(S)_0,\\
\mathcal{T}^{\vee^*}_0(S)&:=\mathcal{T}^{\vee^*}(S)-\mathcal{T}^{\vee^*}(S)_0,\\
\mathcal{T}^{\sum^*}_0(S)&:=\mathcal{T}^{\sum^*}(S)-\mathcal{T}^{\sum^*}(S)_0,
\end{align*}
respectively, where $T^{\vee}$, $T^{\sum}$, $\mathcal{T}^{\vee^*}$ and $\mathcal{T}^{\sum^*}$ are defined in Table \ref{Mtable}. Moreover, each of the above operators is a bijection on $\mathcal{S}^0\cap\mathcal{S}^{lin}$, with inverse given by \eqref{tsinv}.
\end{cor}

\begin{rem}\label{udreme}
Foreshadowing the discussion of ultra-discretization that will appear in Subsection \ref{udsec}, for P-maps $(F^{(\varepsilon)})_{\varepsilon>0}$ and $F$, we say that $F$ is the \emph{ultra-discretization} of $(F^{(\varepsilon)})_{\varepsilon>0}$ if it is the case that
\begin{equation}\label{fudis}
\lim_{\epsilon \downarrow 0}  \epsilon F^{(\varepsilon)}\left(\frac{x}{\epsilon} , \frac{u}{\epsilon}\right)=F(x,u),\qquad \forall x,u\in\mathbb{R}.
\end{equation}
Or, if it is additionally the case that $F^{(\varepsilon)}=G$ for all $\varepsilon$, then we simply say $F$ is the ultra-discretization of $G$. In particular, similarly to the proofs of Theorems \ref{rrr} and \ref{t319} below, one can check that $K^{\vee}$ is the ultra-discretization of $K^{\sum}$, $K^{\vee^*}$ is the ultra-discretization of $K^{\sum^*}$, and $(K^{\vee^*})^{-1}$ is the ultra-discretization of $(K^{\sum^*})^{-1}$.

The P-map ultra-discretization property extends naturally to path encodings. Indeed, suppose that for each $n$, the P-map $F_n$ is the ultra-discretization of a sequence of P-maps $(F_n^{(\varepsilon)})_{\varepsilon>0}$, with $F_n$ being continuous and the convergence at \eqref{fudis} holding uniformly on compacts. Moreover, suppose that for each $\varepsilon>0$, the locally-defined dynamics $(m,(F_n^{(\varepsilon)})_{n\in\mathbb{Z}})$ admit a canonical carrier $U^{(\varepsilon)}:\mathcal{X}^U\rightarrow \mathbb{R}^\mathbb{Z}$, where $m$ and $\mathcal{X}^U\subseteq \mathbb{R}^\mathbb{Z}$ do not depend on $\varepsilon$, and also
\[U(x)_n:=\lim_{\varepsilon\downarrow0}\varepsilon U^{(\varepsilon)}\left(\varepsilon^{-1}x\right)_n,\qquad\forall x\in\mathcal{X}^U,\:n\in\mathbb{Z}.\]
It is then the case that $U:\mathcal{X}^U\rightarrow \mathbb{R}^\mathbb{Z}$ is a carrier for $(m,(F_n)_{n\in\mathbb{Z}})$, and the associated dynamics satisfy $\mathcal{T}(\mathcal{X}^U)\subseteq \mathcal{X}^U$ and
\[\mathcal{T}(x)_n:=\lim_{\varepsilon\downarrow0}\varepsilon \mathcal{T}^{(\varepsilon)}\left(\varepsilon^{-1}x\right)_n,\qquad\forall x\in\mathcal{X}^U,\:n\in\mathbb{Z}.\]
In particular, applying this in conjunction with the conclusion of the previous paragraph, with $\mathcal{X}^U=\mathcal{X}^{lin}_-$, we find that $T_0^{\vee}$ is the ultra-discretization of $T_0^{\sum}$, and $\mathcal{T}_0^{\vee^*}$ is the ultra-discretization of $\mathcal{T}_0^{\sum^*}$ (cf.\ Remarks \ref{r37} and \ref{r322} below).
\end{rem}

\section{Proof of main results}\label{proofsec}

In this section, we explain how the models of Section \ref{dissec} fit into the framework of Section \ref{general}, and in particular apply the conclusions of the latter section to establish our main results for the \eqref{UDKDV}, \eqref{DKDV}, \eqref{UDTODA} and \eqref{DTODA} systems, namely Theorems \ref{udkdvthm}, \ref{dkdvthm}, \ref{udtodathm} and \ref{dtodathm}. Although the arguments for each will be similar, for the sake of clarity we consider each of the models of interest separately. Moreover, we finish the section with a remark that relates the conservation of $a-2b$ that appears in the definition of a P-map to the conserved quantities of the original maps, as discussed in Subsection \ref{conssec}.

\subsection{Ultra-discrete KdV equation}\label{515151} To reexpress the \eqref{UDKDV} system of Subsection \ref{udsec1} in the notation of Section \ref{general}, for fixed $L\in\mathbb{R}$, let $\mathbb{F}^{(L)}_{udK}=(m,(F_n)_{n \in \Z})$, where $m=0$, and $F_n=F_{udK}^{(L)}$ for all $n$, and we recall the map $F_{udK}^{(L)}$ from Section \ref{lmsec}. To obtain locally-defined dynamics, as per Definition \ref{ldd}, we take the configuration space to be $\mathbb{R}^\mathbb{Z}$, with typical elements being written $(\eta_n)_{n\in\mathbb{Z}}$, and the carrier space to also be $\mathbb{R}^\mathbb{Z}$, with typical elements being denoted $(U_n)_{n\in\mathbb{Z}}$. We also introduce the following maps:
\vspace{-20pt}
\begin{multicols}{2}

\begin{eqnarray*}
\mathcal{A}^{udK}:\mathbb{R}^\mathbb{Z}&\rightarrow&\mathbb{R}^\mathbb{Z}\\
\left(\eta_n\right)_{n\in\mathbb{Z}}&\mapsto&\left(\mathcal{A}^{udK}_n(\eta_n)\right)_{n\in\mathbb{Z}},
\end{eqnarray*}

\begin{eqnarray*}
\mathcal{B}^{udK}:\mathbb{R}^\mathbb{Z}&\rightarrow&\mathbb{R}^\mathbb{Z}\\
\left(U_n\right)_{n\in\mathbb{Z}}&\mapsto&\left(\mathcal{B}^{udK}_n(U_n)\right)_{n\in\mathbb{Z}},
\end{eqnarray*}
\end{multicols}
\noindent
where $\mathcal{A}^{udK}_n(\eta_n):=L-2\eta_n$ and $\mathcal{B}^{udK}_n(U_n)=U_n-\frac{L}{2}$. (For simplicity, we suppress the dependence on $L$ from notation for $\mathcal{A}^{udK}$ and $\mathcal{B}^{udK}$.) As the next lemma demonstrates, these maps provide the link between $\mathbb{F}^{(L)}_{udK}$ and $\mathbb{K}^{\vee}$, as defined at the start of Subsection \ref{s45}. Since the proof involves only straightforward applications of the definitions of the relevant objects, we omit it. We recall that $\mathcal{C}_{udK}^{(L)}$ was defined at \eqref{cudkdef}, $S_{udK}^{(L)}$ at \eqref{1pe}, and $\pi$ in Definition \ref{pidef}.

\begin{lem}\label{ll1} For any $L\in\mathbb{R}$, the following statements hold.\\
(a) On $\mathbb{R}^2$,
\[\left(\mathcal{A}^{udK}_{n-m}\times\mathcal{B}^{udK}_n\right)\circ F_{udK}^{(L)}=K^{\vee}\circ\left(\mathcal{A}^{udK}_{n}\times\mathcal{B}^{udK}_{n-1}\right).\]
(b) Both $\mathcal{A}^{udK}:\mathbb{R}^\Z\rightarrow\mathbb{R}^\Z$ and $\mathcal{B}^{udK}:\mathbb{R}^\Z\rightarrow\mathbb{R}^\Z$ are bijections.\\
(c) The map $\pi\circ\mathcal{A}^{udK}:\mathbb{R}^\Z\rightarrow \mathcal{S}^{0}$ is a bijection, and satisfies $\pi\circ\mathcal{A}^{udK}=S_{udK}^{(L)}$.\\
(d) It is the case that $\mathcal{A}^{udK}(\mathcal{C}_{udK}^{(L)})=\mathcal{X}^{lin}$.
\end{lem}

\begin{proof}[Proof of Theorem \ref{udkdvthm}] By applying Theorem \ref{thm:uni-ex-f-gp-3}(b) (with $\tilde{\mathcal{X}}=\tilde{\mathcal{U}}=\mathbb{R}^\mathbb{Z}$ and $\mathcal{X}^U=\mathcal{X}^{lin}$), Corollary \ref{c25} and Lemma \ref{ll1}, we obtain that if $\eta\in\mathcal{C}_{udK}^{(L)}$, then there is a unique solution $(\eta_n^t,U_n^t)_{n,t\in\mathbb{Z}}$ to \eqref{UDKDV} that satisfies the initial condition $\eta^0=\eta$, which has configuration component given by
\begin{eqnarray*}
\eta^t&=&\left(\mathcal{A}^{udK}\right)^{-1}\circ\pi^{-1}\circ\left({T}_0^{\vee}\right)^t\circ\pi\circ\mathcal{A}^{udK}(\eta)\\
&=&\left(S_{udK}^{(L)}\right)^{-1}\circ\left({T}_0^{\vee}\right)^t\circ S_{udK}^{(L)} (\eta).
\end{eqnarray*}
Now, it is elementary to see that $({T}^{\vee})^t(S)=({T}_0^{\vee})^t(S)+({T}^{\vee})^t(S)_0$ for any $S\in\mathcal{S}^{lin}$, and so recalling the extension of the definition of $(S_{udK}^{(L)})^{-1}$ from $\mathcal{S}^0$ to $\mathcal{S}$, as given in Subsection \ref{udsec1}, we see that
\[\eta^t=\left(S_{udK}^{(L)}\right)^{-1}\circ\left({T}^{\vee}\right)^t\circ S_{udK}^{(L)} (\eta).\]
Similar manipulations yield the corresponding carrier is given by
\[U^t=\left(\mathcal{B}^{udK}\right)^{-1}\circ W^{\vee}_{\pi}\circ\mathcal{A}^{udK}(\eta^t)=\left(W_{udK}^{(L)}\right)^{-1}\circ W^{\vee}\circ\left({T}^{\vee}\right)^t\circ S_{udK}^{(L)} (\eta),\]
where $W^{\vee}_{\pi}$ was defined at \eqref{wpiv}, and $W^{\vee}$ at \eqref{wveedef}. Moreover, by construction, it holds that $\mathcal{A}^{udK}(\eta^t)\in\mathcal{X}^{lin}$ for all $t$, and so $\eta^t\in \mathcal{C}_{udK}^{(L)}$ for all $t$, as desired.
\end{proof}

\subsection{Discrete KdV equation} We rewrite the \eqref{DKDV} system of Subsection \ref{dsec1} by setting, for $\delta\in(0,\infty)$, $\mathbb{F}^{(\delta)}_{dK}=(m,(F_n)_{n \in \Z})$, where $m=0$, and $F_n=F_{dK}^{(\delta)}$ for all $n$, and we recall the map $F_{dK}^{(\delta)}$ from Section \ref{lmsec}. We take the configuration space to be $(0,\infty)^\mathbb{Z}$, with typical elements being written $(\omega_n)_{n\in\mathbb{Z}}$, and the carrier space to also be $(0,\infty)^\mathbb{Z}$, with typical elements being denoted $(U_n)_{n\in\mathbb{Z}}$. We also introduce the following maps:
\vspace{-20pt}
\begin{multicols}{2}

\begin{eqnarray*}
\mathcal{A}^{dK}:(0,\infty)^\mathbb{Z}&\rightarrow&\mathbb{R}^\mathbb{Z}\\
\left(\eta_n\right)_{n\in\mathbb{Z}}&\mapsto&\left(\mathcal{A}^{dK}_n(\eta_n)\right)_{n\in\mathbb{Z}},
\end{eqnarray*}

\begin{eqnarray*}
\mathcal{B}^{dK}:(0,\infty)^\mathbb{Z}&\rightarrow&\mathbb{R}^\mathbb{Z}\\
\left(U_n\right)_{n\in\mathbb{Z}}&\mapsto&\left(\mathcal{B}^{dK}_n(U_n)\right)_{n\in\mathbb{Z}},
\end{eqnarray*}
\end{multicols}
\noindent
where $\mathcal{A}^{dK}_n(\omega_n):=-\log\delta-2\log \omega_n$ and $\mathcal{B}^{dK}_n(U_n)=\log U_n+\frac{\log\delta}{2}$. (Again, for simplicity, we have suppressed the dependence on the parameter, in this case $\delta$, from the notation.) Corresponding to Lemma \ref{ll1}, we have the following result, which gives the link between $\mathbb{F}^{(\delta)}_{dK}$ and $\mathbb{K}^{\sum}$, as defined at the start of Subsection \ref{s45}. Again, since the proof involves only straightforward applications of the definitions of the relevant objects, we choose not to include it here.

\begin{lem}\label{ll2} For any $\delta\in(0,\infty)$, the following statements hold.\\
(a) On $(0,\infty)^2$,
\[\left(\mathcal{A}^{dK}_{n-m}\times\mathcal{B}^{dK}_n\right)\circ F_{dK}^{(\delta)}=K^{\sum}\circ\left(\mathcal{A}^{dK}_{n}\times\mathcal{B}^{dK}_{n-1}\right).\]
(b) Both $\mathcal{A}^{dK}:(0,\infty)^\mathbb{Z}\rightarrow\mathbb{R}^\Z$ and $\mathcal{B}^{dK}:(0,\infty)^\mathbb{Z}\rightarrow\mathbb{R}^\Z$ are bijections.\\
(c) The map $\pi\circ\mathcal{A}^{dK}:(0,\infty)^\mathbb{Z}\rightarrow \mathcal{S}^{0}$ is a bijection, and satisfies $\pi\circ\mathcal{A}^{dK}=S_{dK}^{(\delta)}$.\\
(d) It is the case that $\mathcal{A}^{dK}(\mathcal{C}_{dK}^{(\delta)})=\mathcal{X}^{lin}$.
\end{lem}

\begin{proof}[Proof of Theorem \ref{dkdvthm}] Given Lemma \ref{ll2}, the proof of Theorem \ref{dkdvthm} is identical to that of Theorem \ref{udkdvthm}, namely an application of Theorem \ref{thm:uni-ex-f-gp-3}(b) and Corollary \ref{c25}, so we omit it.
\end{proof}

\subsection{Ultra-discrete Toda lattice} Since \eqref{UDTODA} has three variables, to put it in the framework of Section \ref{general} we need to rewrite the equation. To this end, we first define a map $F_{udT^*}:\R^2 \to \R^2$ by setting
\[F_{udT^*}(a,b)=\left(\min\{a,b\}, b-\frac{a}{2}-\frac{\min\{a,b\}}{2}\right).\]
It is easy to check that $F_{udT^*}$ is a bijection with inverse function given by
\[F_{udT^*}^{-1}(a,b)=\left(a+\max\{0,-2b\},a+\max\{0,b\}\right).\]
The connection with the map $F_{udT}$ of Section \ref{lmsec} is provided by the following lemma, the proof of which is a straightforward calculation.

\begin{lem}\label{ls1} For any $a,b,c \in \R$,
\[F_{udT}(a,b,c)=\left(F_{udT^*}^{(1)}(b,c), F_{udT^*}^{-1}\left(a,F_{udT^*}^{(2)}(b,c)\right)\right).\]
In particular, $F_{udT}(a,b,c)=(d,e,f)$ if and only if $F_{udT^*}^{(1)}(b,c)=d$ and $F_{udT^*}^{-1}(a,g)=(e,f)$ where $g:=F_{udT^*}^{(2)}(b,c)$.
\end{lem}

\begin{rem}\label{rs1} We can understand Lemma \ref{ls1} graphically as a decomposition of the lattice ultra-discrete Toda lattice structure, as discussed in Section \ref{lmsec}. In particular, we replace the single map $F_{udT}$ with three inputs and three outputs by two maps $F_{udT^*}$ and $F_{udT^*}^{-1}$, each with two inputs and two outputs:
\[\xymatrix@C-15pt@R-15pt{\boxed{F_{udT^*}\vphantom{F_{udT^*}^{-1}}} & \min\{b,c\} &\boxed{F_{udT^*}^{-1}}& a+b-\min\{b,c\}&\\
            c \ar[rr] && c-\frac{b}{2}-\frac{\min\{b,c\}}{2}\ar[rr]&& a+c-\min\{b,c\}.\\
             & b \ar[uu]&&a\ar[uu]&}\]
\end{rem}

Lemma \ref{ls1} motivates the introduction of the locally-defined dynamics $\mathbb{F}_{udT}=(m,(F_n)_{n \in \Z})$, where $m=1$, and $F_n=F_{udT^*}$ for $n$ even, and $F_{n}=(F_{udT^*})^{-1}$ for $n$ odd. For this system, we take the configuration space to be $\mathbb{R}^\mathbb{Z}$, with typical elements being written $(\eta_n)_{n\in\mathbb{Z}}$, and the carrier space to also be $\mathbb{R}^\mathbb{Z}$, with typical elements being denoted $(\tilde{U}_n)_{n\in\mathbb{Z}}$. In particular, we have the following ready consequence of the preceding result.

\begin{cor}\label{c57} If $(Q_n^t,E_n^t,U_n^t)_{n,t\in\mathbb{Z}}$ is a solution to \eqref{UDTODA} and we set
\begin{equation}\label{ad1}
\eta_{2n-1}^t=Q_n^t,\qquad \eta_{2n}^t=E_n^t,\qquad \tilde{U}_{2n-1}^t=U_n^t,\qquad \tilde{U}_{2n}^t=F_{udT^*}^{(2)}\left(E_n^t,U_n^t\right),
\end{equation}
then $(\eta_n^t,\tilde{U}_n^t)_{n,t\in\mathbb{Z}}$ is a solution to
\begin{equation}\label{ad2}
\left(\eta_{n-1}^{t+1},\tilde{U}_n^t\right)=F_n\left(\eta_n^t,\tilde{U}_{n-1}^t\right),\qquad \forall n,t\in\mathbb{Z},
\end{equation}
where the maps $(F_n)_{n\in\mathbb{Z}}$ are given by $\mathbb{F}_{udT}$. Moreover, if $(\eta_n^t,\tilde{U}_n^t)_{n,t\in\mathbb{Z}}$ is a solution to \eqref{ad2} and the first three equations of \eqref{ad1} hold, then $(Q_n^t,E_n^t,U_n^t)_{n,t\in\mathbb{Z}}$ is a solution to \eqref{UDTODA}.
\end{cor}

Since the map $\pi_{udT}$ from $(Q,E)\in(\mathbb{R}^2)^\Z$ to $\eta\in\mathbb{R}^\Z$ given by setting $\eta_{2n-1}=Q_n$ and $\eta_{2n}=E_n$ is clearly a bijection, to tackle the initial value problem of Theorem \ref{udtodathm}, it thus suffices to study the corresponding problem for the system $\mathbb{F}_{udT}$. To align with Section \ref{general}, we also introduce the following maps:
\vspace{-20pt}
\begin{multicols}{2}

\begin{eqnarray*}
\mathcal{A}^{udT}:\mathbb{R}^\mathbb{Z}&\rightarrow&\mathbb{R}^\mathbb{Z}\\
\left(\eta_n\right)_{n\in\mathbb{Z}}&\mapsto&\left(\mathcal{A}^{udT}_n(\eta_n)\right)_{n\in\mathbb{Z}},
\end{eqnarray*}

\begin{eqnarray*}
\mathcal{B}^{udT}:\mathbb{R}^\mathbb{Z}&\rightarrow&\mathbb{R}^\mathbb{Z}\\
\left(\tilde{U}_n\right)_{n\in\mathbb{Z}}&\mapsto&\left(\mathcal{B}^{udT}_n(\tilde{U}_n)\right)_{n\in\mathbb{Z}},
\end{eqnarray*}
\end{multicols}
\noindent
where $\mathcal{A}^{udT}_n(\eta_n):=(-1)^n\eta_n$ and $\mathcal{B}^{udT}_n(\tilde{U}_n)=\tilde{U}_n$. As the next lemma demonstrates, these maps provide the link between $\mathbb{F}_{udT}$ and $\mathbb{K}^{\vee^*}$, as defined at the start of Subsection \ref{s45}. The easy proof is omitted.

\begin{lem}\label{ll3} The following statements hold.\\
(a) On $\mathbb{R}^2$,
\[\left(\mathcal{A}^{udT}_{2n-1}\times\mathcal{B}^{udT}_{2n}\right)\circ F_{udT^*}=K^{\vee^*}\circ\left(\mathcal{A}^{udT}_{2n}\times\mathcal{B}^{udT}_{2n-1}\right),\]
and also
\[\left(\mathcal{A}^{udT}_{2n}\times\mathcal{B}^{udT}_{2n+1}\right)\circ F_{udT^*}^{-1}=\left(K^{\vee^*}\right)^{-1}\circ\left(\mathcal{A}^{udT}_{2n+1}\times\mathcal{B}^{udT}_{2n}\right).\]
(b) Both $\mathcal{A}^{udT}:\mathbb{R}^\Z\rightarrow\mathbb{R}^\Z$ and $\mathcal{B}^{udT}:\mathbb{R}^\Z\rightarrow\mathbb{R}^\Z$ are bijections.\\
(c) The map $\pi\circ\mathcal{A}^{udT}:\mathbb{R}^\Z\rightarrow \mathcal{S}^{0}$ is a bijection, and moreover $\pi\circ\mathcal{A}^{udT}\circ\pi_{udT}=S_{udT}$.\\
(d) It is the case that $\mathcal{A}^{udT}\circ\pi_{udT}(\mathcal{C}_{udT})=\mathcal{X}^{lin}$.
\end{lem}

\begin{proof}[Proof of Theorem \ref{udtodathm}] Let $(Q,E)\in\mathcal{C}_{udT}$. Following the proof of Theorem \ref{udkdvthm}, by applying Theorem \ref{thm:uni-ex-f-gp-3}(b),  Corollary \ref{c25} and Lemma \ref{ll3} (in place of Lemma \ref{ll1}), we find that there exists a unique solution $(\eta_n^t,\tilde{U}_n^t)_{n,t\in\mathbb{Z}}$ to \eqref{ad2} with $\eta^0=\eta:=\pi_{udT}(Q,E)$, which is given by
\[\eta^t=\left(\mathcal{A}^{udT}\right)^{-1}\circ \pi^{-1}\circ \left(T^{\vee^*}_0\right)^t\circ \pi\circ\mathcal{A}^{udT}(\eta),\]
and
\[\tilde{U}^t=\left(\mathcal{B}^{udT}\right)^{-1}\circ W^{\vee^*}_\pi\circ \mathcal{A}^{udT}(\eta^t).\]
Hence, from Corollary \ref{c57}, it follows that the unique solution $(Q^t_n,E^t_n,U^t_n)_{n,t\in\mathbb{Z}}$ to \eqref{UDTODA} with initial condition $(Q^0,E^0)=(Q,E)$ is given by
\begin{align*}
\left(Q^t,E^t\right)&=\pi_{udT}^{-1}\circ\left(\mathcal{A}^{udT}\right)^{-1}\circ \pi^{-1}\circ \left(\mathcal{T}^{\vee^*}_0\right)^t\circ \pi\circ\mathcal{A}^{udT}\circ\pi_{udT}(Q,E)\\
&=\left(S_{udT}\right)^{-1}\circ\left(\mathcal{T}^{\vee^*}\right)^t\circ S_{udT} (Q,E),
\end{align*}
and
\[{U}^t_n= W^{\vee^*}\circ \left(\mathcal{T}^{\vee^*}\right)^t\circ S_{udT} (Q,E)_{2n-1},\]
where, similarly to the proof of Theorem \ref{udkdvthm}, we note that $(\mathcal{T}^{\vee^*})^t(S)=(\mathcal{T}_0^{\vee^*})^t(S)+(\mathcal{T}^{\vee^*})^t(S)_0$ for $S\in\mathcal{S}^{lin}$, and consider the extension of $(S_{udT})^{-1}$ to $\mathcal{S}$. Moreover, by construction $\eta^t\in(\mathcal{A}^{udT})^{-1}(\mathcal{X}^{lin})$ for all $t$, and so $(Q^t,E^t)=\pi_{udT}^{-1}(\eta^t)\in \mathcal{C}_{udT}$ for all $t$.
\end{proof}

\subsection{Discrete Toda lattice} As for the ultra-discrete Toda lattice, to place the discrete Toda lattice into the framework of Section \ref{general}, we first need to decompose the map $F_{dT}$ introduced in Section \ref{lmsec}. Specifically, we define a map $F_{dT^*}:(0,\infty)^2 \to (0,\infty)^2$ by setting
\[F_{dT^*}(a,b)=\left(a+b, \frac{b}{\sqrt{a^2+ab}}\right).\]
This is a bijection with inverse function given by
\[F_{dT^*}^{-1}(a,b)=\left(a\left(\frac{\sqrt{b^2+4}-b}{2}\right)^2 , a\left(\frac{\sqrt{b^4+4b^2}-b^2}{2}\right)\right).\]
The connection with the map $F_{dT}$ is provided by the following lemma, the proof of which is a straightforward calculation.

\begin{lem}\label{ls2} For any $a,b,c \in \R$,
\[F_{dT}(a,b,c)=\left(F_{dT^*}^{(1)}(b,c), F_{dT^*}^{-1}\left(a,F_{dT^*}^{(2)}(b,c)\right)\right).\]
In particular, $F_{dT}(a,b,c)=(d,e,f)$ if and only if $F_{dT^*}^{(1)}(b,c)=d$ and $F_{dT^*}^{-1}(a,g)=(e,f)$ where $g:=F_{dT^*}^{(2)}(b,c)$.
\end{lem}

\begin{rem}\label{rs2} Lemma \ref{ls2} is the discrete analogue of the ultra-discrete Lemma \ref{ls1}, and we have a similar graphical interpretation in this case. In particular, Lemma \ref{ls2} shows that we can decompose the single map $F_{dT}$ into two maps $F_{dT^*}$ and $F_{dT^*}^{-1}$:
\[\xymatrix@C-15pt@R-15pt{\boxed{F_{dT^*}\vphantom{F_{udT^*}^{-1}}} & b+c &\boxed{F_{dT^*}^{-1}}& \frac{ab}{b+c}&\\
            c \ar[rr] &&\frac{c}{\sqrt{b^2+bc}} \ar[rr]&& \frac{ac}{b+c}.\\
             & b \ar[uu]&&a\ar[uu]&}\]
\end{rem}

In light of Lemma \ref{ls2}, we introduce locally-defined dynamics $\mathbb{F}_{dT}=(m,(F_n)_{n \in \Z})$, where $m=1$, and $F_n=F_{dT^*}$ for $n$ even, and $F_{n}=(F_{dT^*})^{-1}$ for $n$ odd. For this system, we take the configuration space to be $(0,\infty)^\mathbb{Z}$, with typical elements being written $(\eta_n)_{n\in\mathbb{Z}}$, and the carrier space to also be $(0,\infty)^\mathbb{Z}$, with typical elements being denoted $(\tilde{U}_n)_{n\in\mathbb{Z}}$. In particular, corresponding to Corollary \ref{c57} in the discrete case, we have the following.

\begin{cor}\label{c58} If $(I_n^t,J_n^t,U_n^t)_{n,t\in\mathbb{Z}}$ is a solution to \eqref{DTODA} and we set
\begin{equation}\label{ad3}
\eta_{2n-1}^t=I_n^t,\qquad \eta_{2n}^t=J_n^t,\qquad \tilde{U}_{2n-1}^t=U_n^t,\qquad \tilde{U}_{2n}^t=F_{dT^*}^{(2)}\left(J_n^t,U_n^t\right),
\end{equation}
then $(\eta_n^t,\tilde{U}_n^t)_{n,t\in\mathbb{Z}}$ is a solution to
\begin{equation}\label{ad4}
\left(\eta_{n-1}^{t+1},\tilde{U}_n^t\right)=F_n\left(\eta_n^t,\tilde{U}_{n-1}^t\right),\qquad \forall n,t\in\mathbb{Z},
\end{equation}
where the maps $(F_n)_{n\in\mathbb{Z}}$ are given by $\mathbb{F}_{dT}$. Moreover, if $(\eta_n^t,\tilde{U}_n^t)_{n,t\in\mathbb{Z}}$ is a solution to \eqref{ad4} and the first three equations of \eqref{ad3} hold, then $(I_n^t,J_n^t,U_n^t)_{n,t\in\mathbb{Z}}$ is a solution to \eqref{DTODA}.
\end{cor}

The map $\pi_{dT}$ from $(I,J)\in((0,\infty)^2)^\Z$ to $\eta\in(0,\infty)^\Z$ given by setting $\eta_{2n-1}=I_n$ and $\eta_{2n}=J_n$ is clearly a bijection, and so to tackle the initial value problem of Theorem \ref{dtodathm}, it thus suffices to study the corresponding problem for the system $\mathbb{F}_{dT}$. To align with Section \ref{general}, we introduce the following maps:
\vspace{-20pt}
\begin{multicols}{2}

\begin{eqnarray*}
\mathcal{A}^{dT}:(0,\infty)^\mathbb{Z}&\rightarrow&\mathbb{R}^\mathbb{Z}\\
\left(\eta_n\right)_{n\in\mathbb{Z}}&\mapsto&\left(\mathcal{A}^{dT}_n(\eta_n)\right)_{n\in\mathbb{Z}},
\end{eqnarray*}

\begin{eqnarray*}
\mathcal{B}^{dT}:(0,\infty)^\mathbb{Z}&\rightarrow&\mathbb{R}^\mathbb{Z}\\
\left(\tilde{U}_n\right)_{n\in\mathbb{Z}}&\mapsto&\left(\mathcal{B}^{dT}_n(\tilde{U}_n)\right)_{n\in\mathbb{Z}},
\end{eqnarray*}
\end{multicols}
\noindent
where $\mathcal{A}^{dT}_n(\eta_n):=(-1)^{n+1}\log\eta_n$ and $\mathcal{B}^{dT}_n(\tilde{U}_n)=-\log\tilde{U}_n$. As the next lemma demonstrates, these maps provide the link between $\mathbb{F}_{dT}$ and $\mathbb{K}^{\sum^*}$, as defined at the start of Subsection \ref{s45}. The easy proof is omitted.

\begin{lem}\label{ll4} The following statements hold.\\
(a) On $(0,\infty)^2$,
\[\left(\mathcal{A}^{dT}_{2n-1}\times\mathcal{B}^{dT}_{2n}\right)\circ F_{dT^*}=K^{\sum^*}\circ\left(\mathcal{A}^{dT}_{2n}\times\mathcal{B}^{dT}_{2n-1}\right),\]
and also
\[\left(\mathcal{A}^{dT}_{2n}\times\mathcal{B}^{dT}_{2n+1}\right)\circ F_{dT^*}^{-1}=\left(K^{\sum^*}\right)^{-1}\circ\left(\mathcal{A}^{dT}_{2n+1}\times\mathcal{B}^{dT}_{2n}\right).\]
(b) Both $\mathcal{A}^{dT}:(0,\infty)^\Z\rightarrow\mathbb{R}^\Z$ and $\mathcal{B}^{dT}:(0,\infty)^\Z\rightarrow\mathbb{R}^\Z$ are bijections.\\
(c) The map $\pi\circ\mathcal{A}^{dT}:(0,\infty)^\Z\rightarrow \mathcal{S}^{0}$ is a bijection, and moreover $\pi\circ\mathcal{A}^{dT}\circ\pi_{dT}=S_{dT}$.\\
(d) It is the case that $\mathcal{A}^{dT}\circ\pi_{dT}(\mathcal{C}_{dT})=\mathcal{X}^{lin}$.
\end{lem}

\begin{proof}[Proof of Theorem \ref{dtodathm}] Since the proof is essentially identical to that of Theorem \ref{udtodathm}, but applying Corollary \ref{c58} and Lemma \ref{ll4} in place of Corollary \ref{c57} and Lemma \ref{ll3}, respectively, we omit it.
\end{proof}

\begin{rem}\label{consrem} For the KdV-type systems, it is clear how the conserved quantities of the related P-map relate to conserved quantities of the original system. Indeed, in the ultra-discrete case, the conserved quantity for $K^\vee$ written in terms of the original variables is equal to
\[\mathcal{A}_n^{udK}(\eta_n)-2\mathcal{B}_b^{udK}(U_{n-1})=L-1-2(\eta_n+U_{n-1}),\]
and thus the P-map property is equivalent to the conservation of $\eta_n+U_{n-1}$ by $F_{udK}^{(L)}$. Similarly, in the discrete case, the conserved quantity for $K^{\sum}$ written in terms of the original variables is equal to
\[\mathcal{A}_n^{dK}(\omega_n)-2\mathcal{B}_b^{dK}(U_{n-1})=-2\log\delta-2(\log\omega_n+\log U_{n-1}),\]
and thus the P-map property is equivalent to the conservation of $\log\omega_n+\log U_{n-1}$ by $F_{dK}^{(\delta)}$. For the Toda-type maps, the situation is not quite so straightforward, but we can still demonstrate that the P-map property is related to a conserved quantity for the original map. Indeed, in the ultra-discrete case, the conservation laws for $K^{\vee^*}$ and $(K^{\vee^*})^{-1}$ can be written in terms of the original variables as:
\begin{eqnarray*}
\mathcal{A}^{udT}_{2n}(E^t_n)-2\mathcal{B}^{udT}_{2n-1}(U^t_n)&=&\mathcal{A}^{udT}_{2n-1}(Q^{t+1}_n)-2\mathcal{B}^{udT}_{2n}\left(F_{udT*}^{(2)}(E_n^t,U_n^t)\right),\\
\mathcal{A}^{udT}_{2n+1}(Q^t_{n+1})-2\mathcal{B}^{udT}_{2n}\left(F_{udT*}^{(2)}(E_n^t,Q_n^t)\right)&=&\mathcal{A}^{udT}_{2n}(E^{t+1}_n)-2\mathcal{B}^{udT}_{2n+1}(U^t_{n+1}).
\end{eqnarray*}
Eliminating the term $2\mathcal{B}^{udT}_{2n}(F_{udT*}^{(2)}(E_n^t,Q_n^t))$, this gives the following conservation law for the original lattice variables:
\[E_n^t-Q_{n+1}^t-2U_n^t=E_n^{t+1}-Q_n^{t+1}-2U_{n+1}^t,\]
which, as we observed in Subsection \ref{conssec}, is obtained by combining the conservation laws for the mass and length of the interval to which the local dynamics applies. An essentially similar argument yields the corresponding result in the case of the discrete Toda lattice.
\end{rem}

\section{Discussion and modifications of main results}\label{discussionsec}

In this section, we discuss other possible boundary conditions for the equations \eqref{UDKDV}, \eqref{DKDV}, \eqref{UDTODA} and \eqref{DTODA}, various adaptations of the conclusions of Section \ref{dissec}, and how our results incorporate and extend known results, including how solutions for the ultra-discrete systems can be obtained from those for discrete systems via ultra-discretization.

\subsection{Boundary conditions}\label{boundarysec}

Given the lattice structures that \eqref{UDKDV}, \eqref{DKDV}, \eqref{UDTODA} and \eqref{DTODA} give rise to, it is easy to see that if we are given a boundary condition that involves configuration and carrier data along a down-right path from the upper-left corner of the plane to the lower-right corner (see Figure \ref{bcfig}(a) for an example), then the remaining configuration and carrier variables are determined uniquely; this might be seen as a solution to a kind of `Cauchy problem'. Similarly, if one has a boundary condition that runs along a down-right path from $+\infty$ in the vertical direction to $+\infty$ in the horizontal direction (as shown in Figures \ref{bcfig}(b,c)), one can solve the equations uniquely in the upper-right part of the plane. (Note that the configuration of Figure \ref{bcfig}(b) gives what might be described as a `Goursat problem'.) And, applying the self-inverse symmetry discussed in Subsection \ref{sisec}, one can solve the equations uniquely in the lower-left part of the plane if we have a boundary condition that runs along a down-right path from $-\infty$ in the horizontal direction to $-\infty$ in the vertical direction.

The horizontal boundary condition shown in Figure \ref{bcfig}(d), which is the one considered in this article, presents more of a challenge. As we have already noted, in this case, both the existence and uniqueness of solutions are not immediate. Whilst our main results tackle this problem for configurations whose path encoding $S$ lies in $\mathcal{S}^0\cap\mathcal{S}^{lin}$, our arguments also imply that one can construct a carrier for the initial configuration if and only if the associated past maximum for the path encoding, $M(S)_n$ (where $M$ represents the relevant functional for the model, as shown in Table \ref{Mtable}), is finite for $n\in\mathbb{Z}$.
Indeed, this claim readily follows from the argument we give to show non-existence of a carrier when $\lim_{n\rightarrow-\infty}S_n=\infty$, which is Assumption \ref{a1}(a) and essentially checked for each of the models in Theorem \ref{thm:examples}. Hence one can not even start the forward-in-time dynamics from configurations for which $\lim_{n\rightarrow-\infty}S_n=\infty$ holds. Similarly, the backward-in-time dynamics can not be started from configurations with path encodings satisfying $\lim_{n\rightarrow\infty}S_n=-\infty$. That we can construct one time step of the dynamics, however, does not guarantee we can continue to iterate the procedure for all time; see the next subsection for continuation of this point.

\begin{figure}
\begin{multicols}{2}
\raggedright

(a)\medskip

\includegraphics[width = 0.45\textwidth]{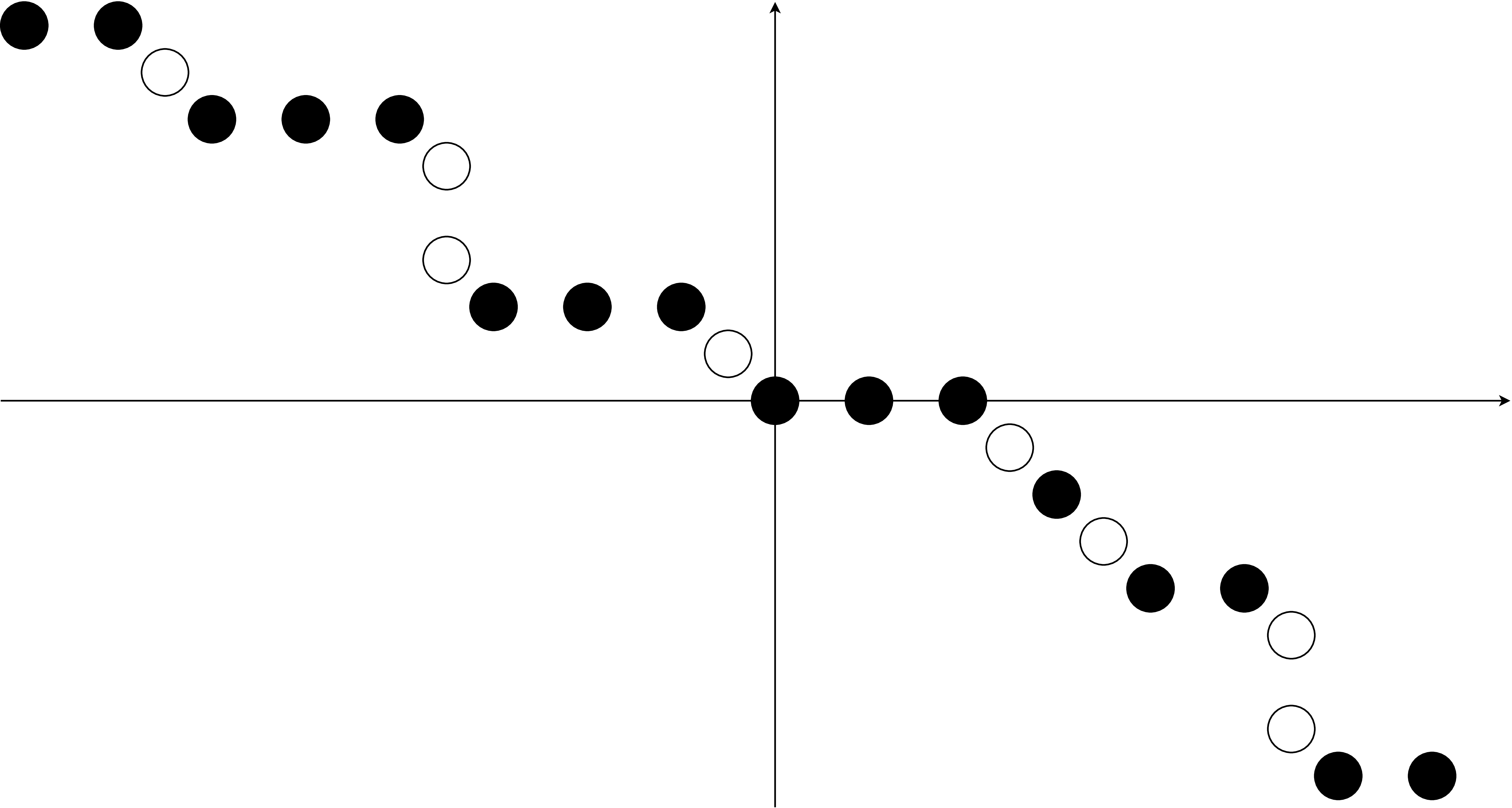}

(b)\medskip

\includegraphics[width = 0.45\textwidth]{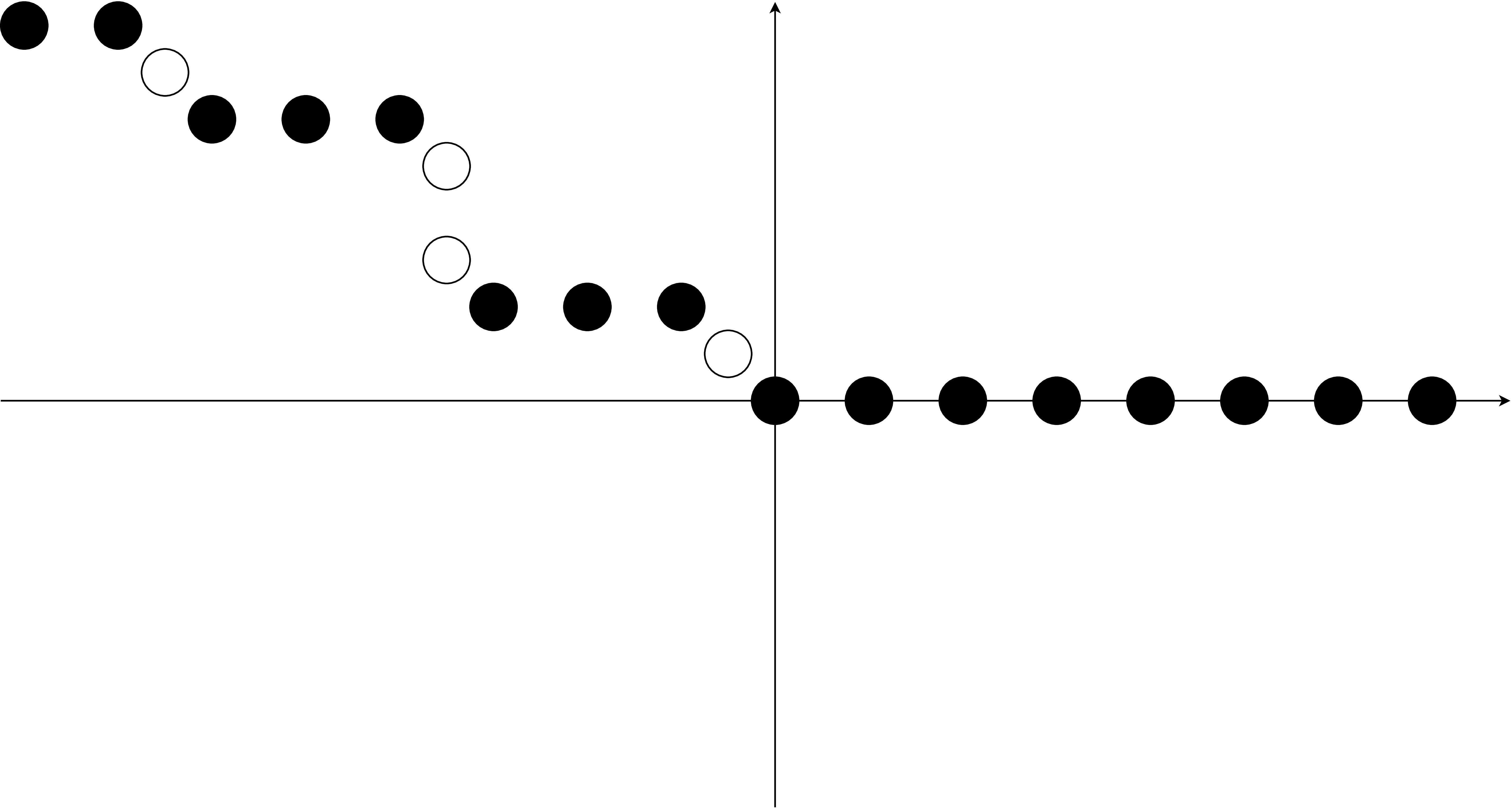}

(c)\medskip

\includegraphics[width = 0.45\textwidth]{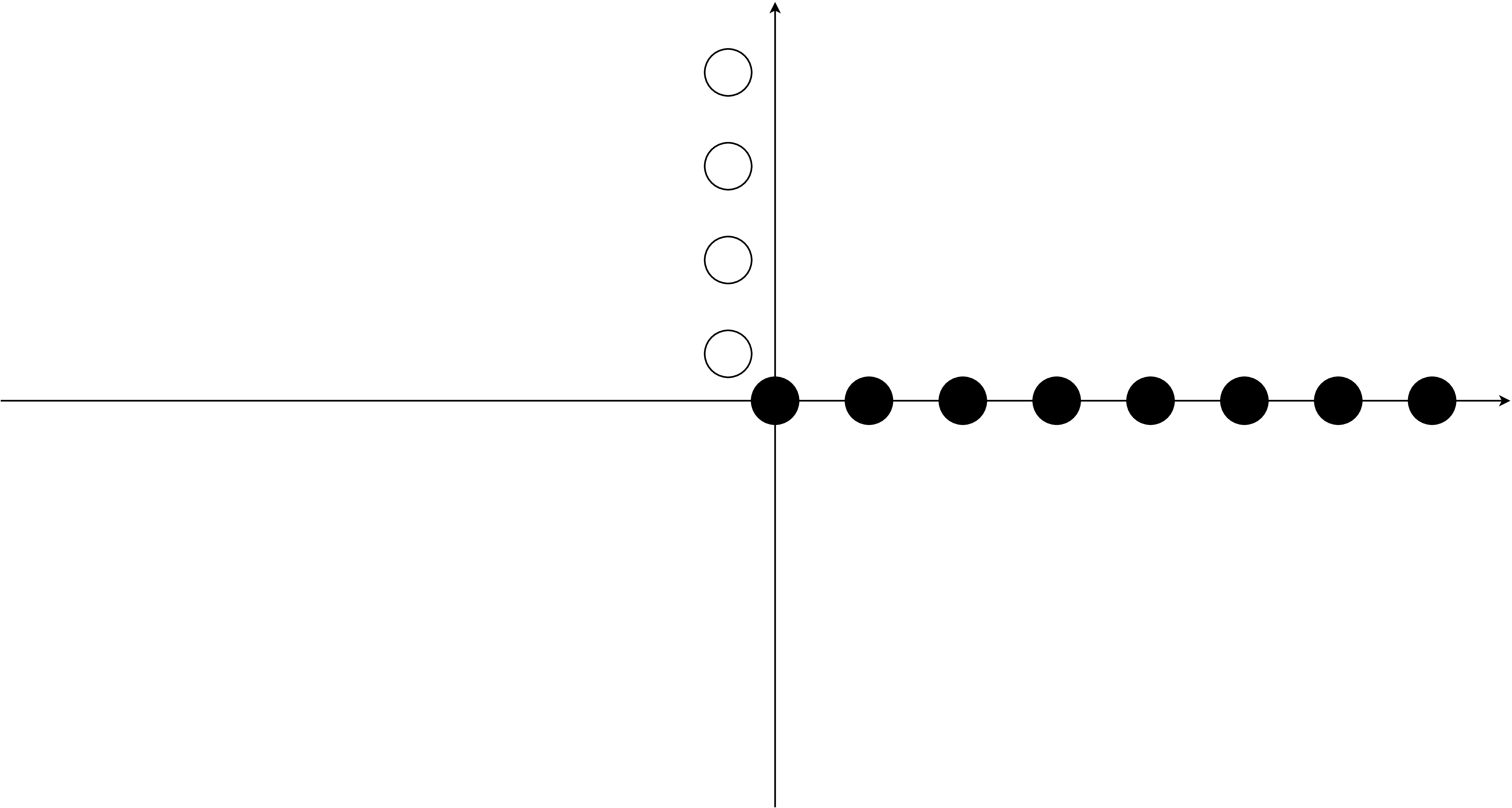}

(d)\medskip

\includegraphics[width = 0.45\textwidth]{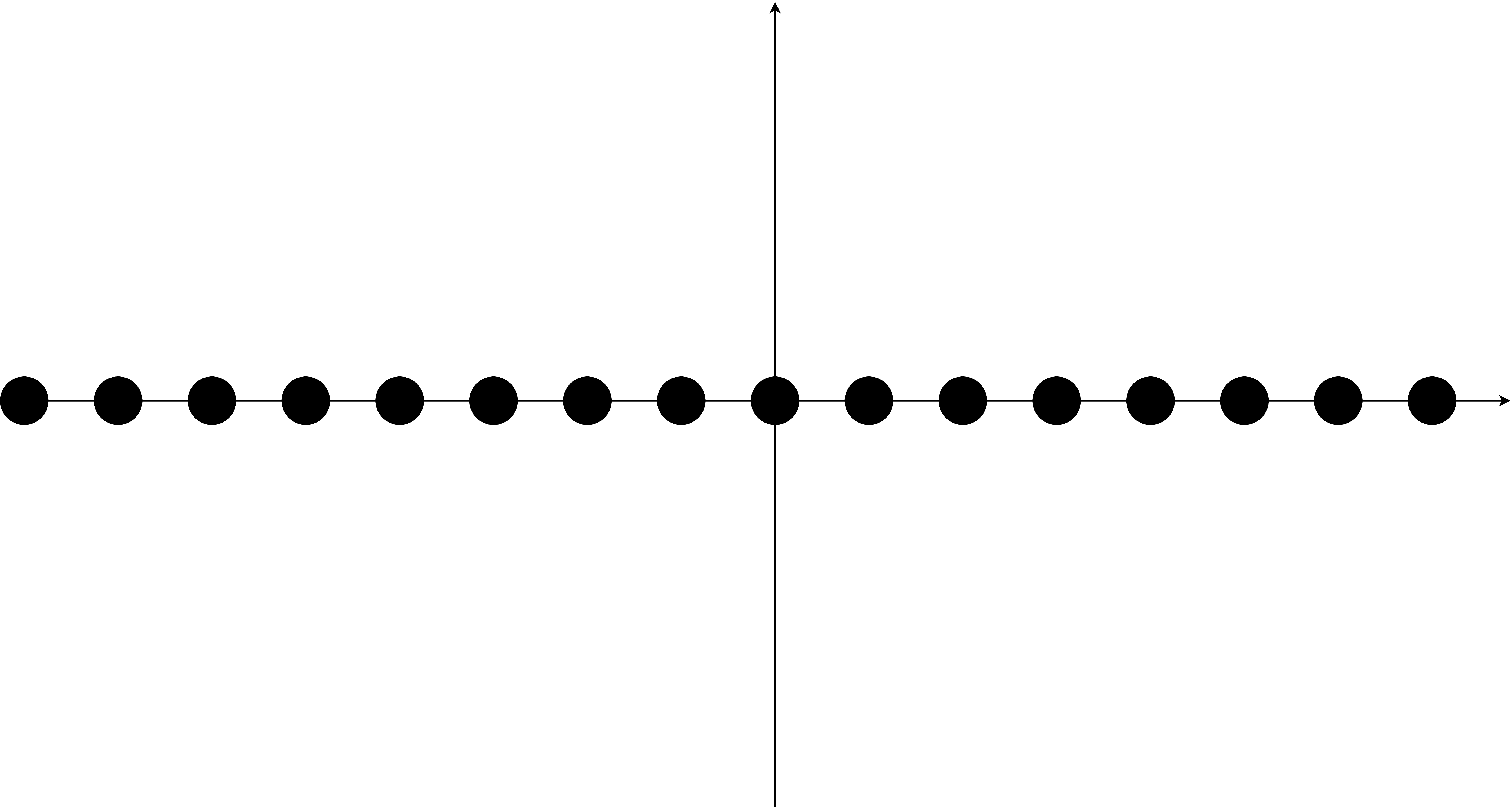}
\end{multicols}
\caption{Schematic representation of some possible boundary conditions. Black circles represent lattice locations corresponding to elements of the configuration space (in the Toda case, these circles correspond to a two-dimensional variable), and white circles to elements of the carrier space.}\label{bcfig}
\end{figure}

\subsection{Global and local solutions}

For each of the four systems studied here, it is straightforward to construct configurations for which the path encoding satisfies $M(S)_n<\infty$ for $n\in\mathbb{Z}$, and so one forward-in-time step of the dynamics is possible, but for which the subsequent configuration does not admit a carrier. In particular, for the ultra-discrete systems, any function $S$ which satisfies $\lim_{n\rightarrow-\infty}M(S)_n\in\mathbb{R}$ and $\liminf_{n\rightarrow-\infty}(S_n+S_{n-1})=-\infty$ has this property. For the discrete systems, since it is always the case that $\lim_{n\rightarrow-\infty}M(S)_n=-\infty$ (when $M(S)_n$ is finite for $n\in\mathbb{Z}$), one has to work slightly harder, but still it is an elementary exercise to find an $S$ admitting a carrier and also satisfying $M(2M(S)-S)_n=\infty$ for any $n\in\mathbb{Z}$, which implies there is no carrier after one time step.
Hence the existence of local solutions does not imply the existence of global solutions for the initial value problem.

In the other direction, we remark that whilst the existence of global solutions clearly implies the existence of local solutions, the uniqueness of the former does not imply the uniqueness of the latter. Indeed, again pointing to \cite[Remark 2.11]{CKST} for concrete examples, we expand on the comment in the introduction to note that for BBS configurations with $\limsup_{n\rightarrow-\infty}S_n<\limsup_{n\rightarrow+\infty}S_n$, there exist multiple carriers (in fact, infinitely many), which means that multiple possible interpretations of the dynamics are possible (as many as the difference between the two $\limsup$s).
However, if $\lim_{n\rightarrow-\infty}S_n=-\infty$, then at most one choice of these carriers results in a new configuration that also has a carrier. In the general setting of this work, we have taken the possibility of continuing the dynamics beyond one step as definitive, using it to introduce the notion of a canonical carrier (recall Definition \ref{carrierdef}(c) above), and built our study of the dynamics of the relevant systems on this concept. In the context of the BBS, alternative definitions of a canonical carrier were given in \cite{CS,CKST}. These definitions are perhaps more straightforward to justify physically, in that they are clearly presented in terms of the current state of the system, rather than a consideration of its future evolution, and they also make clear that the choice of canonical carrier is related to the prohibition of mass entering the system from $-\infty$. We stress, however, that the various definitions of canonical carrier coincide on the intersection of the models under consideration.

\subsection{Forward and backward equations}\label{fbeq} For a given initial condition, it is clear that solving any of the initial value problems \eqref{UDKDV}, \eqref{DKDV}, \eqref{UDTODA} and \eqref{DTODA} is equivalent to solving simultaneously both the corresponding forward, i.e.\ $t\in\mathbb{Z}_+$, and backward, i.e.\ $t\in\mathbb{Z}_-$, versions of the problem, where we define $\mathbb{Z}_-:=-\mathbb{N}=\{\dots,-2,-1\}$. Moreover, given the self-inverse property of the various systems (as discussed in Subsection \ref{sisec}), the forward and backward problems are themselves equivalent. Hence one might alternatively have chosen to present as the main results of this article conclusions regarding the forward problems, and the aim of this subsection is to present a statement along these lines. In this direction, we start by noting that the structure of the lattices means that the issue of existence and uniqueness of solutions to the forward problem depends on the behaviour of the configuration at $-\infty$, and we correspondingly introduce the subset of path encodings
\[\mathcal{S}_-^{lin}:=\left\{S\in\mathcal{S}:\:\lim_{n\rightarrow-\infty}\frac{S_n}{n}>0\right\},\]
and the corresponding configurations
\begin{equation}\label{onesidedC}
\mathcal{C}_{udK,-}^{(L)}:=\left(S_{udK}^{(L)}\right)^{-1}\left(\mathcal{S}_-^{lin}\right),\qquad\mathcal{C}_{dK,-}^{(\delta)}:=\left(S_{dK}^{(\delta)}\right)^{-1}\left(\mathcal{S}_-^{lin}\right),
\end{equation}
\[\mathcal{C}_{udT,-}:=\left(S_{udT}\right)^{-1}\left(\mathcal{S}_-^{lin}\right),\qquad\mathcal{C}_{dT,-}:=\left(S_{dT}\right)^{-1}\left(\mathcal{S}_-^{lin}\right).\]
This notation allows us to present the following result, which is proved in Section \ref{pmodsec}.

\begin{thm}\label{fthm}
The forward (i.e.\ $t\in\mathbb{Z}_+$) versions of the initial value problems \eqref{UDKDV}, \eqref{DKDV}, \eqref{UDTODA} and \eqref{DTODA} have a unique solution whenever the initial condition lies in $\mathcal{C}_{udK,-}^{(L)}$, $\mathcal{C}_{dK,-}^{(\delta)}$, $\mathcal{C}_{udT,-}$ and $\mathcal{C}_{dT,-}$, respectively. For the range of $t$ considered, these solutions have the same form as the solutions identified in Theorems \ref{udkdvthm}, \ref{dkdvthm}, \ref{udtodathm} and \ref{dtodathm}, and remain in the sets $\mathcal{C}_{udK,-}^{(L)}$, $\mathcal{C}_{dK,-}^{(\delta)}$, $\mathcal{C}_{udT,-}$ and $\mathcal{C}_{dT,-}$.
\end{thm}

\begin{rem}
The above result covers the `supercritical solutions' for the original BBS studied in \cite[Section 4.1]{CKST} and \cite{LLP}, for example, in which the particle density of a semi-infinite configuration is greater than that of empty boxes. As is discussed in \cite{CKST}, this leads to particles being transported out of the system on the first step of the dynamics, which is a non-reversible outcome. The article \cite{LLP} discusses the sizes of the solitons that appear in the case of an independent and identically distributed supercritical configuration.
\end{rem}

As a corollary, we present a similar statement for the backward problem. For the statement of this, we set
\[\mathcal{S}_+^{lin}:=\left\{S\in\mathcal{S}:\:\lim_{n\rightarrow+\infty}\frac{S_n}{n}>0\right\},\]
and define the configurations $\mathcal{C}_{udK,+}^{(L)}$, $\mathcal{C}_{dK,+}^{(\delta)}$, $\mathcal{C}_{udT,+}$ and $\mathcal{C}_{dT,+}$ analogously to \eqref{onesidedC}.

\begin{cor}\label{bcor} The backward (i.e.\ $t\in\mathbb{Z}_-$) versions of the initial value problems \eqref{UDKDV}, \eqref{DKDV}, \eqref{UDTODA} and \eqref{DTODA} have a unique solution whenever the initial condition lies in $\mathcal{C}_{udK,+}^{(L)}$, $\mathcal{C}_{dK,+}^{(\delta)}$, $\mathcal{C}_{udT,+}$ and $\mathcal{C}_{dT,+}$, respectively. For the range of $t$ considered, these solutions have the same form as the solutions identified in Theorems \ref{udkdvthm}, \ref{dkdvthm}, \ref{udtodathm} and \ref{dtodathm}, and remain in the sets $\mathcal{C}_{udK,-}^{(L)}$, $\mathcal{C}_{dK,-}^{(\delta)}$, $\mathcal{C}_{udT,-}$ and $\mathcal{C}_{dT,-}$.
\end{cor}

\subsection{Closed form equations for discrete systems}\label{cfsec}

It is possible to rewrite both the discrete KdV and Toda equations in terms of the variables $\omega$ and $(I,J)$ alone, i.e.\ without the auxiliary carrier variable $U$ (see equations \eqref{DKomega1} and \eqref{DToda2} below). In this subsection, we describe consequences of our main theorems for the existence and uniqueness of solutions to the initial value problems based on these closed form equations. We first present our result for the discrete KdV model.

\begin{thm}\label{ttt1} (a) Suppose $\omega \in \mathcal{C}_{dK}^{(\delta)}$. It then holds that if $(\omega^t)_{t\in\mathbb{Z}}$ is defined as at
\eqref{omegatdef}, then $(\omega^t_n)_{n,t \in \Z}$ solves the following quad-equation:
\begin{align}\label{DKomega1}
\begin{cases}
& \omega^0  =\omega, \\
& \frac{1}{\omega_{n+1}^{t+1}} -\frac{1}{\omega_n^t} =\delta(\omega_{n+1}^{t}-\omega_n^{t+1}),
\end{cases}
\end{align}
for all $n,t\in\mathbb{Z}$.\\
(b) Suppose $\omega \in \mathcal{C}_{dK}^{(\delta)}$ and $(\omega^t)_{t\in\mathbb{Z}}$ are as in part (a). It then holds that  $(\omega^t_n)_{n,t \in \Z}$ is the unique solution to \eqref{DKomega1} with $\omega^t \in \mathcal{C}_{dK}^{(\delta)}$ for all $t\in\mathbb{Z}$.\\
(c) If $\hat{\omega} \in (0,\infty)^{\Z}$ satisfies
\begin{equation}\label{ahjk}
\lim_{n \to \pm \infty}\frac{\sum_{k=1}^{n}\log \hat{\omega}_k}{n} > \frac{-\log \delta}{2}
\end{equation}
then there exists a unique solution $(\hat{\omega}^t_n)_{n,t \in \Z}$ to \eqref{DKomega1} with initial condition $\hat{\omega}$, and such that $\hat{\omega}^t$ satisfies \eqref{ahjk} for all $t\in\mathbb{Z}$.\\
(d) If the $\omega$ and $\hat{\omega}$ of parts (b) and (c) are related by $\hat{\omega}_n=\frac{1}{\delta \omega_n}$, $n\in\mathbb{Z}$, then it holds that $\hat{\omega}_n^t=\frac{1}{\delta\omega^{-t}_n}$ for all $n,t\in\mathbb{Z}$.
\end{thm}

\begin{rem}\label{nouni} It is easy to see that non-uniqueness does not hold under \eqref{DKomega1} alone. Indeed, suppose $\omega\in  \mathcal{C}_{dK}^{(\delta)}$, and define $(\omega^t_n)_{n,t \in \Z}$ by setting $\omega^0=\omega$ and via the relation $\omega^{t+1}_n=\frac{1}{\delta \omega^t_n}$ for all $n,t\in\mathbb{Z}$ otherwise. It is then easy to check that \eqref{DKomega1} holds, but $\omega^t\not\in\mathcal{C}_{dK}^{(\delta)}$ for odd $t$. Hence we obtain a different solution to that given by \eqref{omegatdef}. An analogous comment applies when the initial condition is chosen to satisfy \eqref{ahjk}. These alternating dynamics are reminiscent of the trivial dynamics seen under invariant measures for the BBS at critical density (see \cite[Theorem 1.4]{CKST}), whereby, on each time step, particles and empty boxes are exchanged. The latter dynamics are consistent with \eqref{bbsudkdv2} if we suppose the carrier always takes the value $\infty$. As one can see from the form of the equation \eqref{DKDV}, we similarly see that setting $U_n^t=\infty$ for all $n,t\in\mathbb{Z}$ in \eqref{DKDV} yields the solution for \eqref{DKomega1} just described.
\end{rem}

\begin{rem}\label{per} The path encodings of the configurations $\omega$ and $\hat{\omega}$ of part (d) above satisfy $S_{dK}^{(\delta)}(\hat{\omega})=-S_{dK}^{(\delta)}({\omega})$. In particular, the condition at \eqref{ahjk} gives a path encoding with negative drift, and so it is not possible to define a carrier for $\hat{\omega}$ as in Subsection \ref{dsec1}. Nonetheless, the result shows that we can still understand the dynamics of the system via a Pitman-type transformation. Indeed, it follows from Theorems \ref{dkdvthm} and \ref{ttt1} that
\[\hat{\omega}^t=\left(-S_{dK}^{(\delta)}\right)^{-1}\circ\left(T^{\sum}\right)^{-t}\circ \left(-S_{dK}^{(\delta)}\right)\left(\hat{\omega}\right),\]
for which the operation on the path encoding can be thought of as repeated reflection in the `future maximum' (where we interpret `future maximum' in the same loose sense we interpreted `past maximum' in Table \ref{Mtable}).
\end{rem}

\begin{rem}The ultra-discrete KdV equation \eqref{UDKDV} is formally rewritten as
\begin{equation}\label{formalrem}
\eta_n^{t+1}=\min\left\{L-\eta^t_n,\sum_{m=-\infty}^{n-1}(\eta_m^t-\eta_m^{t+1})\right\}
\end{equation}
(cf.\ \eqref{bbsudkdv1}), but the infinite sum does not converge in a generality that would mean it makes sense to study it rigourously.
\end{rem}

We next turn to the discrete Toda system.

\begin{thm}\label{ttt2}
(a) Suppose $(I,J) \in \mathcal{C}_{dT}$. It then holds that if $(I^t,J^t)_{t\in\mathbb{Z}}$ is defined as at
\eqref{ijij}, then $(I^t_n,J^t_n)_{n,t \in \Z}$ solves
\begin{equation}\label{DToda2}
\begin{cases}
(I^0,J^0)=(I,J), \\
I^{t+1}_n =I^t_n+J_n^t-J_{n-1}^{t+1}, \\
J_n^{t+1} =\frac{I^t_{n+1}J_n^t}{I_n^{t+1}},
\end{cases}
\end{equation}
for all $n,t\in\mathbb{Z}$.\\
(b) Suppose $(I,J) \in \mathcal{C}_{dT}$ and $(I^t,J^t)_{t\in\mathbb{Z}}$ are as in part (a). It then holds that  $(I^t_n,J^t_n)_{n,t \in \Z}$ is the unique solution to \eqref{DToda2} with $(I^t,J^t) \in \mathcal{C}_{dT}$ for all $t\in\mathbb{Z}$.\\
(c) Let $(\hat{I},\hat{J}) \in ((0,\infty)^2)^{\Z}$, and define $(I,J)\in((0,\infty)^2)^{\Z}$ by setting
\begin{equation}
I_{n+1}=\hat{J}_{n},\qquad J_{n}=\hat{I}_{n},\qquad \forall n\in\mathbb{Z}.\label{huil}
\end{equation}
If $(I,J)\in\mathcal{C}_{dT}$, then there exists a unique solution $(\hat{I}^t_n,\hat{J}^t_n)_{n,t \in \Z}$ to \eqref{DToda2} with initial condition $(\hat{I},\hat{J})$, and such that if $({I}^{-t},{J}^{-t})$ is defined from $(\hat{I}^t,\hat{J}^t)$ by setting
\begin{equation}
I_{n+1+t}^{-t}=\hat{J}^{t}_{n},\qquad J_{n+t}^{-t}=\hat{I}^{t}_{n},\qquad \forall n\in\mathbb{Z},\label{huil2}
\end{equation}
then $({I}^t,{J}^t)\in\mathcal{C}_{dT}$ for all $t\in\mathbb{Z}$.\\
(d) If the $(I,J)$ and $(\hat{I},\hat{J})$ of parts (b) and (c) are related by \eqref{huil}, then the corresponding $(I^{-t},J^{-t})$ and $(\hat{I}^t,\hat{J}^t)$ are related by \eqref{huil2} for all $t\in\mathbb{Z}$.
\end{thm}

\begin{rem} Similarly to Remark \ref{nouni}, we have that non-uniqueness does not hold under \eqref{DToda2} alone. Indeed, suppose $(I,J)\in  \mathcal{C}_{dT}$, and define $(I^t_n,J^t_n)_{n,t \in \Z}$ by setting $(I^0,J^0)=(I,J)$ and via the relations $I^{t+1}_n=J_n^t$, $J_n^{t+1}=I^t_{n+1}$ for all $n,t\in\mathbb{Z}$ otherwise. It is then easy to check that \eqref{DToda2} holds, but $(I^t,J^t)\not\in\mathcal{C}_{dT}$ for odd $t$. Hence we obtain a different solution to that given by \eqref{ijij}. A similar comment applies when the initial condition is chosen as per $(\hat{I},\hat{J})$ of part (c) of the above result. Moreover, we note that these alternating dynamics arise from setting $U_n^t=0$ for all $n,t\in\mathbb{Z}$ in \eqref{DTODA} (cf.\ the final comment in Remark \ref{nouni}).
\end{rem}

\begin{rem}[Cf.\ Remark \ref{per}.] The path encodings of the configurations $(I,J)$ and $(\hat{I},\hat{J})$ of part (d) above satisfy $S_{dT}(\hat{I},\hat{J})=-\theta\circ S_{dT}({I,J})$. In particular, since $(I,J)\in\mathcal{C}_{dT}$, the path encoding of $(\hat{I},\hat{J})$ has negative drift, and so it is not possible to define a carrier for $(\hat{I},\hat{J})$ as in Subsection \ref{dsec1toda}. Nonetheless, the result shows that we can still understand the dynamics of the system via a Pitman-type transformation. Indeed, it follows from Theorems \ref{dtodathm} and \ref{ttt2} that
\begin{eqnarray*}
-\theta\circ S_{dT}\left(\hat{I}^t,\hat{J}^t\right)&=&\theta^{2(t+1)}\circ S_{dT}\left({I}^{-t},{J}^{-t}\right)\\
&=&\theta^{2(t+1)}\circ\left(\theta \circ T^{\sum^*}\right)^{-t}\circ S_{dT}\left({I},{J}\right)\\
&=&\theta^2\circ\left(\theta^{-1} \circ T^{\sum^*}\right)^{-t}\circ \left(-\theta^{-1}\circ S_{dT}\right)\left(\hat{I},\hat{J}\right).
\end{eqnarray*}
This yields
\[\left(\hat{I}^t,\hat{J}^t\right)=\left(-\theta^{-1}\circ S_{dT}\right)^{-1}\circ\left(\theta^{-1}\circ T^{\sum^*}\right)^{-t}\circ \left(-\theta^{-1}\circ S_{dT}\right)\left(\hat{I},\hat{J}\right),\]
for which the operation on the path encoding can again be thought of as repeated reflection in the future maximum (combined with an appropriate shift).
\end{rem}

\begin{rem} As for the ultra-discrete Toda lattice equation \eqref{UDTODA}, this has the formal closed form
\[\begin{cases}
Q_{n}^{t+1}=\min \left\{\sum_{m=-\infty}^nQ_m^t-\sum_{m=-\infty}^{n-1}Q_m^{t+1},E_n^t\right\}, \\
E_{n}^{t+1}=Q_{n+1}^t+E_{n}^t-Q_{n}^{t+1},\\
\end{cases}\]
(cf. \eqref{formalrem}), which again incorporates infinite sums that do not converge in broad generality, and so we do not pursue it further.
\end{rem}

\begin{rem} The forward versions of Theorems \ref{ttt1}(a,b) and \ref{ttt2}(a,b) hold for $\omega\in\mathcal{C}_{dK,-}^{(\delta)}$, $(I,J)\in\mathcal{C}_{dT,-}$, respectively, and similarly backward versions for $\omega\in\mathcal{C}_{dK,+}^{(\delta)}$, $(I,J)\in\mathcal{C}_{dT,+}$. It is further straightforward to deduce one-sided (in time) versions of Theorems \ref{ttt1}(c,d) and \ref{ttt2}(c,d).
\end{rem}

\subsection{Special cases of configurations}\label{specialsec}

For the two ultra-discrete systems, there are important subclasses of solutions which represent special cases of the models in question. Indeed, the original \eqref{UDKDV} equation corresponds to solutions $\eta\in[0,L]^\mathbb{Z}$, and the BBS with box capacity $L$ to solutions $\eta\in\{0,1,\dots,L\}^\mathbb{Z}$ \cite{TH, TMcar}. Moreover, it is also the case that the original BBS, with box capacity $1$, can be described by \eqref{UDTODA} with configurations such that $(Q_n,E_n)\in\mathbb{N}^2$ \cite{NTT}. And, it was for configurations such that $(Q_n,E_n)\in(0,\infty)^2$ for which \eqref{UDTODA} was originally formulated. The following theorem demonstrates that these particular examples can be included in our framework, and we expand on the connections in the subsequent remark. To state the result, we will say a subset of the configuration space is \emph{closed} if when the dynamics are started from a configuration in this subset, the systems remains there for all time. More precisely, for the \eqref{UDKDV} system, we say $\mathcal{C}\subseteq \mathcal{C}_{udK}^{(L)}$ is closed if $\eta^0\in \mathcal{C}$ implies that the global solution of Theorem \ref{udkdvthm} satisfies $\eta^t\in\mathcal{C}$ for all $t\in\mathbb{Z}$, and similarly for the other systems.

\begin{thm}\label{specialthm} (a) The following subsets are closed for \eqref{UDKDV}:
\begin{enumerate}
  \item[(i)] $[L_0,L-L_0]^\mathbb{Z}\cap\mathcal{C}_{udK}^{(L)}$, for any $L_0<L/2$;
  \item[(ii)] $(\alpha\mathbb{Z})^\Z\cap\mathcal{C}_{udK}^{(L)}$, for any $\alpha>0$ such that $L\in\alpha\mathbb{Z}$;
  \item[(iii)] $\{0,1,\dots,L\}^\mathbb{Z}\cap\mathcal{C}_{udK}^{(L)}$, for any $L\in\mathbb{N}$.
\end{enumerate}
Moreover, if the configuration takes a value in these subsets, then the carrier corresponding to the global solution of Theorem \ref{udkdvthm} takes values in, respectively:
\begin{enumerate}
  \item[(i)] $[L_0,\infty)^\mathbb{Z}$;
  \item[(ii)] $(\alpha\mathbb{Z})^\Z$;
  \item[(iii)] $\mathbb{Z}_+^\mathbb{Z}$.
\end{enumerate}
(b) The following subsets are closed for \eqref{UDTODA}:
\begin{enumerate}
  \item[(i)] $((0,\infty)^2)^\mathbb{Z}\cap\mathcal{C}_{udT}$;
  \item[(ii)] $((\alpha\mathbb{Z})^2)^\Z\cap\mathcal{C}_{udT}$, for any $\alpha>0$;
  \item[(iii)] $(\mathbb{N}^2)^\mathbb{Z}\cap\mathcal{C}_{udT}$.
\end{enumerate}
Moreover, if the configuration takes a value in these subsets, then the carrier corresponding to the global solution of Theorem \ref{udtodathm} takes values in, respectively:
\begin{enumerate}
  \item[(i)] $(0,\infty)^\mathbb{Z}$;
  \item[(ii)] $(\alpha\mathbb{Z})^\Z$;
  \item[(iii)] $\mathbb{N}^\mathbb{Z}$.
\end{enumerate}
\end{thm}

\begin{rem}\label{origbbsrem} (a) As noted above, the example (a)(iii) is the BBS with box capacity $L$ (and unbounded carrier capacity). Moreover, in \cite{CKST} it was shown that the dynamics of the system are given by the action of the Pitman transform $T$, as defined at \eqref{originalpitman}, on the path encoding defined at \eqref{incs}. Whilst the path encoding of \cite{CKST} matches that of \eqref{1pe} in this case, the operator $T$ is in general different to $T^\vee$, as defined at \eqref{tveedef}. However, for the configurations of example (a)(iii) with $L=1$ (which is the original BBS of \cite{takahashi1990}), we have that
\[M^{\vee}_n=\sup_{m \le n} \left(\frac{S_m+S_{m-1}}{2}\right)=\sup_{m \le n}S_m-\frac12=M_n-\frac12,\]
where $M^\vee$ and $M$ were introduced at \eqref{mveedef} and \eqref{originalM}, respectively. Thus
\[T^\vee(S)=T(S)-1,\]
and since the vertical shift of the path encoding is irrelevant to the dynamics, we see that the two frameworks describe the same evolution of the particle system. Indeed, the carriers arising from the two definitions are the same, since
\[\left(W_{udK}^{(1)}\right)^{-1}\circ W^{\vee}(S)=\left(W_{udK}^{(1)}\right)^{-1}\left(M^{\vee}(S)-S\right)=M^{\vee}(S)-S+\frac{1}{2}=M(S)-S=W,\]
where $W$ was defined at \eqref{originalW}.\\
(b) Recently, the BBS model has been extended to include negative solitons, and in particular configurations taking values in $\mathbb{Z}^\Z$ \cite{HNew, WNSRG}. This version of the model is included in example (a)(ii) of the preceding theorem. We note that if $(\eta^n_t,U^n_t)_{n,t\in\mathbb{Z}}$ is a solution of \eqref{UDKDV} with parameter $L$, then for any $c \in \R$, $(\eta^n_t+c, U^n_t+c)_{n,t\in\mathbb{Z}}$ is a solution of \eqref{UDKDV} with parameter $L+2c$. Such a transformation has been used to map the BBS with negative solitons to the traditional box-ball system \cite{KMT}.\\
(c) As noted prior to the theorem, the original \eqref{UDKDV} system corresponds to example (a)(i) with $L_0=L$.\\
(d) Suppose that we have a BBS configuration $\eta\in\{0,1\}^\mathbb{Z}\cap\mathcal{C}_{udK}^{(L)}$ with $\eta_{0}=0$, $\eta_1=1$, and which has infinitely many $0$s and $1$s in either direction. Define $\dots,Q_0,E_0,Q_1, E_1,\dots$ by supposing $Q_1$ is the number of consecutive $1$s in the configuration that extend from $\eta_1$, $E_1$ is the number of $0$s that follow this, $Q_2$ is the number of $1$s that follow this, and so on (similarly in the negative direction). This gives a configuration $(Q,E)\in(\mathbb{N}^2)^\mathbb{Z}$. If the dynamics of $\eta$ are given by the BBS system, but we appropriately shift space in a way that ensures the boundary condition  $\eta_{0}=0$, $\eta_1=1$ is always fulfilled, then the dynamics of $(Q,E)$ are naturally described by \eqref{UDTODA}, i.e.\ the system is covered by example (b)(iii). For background on the connection between the two systems, see \cite[Section 3]{CST}. Note that the relation between the finite particle \eqref{UDKDV} and \eqref{UDTODA} systems (both of which will be discussed further in the subsequent subsection) was previously known, see \cite{NTT} and \cite[Subsection 3.2]{TTeng}.\\
(e) An alternative path encoding for (a subset of) configurations of the ultra-discrete Toda lattice with values in $(0,\infty)^\mathbb{Z}$, i.e.\ example (b)(i), was introduced in \cite{CST}, which consisted of piecewise linear paths whose gradient alternates between $-1$ and $1$. For this path encoding, the ultra-discrete Toda dynamics are described by the original Pitman transform \eqref{originalpitman}, together with a shift to maintain a local maximum at the origin. The appeal of the picture of \cite{CST} is that it gives a natural continuous state space version of example (a)(iii), but it does not allow the generalization to configurations that admit negative values, as the framework of this article does.\\
(f) With reference to Theorem \ref{fthm} or Corollary \ref{bcor}, one can define a system as being closed for the forward in time or backward in time dynamics, respectively. Replacing $\mathcal{C}_{udK}^{(L)}$ by  $\mathcal{C}_{udK,-}^{(L)}$ or  $\mathcal{C}_{udK,+}^{(L)}$, and $\mathcal{C}_{udT}$ by $\mathcal{C}_{udT,-}$ or $\mathcal{C}_{udT,+}$ (recalling the relevant definitions from Subsection \ref{fbeq}), in the above theorem gives the corresponding forward or backward versions.
\end{rem}

\subsection{Known solutions to initial value problems}\label{knownsec}

In this subsection, we discuss three classes of known solutions to the initial value problems for our integrable systems. Avoiding the issues that come with handling an infinite configuration as we do in this article, the first of these is the class of periodic solutions. The second, for which we will focus on the ultra-discrete KdV system, is the class of finite configurations. (See Remark \ref{fintodarem} for comments on finite configuration solutions to \eqref{UDTODA} and \eqref{DTODA}.) Third, we give a particular class of solutions to \eqref{DKDV} for which the carrier can be expressed explicitly in terms of the configuration. The latter class includes the analogues of the finite configurations of the ultra-discrete model.

\subsubsection{Periodic configurations}

To state our result concerning periodic solutions of \eqref{UDKDV}, \eqref{DKDV}, \eqref{UDTODA} and \eqref{DTODA}, we first introduce the space of path encodings with periodic increments by setting, for $N\in\mathbb{N}$,
\[\mathcal{S}^{per(N)}:=\left\{S\in\mathcal{S}:\:S_{N+n}-S_N=S_n-S_0,\:\forall n\in\mathbb{Z}\right\}.\]
Recalling the definition of `closedness' from the previous subsection, we then have the following.

\begin{thm}\label{perthm} For $N\in\mathbb{N}$, the subsets
\begin{align*}
\mathcal{C}_{udK,per(N)}^{(L)}&:=\mathcal{C}_{udK}^{(L)}\cap\left(S_{udK}^{(L)}\right)^{-1}\left(\mathcal{S}^{per(N)}\right),\\
\mathcal{C}_{dK,per(N)}^{(\delta)}&:=\mathcal{C}_{dK}^{(\delta)}\cap\left(S_{dK}^{(\delta)}\right)^{-1}\left(\mathcal{S}^{per(N)}\right),\\
\mathcal{C}_{udT,per(N)}&:=\mathcal{C}_{udT}\cap\left(S_{udT}\right)^{-1}\left(\mathcal{S}^{per(2N)}\right),\\
\mathcal{C}_{dT,per(N)}&:=\mathcal{C}_{dT}\cap\left(S_{dT}\right)^{-1}\left(\mathcal{S}^{per(2N)}\right),
\end{align*}
are closed for \eqref{UDKDV}, \eqref{DKDV}, \eqref{UDTODA}, \eqref{DTODA}, respectively.
\end{thm}

\begin{rem} The subsets of the configuration spaces that appear in the statement of Theorem \ref{perthm} are in one-to-one correspondence with the sets
\[\left\{\left(\eta_1,\dots,\eta_N\right)\in\mathbb{R}^{N}:\:\sum_{m=1}^N\eta_m<\frac{NL}{2}\right\},\]
\[\left\{\left(\omega_1,\dots,\omega_N\right)\in(0,\infty)^{N}:\:\sum_{m=1}^N\log\omega_m<\frac{-N\log\delta}{2}\right\},\]
\[\left\{\left(Q_1,E_1,\dots,Q_N,E_N\right)\in\mathbb{R}^{2N}:\:\sum_{m=1}^NQ_m<\sum_{m=1}^NE_m\right\},\]
\[\left\{\left(I_1,J_1,\dots,I_N,J_N\right)\in(0,\infty)^{2N}:\:\sum_{m=1}^N\log I_m>\sum_{m=1}^N\log J_m\right\},\]
respectively. Indeed, for periodic configurations, to check that the `density conditions' specified in the definitions of $\mathcal{C}_{udK}^{(L)}$, $\mathcal{C}_{dK}^{(\delta)}$, $\mathcal{C}_{udT}$ and $\mathcal{C}_{dT}$ are satisfied, it is enough to confirm that the analogous conditions hold within a single period. The images of the \eqref{UDKDV}, \eqref{DKDV}, \eqref{UDTODA} and \eqref{DTODA} dynamics on the above finite dimensional spaces give the periodic versions of the four systems. In particular, Theorems \ref{udkdvthm}, \ref{dkdvthm}, \ref{udtodathm} and \ref{dtodathm} give the existence and uniqueness of solutions to the corresponding periodic initial value problems. Thus our results include the periodic solutions discussed in \cite{TT, TTeng}, for example. Moreover, by combining Theorem \ref{perthm} with Theorem \ref{specialthm}, we see that our framework includes the periodic version of the original BBS. In this case, the relevant density condition ensures that the carrier empties itself at least once within a given period.
\end{rem}

\subsubsection{Finite configuration solutions to \eqref{UDKDV}}\label{finconfsec}

Our result concerning finite configurations in the ultra-discrete KdV system is as follows.

\begin{thm}\label{finthm} Let $a\in\mathbb{R}$ be such that $L-2a>0$. It is then the case that
\begin{equation}\nonumber
\mathcal{C}_{udK,fin(a)}^{(L)}:=\left\{\eta_n=a\mbox{ for all but finitely many }n\in\mathbb{Z}\right\}
\end{equation}
is a subset of $\mathcal{C}_{udK}^{(L)}$, and is closed for \eqref{UDKDV}. Moreover, if the initial configuration is taken in $\mathcal{C}_{udK,fin(a)}^{(L)}$, then the carrier of Theorem \ref{udkdvthm} satisfies: for each $t\in\mathbb{Z}$, $U_n^t=a$ eventually as $n\rightarrow\pm\infty$.
\end{thm}

\begin{rem} The particular case with $a=0$ is the finite configuration case, and has been widely researched. In such work, the property that $U_n^t=0$ eventually as $n\rightarrow\infty$ (or indeed that $U_n^t=0$ for $n<\inf\{m:\:\eta^t_m\neq0\}$) is typically taken as an assumption in order to find solutions to \eqref{UDKDV} recursively. For example, this was the case in the original BBS paper of \cite{takahashi1990}, which has dynamics given by \eqref{UDKDV} with $L=1$.
\end{rem}

\begin{rem}\label{fintodarem}
Although we do not pursue it in detail here, by making a relatively minor extension of the framework of this article, it would also be possible to include finite/one-sided infinite configurations for \eqref{UDTODA} or \eqref{DTODA}. Indeed, for \eqref{UDTODA}, one could consider
\begin{align*}
\lefteqn{\tilde{\mathcal{C}}_{udT}:=}\\
&\left\{(Q,E)\in(\mathbb{R}\times(\mathbb{R}\cup\{\infty\})^\mathbb{Z}\::\:
 \begin{array}{l}
   \lim_{n \to \infty}\frac{\sum_{m=1}^{n}(Q_m-E_m)}{n} =  \lim_{n \to \infty}\frac{\sum_{m=1}^{n}(Q_m-E_m)+Q_{n+1}}{n} <0,\\
 \lim_{n \to -\infty}\frac{\sum_{m=1}^{n}(Q_m-E_m)}{n}=  \lim_{n \to -\infty}\frac{\sum_{m=1}^{n}(Q_m-E_m)+E_{n}}{n} <0
 \end{array}\right\},
 \end{align*}
where the limits are: defined to be equal to $-\infty$ if $E_m=\infty$ infinitely often in the relevant direction; obtained by replacing $\lim_{n\rightarrow\infty}n^{-1}\sum_{m=1}^n$ by $\lim_{n\rightarrow\infty}n^{-1}\sum_{m=n_0+1}^n$ if $n_0:=\sup\{m:\:E_m=\infty\}<\infty$; and obtained by replacing $\lim_{n\rightarrow-\infty}n^{-1}\sum_{m=1}^n$ by $\lim_{n\rightarrow-\infty}n^{-1}\sum_{m=n_0}^n$ if $n_0:=\inf\{m:\:E_m=\infty\}>-\infty$. Via suitable modifications of our arguments, it is then possible to verify that for any $(Q,E) \in \tilde{\mathcal{C}}_{udT}$, there exists a unique solution $(Q^t_n,E^t_n, U^t_n)_{n,t \in \Z}$ to \eqref{UDTODA} such that $(Q^t,E^t) \in \tilde{\mathcal{C}}_{udT}$ and $U^t_n \in \R$ for all $n,t$. It is also possible to establish natural forward and backward versions of this statement, as well as showing that the system is closed when the variables are restricted to taking integer values (when they are not infinite). The finite ultra-discrete Toda lattice (or the Toda description of the finite BBS) is obtained in this formulation with $Q_n>0$ (or $Q_n \in \mathbb{N}$) for $n=1,\dots,N$, $E_n>0$ (or $E_n \in \mathbb{N}$) for $n=1,\dots,N-1$, and $Q_n =0$ and $E_n = \infty$ otherwise.

For \eqref{DTODA}, the extension would be to consider
\begin{align*}
\lefteqn{\tilde{\mathcal{C}}_{dT}:=}\\
&\left\{(I,J)\in((0,\infty)\times[0,\infty))^\mathbb{Z}:
 \begin{array}{l}
   \lim\limits_{n \to \infty}\frac{\sum_{m=1}^{n}(\log J_m-\log I_m)}{n} =  \lim\limits_{n \to \infty}\frac{\sum_{m=1}^{n}(\log J_m-\log I_m)-\log I_{n+1}}{n}<0,\\
\lim\limits_{n \to-\infty}\frac{\sum_{m=1}^{n}(\log J_m-\log I_m)}{n}=  \lim\limits_{n \to -\infty}\frac{\sum_{m=1}^{n}(\log J_m-\log I_m)-\log J_{n}}{n} <0
 \end{array}\right\},
 \end{align*}
i.e.\ allow 0 as a value for the $J$ variables, where we again have to modify the definitions of the relevant limits. For any $(I,J) \in \tilde{\mathcal{C}}_{dT}$, there exists a unique solution $(I^t_n,J^t_n)_{n,t \in \Z}$ to \eqref{DTODA} satisfying $(I^t,J^t) \in \tilde{\mathcal{C}}_{dT}$. In particular, the finite discrete Toda lattice is obtained in this formulation with $I_n>0$ for $n=1,\dots,N$, $J_n>0$ for $n=1,\dots,N-1$, and $I_n=1$ and $J_n = 0$ otherwise \cite[Section 3.2]{TTeng}.
\end{rem}

\subsubsection{Product convergent solutions to \eqref{DKDV}} In the case $\delta\in(0,1)$, for the special class $\mathcal{C}_{dK,conv}$ of initial configurations for \eqref{DKDV} for which $\prod_{m=-\infty}^0\omega_m$ converges, we are able to deduce an explicit expression for the carrier. We also show that the class is closed for the forward dynamics (i.e.\ for such $\omega$, it holds that $\omega\in\mathcal{C}_{dK,-}^{(\delta)}$, and moreover $\omega^t\in\mathcal{C}_{dK,conv}$ for all $t\in\mathbb{Z}_+$). We note that the class $\mathcal{C}_{dK,conv}$ clearly contains configurations that satisfy $\omega_m=1$ eventually as $m\rightarrow-\infty$, which can be seen as the \eqref{DKDV} analogue of the finite configurations for \eqref{UDKDV} of Subsection \ref{finconfsec}. Moreover, it is possible to check that the class $\mathcal{C}_{dK,conv}$ contains certain known soliton solutions for \eqref{DKDV}.

\begin{thm}\label{yoyoyo} Let $\delta\in(0,1)$. The set
\begin{equation}\nonumber
\mathcal{C}_{dK,conv}:=\left\{\omega \in (0,\infty)^\mathbb{Z}:\:\prod_{m=-\infty}^0\omega_m\mbox{ converges in }(0,\infty)\right\}
\end{equation}
is a subset of $\mathcal{C}_{dK,-}^{(\delta)}$, and is closed for the forward version of \eqref{DKDV}. Moreover, if the initial configuration is taken in $\mathcal{C}_{dK,conv}$, then the carrier of Theorem \ref{fthm} is given by
\begin{equation}\label{untexp}
U_n^t=\sum_{m=-\infty}^n\delta^{n-m}\omega^t_{m}\prod_{k=m+1}^n(\omega^t_k)^2,\qquad\forall n\in\mathbb{Z},\:t\in\mathbb{Z}_+,
\end{equation}
where $\prod_{k=n+1}^n \omega_k^2 := 1$, and, for each $t\in\mathbb{Z}_+$, $U^t_n$ converges to $(1-\delta)^{-1}$ as $n\rightarrow-\infty$.
\end{thm}

\subsection{Ultra-discretization}\label{udsec}

As shown on Figure \ref{dis}, it is known that \eqref{UDKDV} and \eqref{UDTODA} can be obtained from \eqref{DKDV} and \eqref{DTODA}, respectively, by a procedure known as ultra-discretization, see \cite{TTMS} and \cite[Sections 3.1 and 3.2]{TTeng}. Roughly speaking, if one has a sequence of solutions to the discrete systems parameterized by $\varepsilon>0$, then one can hope to find solutions to the ultra-discrete systems by considering the limits of $\varepsilon \log x(\varepsilon)$ as $\varepsilon\downarrow0$, where $x(\varepsilon)$ represents one of the variables in the relevant discrete system. This procedure was discussed in the generality of the present article in Remark \ref{udreme} and, in the following two results, we make precise how it applies to the solutions that we have presented for the four discrete integrable systems of principal interest here.

\begin{thm}\label{rrr} Suppose that for each $\varepsilon\in(0,1)$, $(\omega(\epsilon)_n^t,U(\epsilon)_n^t)_{n,t\in\mathbb{Z}}$ is a solution to \eqref{DKDV} with model parameter $\delta(\varepsilon)\in(0,\infty)$, and the following limits exist: for $n,t\in\mathbb{Z}$,
\begin{align*}
\eta^t_n := \lim_{\epsilon \downarrow 0} \epsilon \log \omega(\epsilon)_n^t, \qquad U^t_n := \lim_{\epsilon \downarrow 0} \epsilon \log U(\epsilon)_n^t, \qquad
L:= \lim_{\epsilon \downarrow 0}-\epsilon \log \delta(\varepsilon).
\end{align*}
It is then the case that $(\eta^t_n,U^t_n)_{n,t\in\mathbb{Z}}$ satisfies \eqref{UDKDV} with parameter $L$.
\end{thm}
\begin{proof} For any finite or countable sequence $(x_i)$ such that $\sum_{i}e^{x_i}$ converges, we have
\begin{equation}\label{logsumexp}
\lim_{\varepsilon\downarrow 0}\varepsilon\log\sum_{i}e^{\varepsilon^{-1}x_i}=\sup_{i}x_i.
\end{equation}
From this, we readily obtain that if $a=\lim_{\epsilon \downarrow 0}\epsilon \log a(\epsilon)$, $b=\lim_{\epsilon \downarrow 0}\epsilon \log b(\epsilon)$ and also $L=\lim_{\epsilon \downarrow 0}-\epsilon \log \delta(\epsilon)$, then
\[\lim_{\epsilon \downarrow 0}\epsilon \log \left(F_{dK}^{(\delta(\varepsilon))}\right)^{(i)}(a(\epsilon),b(\epsilon)) =\left(F_{udK}^{(L)}\right)^{(i)}(a,b)\]
for $i=1,2$, where we again write $F_{udK}^{(L)}(a,b)=((F_{udK}^{(L)})^{(1)}(a,b),(F_{udK}^{(L)})^{(2)}(a,b))$, and use similar notation for the components of $F_{dK}^{(\delta)}$. The result follows.
\end{proof}

\begin{rem}\label{r37} The preceding conclusion gives an ultra-discretization result for path encodings. In particular, suppose $S(\varepsilon)^t$ is the \eqref{DKDV} path encoding for $\omega(\varepsilon)^t$, as given by \eqref{2pe}, then
\begin{equation}\label{der}
S^t_n:=\lim_{\varepsilon\downarrow0}\varepsilon S(\varepsilon)^t_n,\qquad \forall n\in\mathbb{Z},
\end{equation}
gives the \eqref{UDKDV} path encoding of $\eta^t$, as defined at \eqref{1pe}. Moreover, it is an elementary exercise to relate the Pitman-type transformations of the discrete and ultra-discrete systems. Indeed, for any $S\in\mathcal{S}$ such that $\sum_{m\leq n}\exp(\frac{S_m+S_{m-1}}{2})$ converges, which includes any $S\in\mathcal{S}^{lin}$, we have from \eqref{logsumexp} that
\[T^\vee(S)_n=\lim_{\varepsilon\downarrow0}\varepsilon T^{\sum}\left(\varepsilon^{-1}S\right)_n,\qquad \forall n\in\mathbb{Z}.\]
\end{rem}

\begin{thm}\label{t319} Suppose that for each $\varepsilon\in(0,1)$, $(I(\epsilon)_n^t,J(\epsilon)_n^t,U(\epsilon)_n^t)_{n,t\in\mathbb{Z}}$ is a solution \eqref{DTODA} and the following limits exist: for $n,t\in\mathbb{Z}$,
\begin{align*}
Q^t_n := \lim_{\epsilon \downarrow 0} -\epsilon \log I(\epsilon)_n^t, \qquad
E^t_n:= \lim_{\epsilon \downarrow 0}-\epsilon \log J(\epsilon)_n^t, \qquad
U^t_n:= \lim_{\epsilon \downarrow 0}-\epsilon \log U(\epsilon)_n^t.
\end{align*}
Then, $(Q^t,E^t,U^t)$ satisfies \eqref{UDTODA}.
\end{thm}
\begin{proof} The proof is essentially the same as that of Theorem \ref{rrr}, namely we simply observe that if $a=\lim_{\epsilon \downarrow 0}-\epsilon \log a(\epsilon)$, $b=\lim_{\epsilon \downarrow 0}-\epsilon \log b(\epsilon)$ and $c=\lim_{\epsilon \downarrow 0}-\epsilon \log c(\epsilon)$, then
\[\lim_{\epsilon \downarrow 0}-\epsilon \log F_{dT}^{(i)}(a(\epsilon),b(\epsilon),c(\epsilon)) =F_{udT}^{(i)}(a,b,c)\]
for $i=1,2,3$, where we use the superscript $(i)$ to denote the index of the component of the relevant map.
\end{proof}

\begin{rem}\label{r322} Similarly to Remark \ref{r37}, we have a corresponding ultra-discretization result for path encodings. In particular, suppose $S(\varepsilon)^t$ is the \eqref{DTODA} path encoding for $(I(\varepsilon)^t,J(\varepsilon)^t)$, as given by \eqref{4pe}, then \eqref{der} gives the \eqref{UDTODA} path encoding of $(Q^t,E^t)$, as defined at \eqref{3pe}. Moreover, for any $S\in\mathcal{S}$ such that $\sum_{m\leq \frac{n-1}{2}}\exp(S_{2m})$ converges, which includes any $S\in\mathcal{S}^{lin}$, we have from \eqref{logsumexp} that
\[T^{\vee^*}(S)_n=\lim_{\varepsilon\downarrow0}\varepsilon T^{\sum^*}\left(\varepsilon^{-1}S\right)_n,\qquad \forall n\in\mathbb{Z}.\]
\end{rem}

\subsection{Proofs of modifications of main results}\label{pmodsec}

In this subsection, we prove Theorems \ref{fthm}, \ref{ttt1}, \ref{ttt2}, \ref{specialthm}, \ref{perthm}, \ref{finthm}, \ref{yoyoyo} and Corollary \ref{bcor}.

\begin{proof}[Proof of Theorem \ref{fthm}] On replacing the claims of Lemmas \ref{ll1}(d), \ref{ll2}(d), \ref{ll3}(d) and \ref{ll4}(d) with the simple observations that
\[\mathcal{A}^{udK}\left(\mathcal{C}_{udK,-}^{(L)}\right)=\mathcal{A}^{dK}\left(\mathcal{C}_{dK,-}^{(\delta)}\right)=\mathcal{A}^{udT}\left(\mathcal{C}_{udT,-}\right)=\mathcal{A}^{dT}\left(\mathcal{C}_{dT,-}\right)=\mathcal{X}_-^{lin},\]
Theorem \ref{fthm} follows by applying exactly the same argument as was used to prove Theorems \ref{udkdvthm}, \ref{dkdvthm}, \ref{udtodathm} and \ref{dtodathm}.
\end{proof}

\begin{proof}[Proof of Corollary \ref{bcor}] Given the involutive property of the locally-defined dynamics for the four systems, and the fact that
\[\mathcal{A}^{udK}\left(\mathcal{C}_{udK,+}^{(L)}\right)=\mathcal{A}^{dK}\left(\mathcal{C}_{dK,+}^{(\delta)}\right)=\mathcal{A}^{udT}\left(\mathcal{C}_{udT,+}\right)=\mathcal{A}^{dT}\left(\mathcal{C}_{dT,+}\right)=\mathcal{X}_+^{lin}=R^\mathcal{X}\left(\mathcal{X}_-^{lin}\right),\]
one can adopt the same approach as in the forward case to deduce the result, but using Corollary \ref{bivpun} in place of Theorem \ref{thm:uni-ex-f-gp} at the relevant point of the argument.
\end{proof}

\begin{proof}[Proof of Theorem \ref{ttt1}] Let $\omega\in\mathcal{C}_{dK}^{(\delta)}$, and $(\omega^t,U^t)_{t\in\mathbb{Z}}$ be as at \eqref{omegatdef}. Since we have from \eqref{DKDV} that $(U^t_{n-1})^{-1}=\frac{1}{\omega_{n}^{t+1}}-\delta \omega_n^t$, it must be the case that $\omega^t$ and $\omega^{t+1}$ satisfy
\[\frac{1}{\omega_{n+1}^{t+1}}-\delta \omega_{n+1}^t=\left(\frac{1}{\omega_{n}^{t+1}}-\delta \omega_n^t\right)\frac{\omega_n^{t+1}}{\omega_n^t}=\frac{1}{\omega^t_n}-\delta \omega^{t+1}_n,\qquad\forall n\in\mathbb{Z},\]
or equivalently
\[\frac{1}{\omega^{t+1}_{n+1}}-\frac{1}{\omega^t_n}=\delta\left(\omega^t_{n+1}-\omega^{t+1}_n\right),\qquad\forall n\in\mathbb{Z},\]
and thus we have established part (a).

For part (b), suppose that a pair $\omega, \tilde{\omega}\in(0,\infty)^\mathbb{Z}$ satisfies
\[\frac{1}{\tilde{\omega}_{n+1}}-\frac{1}{\omega_n}=\delta(\omega_{n+1}-\tilde{\omega}_n),\qquad\forall n\in\mathbb{Z}.\]
It is an elementary exercise to check that it must consequently hold that either $\frac{1}{\tilde{\omega}_{n}}-\delta\omega_n > 0$ for all $n\in\mathbb{Z}$, or $\frac{1}{\tilde{\omega}_{n}}-\delta\omega_n = 0$ for all $n\in\mathbb{Z}$, or $\frac{1}{\tilde{\omega}_{n}}-\delta\omega_n < 0$ for all $n\in\mathbb{Z}$. If the second or third case occur, then $\delta \omega_n \tilde{\omega}_n \ge 1$ for all $n\in\mathbb{Z}$. Hence $\log \delta +\frac{1}{n}\sum_{m=1}^n \log \omega_m + \frac{1}{n}\sum_{m=1}^n  \log \tilde{\omega}_m \ge 0$ for all $n\in\mathbb{Z}$, and so either $\omega \notin \mathcal{C}_{dK}^{(\delta)}$ or $\tilde{\omega}\notin \mathcal{C}_{dK}^{(\delta)}$. Therefore, if $\omega^t$ is a solution to \eqref{DKomega1} satisfying $\omega^t \in \mathcal{C}_{dK}^{(\delta)}$ for all $t\in\mathbb{Z}$, then $\frac{1}{\omega^{t+1}_{n}}-\delta\omega^t_n > 0$ for all $n,t\in\mathbb{Z}$. By defining $U^t_{n-1}:=(\frac{1}{\omega^{t+1}_{n}}-\delta\omega^t_n)^{-1}$, we obtain a solution to \eqref{DKDV}, and thus the uniqueness of the solution to \eqref{DKomega1} follows from Theorem \ref{dkdvthm}.

Let $\hat{\omega}\in(0,\infty)^\mathbb{Z}$ satisfy \eqref{ahjk}, and define ${\omega}\in(0,\infty)^\mathbb{Z}$ by setting $\omega_n:=\frac{1}{\delta\hat{\omega}_n}$. It is then readily checked that $\omega\in\mathcal{C}_{dK}^{\delta}$, and so we can apply part (b) to deduce the uniqueness of the solution $(\omega^t)_{t\in\mathbb{Z}}$ to \eqref{DKomega1} with $\omega^t\in \mathcal{C}_{dK}^{\delta}$ for all $t\in\mathbb{Z}$. Setting $\hat{\omega}_n^t:=\frac{1}{\delta\omega_n^{-t}}$, we find a solution to \eqref{DKomega1} with initial condition $\hat{\omega}$ and such that $\hat{\omega}^t$ satisfies \eqref{ahjk} for all $t\in\mathbb{Z}$. Moreover, the uniqueness also follows, since if we have another solution $(\hat{\hat{\omega}}^t)_{t\in\mathbb{Z}}$ for which $\hat{\hat{\omega}}^t$ satisfies \eqref{ahjk} for all $t\in\mathbb{Z}$, then defining $(\tilde{\omega}^t)_{t\in\mathbb{Z}}$ via the relation $\hat{\hat{\omega}}_n^t=\frac{1}{\delta\tilde{\omega}_n^{-t}}$ implies that $(\tilde{\omega}^t)_{t\in\mathbb{Z}}$ solves \eqref{DKDV} and $\tilde{\omega}^t\in\mathcal{C}_{dK}^{(\delta)}$ for all $t\in\mathbb{Z}$. Hence $(\tilde{\omega}^t)_{t\in\mathbb{Z}}$ must be equal to $({\omega}^t)_{t\in\mathbb{Z}}$, and it follows that $(\hat{\hat{\omega}}^t)_{t\in\mathbb{Z}}$ is actually equal to $({\hat{\omega}}^t)_{t\in\mathbb{Z}}$. This completes the proof of part (c), and also yields part (d).
\end{proof}

\begin{proof}[Proof of Theorem \ref{ttt2}] Let $(I,J)\in\mathcal{C}_{dT}$, and $(I^t,J^t,U^t)$ the corresponding solution to \eqref{DTODA}, as given by Theorem \ref{dtodathm}. The first and third conditions of \eqref{DToda2} are then automatically satisfied. Moreover,
\[I_n^{t+1}=J_n^t+U_n^t=J_n^t+\frac{I_{n}^tU_{n-1}^t}{I_{n-1}^{t+1}}=J_n^t+\frac{I_{n}^t\left(I^{t+1}_{n-1}-J_{n-1}^t\right)}{I_{n-1}^{t+1}}=I_n^t+J_n^t-J_{n-1}^{t+1},\]
which is the second condition of \eqref{DToda2}, which confirms part (a).

For part (b), suppose a pair $(I,J), (\tilde{I},\tilde{J})\in((0,\infty)^2)^\mathbb{Z}$ satisfies
\begin{align*}
\begin{cases}
\tilde{I}_n &=I_n+J_n-\tilde{J}_{n-1}, \\
\tilde{J}_n &=\frac{I_{n+1}J_n}{\tilde{I}_n},
\end{cases}
\end{align*}
for all $n\in\mathbb{Z}$. It is then the case that
\[\tilde{I}_n - J_n= I_n-\tilde{J}_{n-1}=\frac{I_n}{\tilde{I}_{n-1}}\left(\tilde{I}_{n-1} - \frac{\tilde{I}_{n-1}}{I_n}\tilde{J}_{n-1}\right)= \frac{I_n}{\tilde{I}_{n-1}}\left(\tilde{I}_{n-1} - J_{n-1}\right)\]
for all $n\in\mathbb{Z}$, and so either $\tilde{I}_n - J_n = I_n-\tilde{J}_{n-1} > 0$ for all $n\in\mathbb{Z}$, or $\tilde{I}_n - J_n = I_n-\tilde{J}_{n-1}= 0$ for all $n\in\mathbb{Z}$, or $\tilde{I}_n - J_n = I_n-\tilde{J}_{n-1}< 0$ for all $n\in\mathbb{Z}$. If the second or third case occurs, then $I_n\tilde{I}_n\le J_n\tilde{J}_{n-1}$ for all $n\in\mathbb{Z}$. Hence
\[\frac{1}{n}\sum_{m=1}^n \log I_m + \frac{1}{n}\sum_{m=1}^n  \log \tilde{I}_m \le \frac{1}{n}\sum_{m=1}^n\log J_m + \frac{1}{n}\sum_{m=1}^n  \log \tilde{J}_{m-1}\]
for all $n\in\mathbb{Z}$, and so either $(I,J) \notin \mathcal{C}_{dT}$ or $(\tilde{I},\tilde{J}) \notin \mathcal{C}_{dT}$. Therefore, if $(I^t,J^t)$ is a solution to \eqref{DToda2} satisfying $(I^t,J^t) \in \mathcal{C}_{dT}$ for all $t\in\mathbb{Z}$, then, ${I}^{t+1}_n - J_n = I_n-{J}^{t+1}_{n-1} > 0$ for all $n,t\in\mathbb{Z}$. By defining $U^t_n:=I^{t+1}_n - J^t_n$, we obtain a solution to \eqref{DTODA}, and thus the uniqueness of the solution to \eqref{DToda2} follows from Theorem \ref{dtodathm}.

Finally, applying the given change of variables, parts (c) and (d) are straightforward to check (cf.\ the proof of Theorem \ref{ttt1}(c,d)).
\end{proof}

\begin{proof}[Proof of Theorem \ref{specialthm}] Given the self-reverse property of the \eqref{UDKDV} and \eqref{UDTODA} systems, in each case it will suffice to check that if the initial configuration is in one of the described subsets, then so is the configuration after one step of the dynamics (as per Theorem \ref{udkdvthm} or \ref{udtodathm}), and also the associated carrier is in the subset indicated in the statement of the result.

First, let $\eta^0=\eta\in[L_0,L-L_0]^\mathbb{Z}\cap\mathcal{C}_{udK}^{(L)}$ for some $L_0<L/2$. Writing $S=S_{udK}^{(L)}(\eta^0)$, we then have that
\[U^0_n= M^{\vee}(S)_n- S_n +\frac{L}{2}\geq\frac{S_n+S_{n-1}}{2}-S_n+\frac{L}{2}=\eta^0_n\geq L_0,\qquad \forall n\in\mathbb{Z}.\]
Moreover, it follows that
\[L_0\leq \eta^1_n=\min \left\{L-\eta^0_n, U_n^0\right\}\leq L-L_0,\qquad \forall n\in\mathbb{Z},\]
i.e.\ $\eta^1\in[L_0,L-L_0]^\mathbb{Z}\cap\mathcal{C}_{udK}^{(L)}$, which completes the proof of (a)(i). Second, let $\eta^0=\eta\in(\alpha\mathbb{Z})^\mathbb{Z}\cap\mathcal{C}_{udK}^{(L)}$ for some $\alpha>0$ with $L\in\alpha\mathbb{Z}$. Again writing $S=S_{udK}^{(L)}(\eta^0)$, we then have that $L-2\eta^0_n\in \alpha\mathbb{Z}$ for all $n\in\mathbb{Z}$, and it follows that
\begin{eqnarray*}
U^0_n&=& M^{\vee}(S)_n- S_n +\frac{L}{2}\\
&=&\sup_{m\leq n}\left(S_{m-1}+\frac{L-2\eta^0_m}{2}\right)-S_n+\frac{L}{2}\\
&=&\sup_{m\leq n}\left(S_{m-1}-S_n+L-\eta^0_m\right)\\
&=&\sup_{m\leq n}\left(-\sum_{k=m}^n(L-2\eta^0_k)+L-\eta_m\right)\\
\end{eqnarray*}
is an element of $\alpha\mathbb{Z}$ for all $n\in\mathbb{Z}$. From this, we readily obtain that $\eta^1_n=\min \left\{L-\eta^0_n, U_n^0\right\}\in\alpha\mathbb{Z}$ for all $n\in\mathbb{Z}$, as desired for (a)(ii). Given (a)(i) and (a)(ii), we can simply take the intersection of the sets (with $L_0=0$ and $\alpha=1$) to obtain (a)(iii).

For part (b), first let $(Q^0,E^0)=(Q,E)\in((0,\infty)^2)^\mathbb{Z}$. Then, from \eqref{UDTODA}, we have that
\[E^1_{n}=Q_{n+1}^0+E^0_n-Q_{n}^1\geq Q_{n+1}^0+\min\{U_n^0,E^0_n\}-Q_{n}^1=Q_{n+1}^0>0,\]
and also
\[U^0_{n+1}=U^0_n+Q_{n+1}^0-Q_{n}^1\geq\min\{U_n^0,E^0_n\}+ Q_{n+1}^0-Q_{n}^1=Q_{n+1}^0>0,\]
for all $n\in\mathbb{Z}$. Hence, we also find that $Q_n^1=\min\{U_n^0,E^0_n\}\geq \min\{Q_n^0,E^0_n\}>0$, which establishes (b)(i). Next, let $(Q^0,E^0)=(Q,E)\in((\alpha\mathbb{Z})^2)^\mathbb{Z}$, and $S$ be the associated path encoding. We then have that
\[U_{n}^0=M^{\vee^*}(S)_{2n-1}-S_{2n-1}=\sup_{m\leq n-1}S_{2m}\in\alpha\mathbb{Z},\qquad\forall n\in\mathbb{Z}.\]
From this, it is readily seen that also $(Q^1,E^1)=(Q,E)\in((\alpha\mathbb{Z})^2)^\mathbb{Z}$, as required for (b)(ii). Also, similarly to the previous paragraph, given (b)(i) and (b)(ii), we can simply take the intersection of the sets (with $\alpha=1$) to obtain (b)(iii).
\end{proof}

\begin{rem}
By considering appropriate choices of $\tilde{\mathcal{X}}$ and $\tilde{\mathcal{U}}$ in Subsection \ref{s43}, Theorem \ref{specialthm} can also be proved directly from the results of Section \ref{general}.
\end{rem}

\begin{proof}[Proof of Theorem \ref{perthm}] The result readily follows from the uniqueness of solutions to the relevant equations, as given by Theorems \ref{udkdvthm}, \ref{dkdvthm}, \ref{udtodathm} and \ref{dtodathm}.
\end{proof}

\begin{proof}[Proof of Theorem \ref{finthm}] Let $\eta\in\mathcal{C}_{udK,fin(a)}^{(L)}$, where $L-2a>0$, and write $S=S_{udK}^{(L)}(\eta)$. By definition, there exists an $N_0\in\mathbb{N}$ such that $S_n-S_{n-1}=L-2a$ whenever $|n|\geq N_0$. It follows that there exists an $N_1\in\mathbb{N}$ such that $M^{\vee}(S)_n=S_n-\frac{L-2a}{2}$ whenever $|n|\geq N_1$, and consequently $T^{\vee}(S)_n=S_n-(L-2a)$ whenever $|n|\geq N_1$. Hence if $\eta^1$ is the configuration after one time step of the \eqref{UDKDV} dynamics started from initial condition $\eta$, as given by Theorem \ref{udkdvthm}, then $\eta^1\in \mathcal{C}_{udK,fin(a)}^{(L)}$. Iterating this argument and applying the self-reverse property of the system yields moreover that $\eta^t\in \mathcal{C}_{udK,fin(a)}^{(L)}$ for all $t\in\mathbb{Z}$, which confirms that $\mathcal{C}_{udK,fin(a)}^{(L)}$ is closed for \eqref{UDKDV}. Moreover, the argument implies that: for each $t\in\mathbb{Z}$,
\[U_n^t=\frac{S_n^{t+1}-S_n^t+L}{2}=\frac{-(L-2a)+L}{2}=a\]
eventually as $n\rightarrow\pm\infty$, where $S^t,S^{t+1}\in\mathcal{S}$ are defined as in the statement of Theorem \ref{udkdvthm}.
\end{proof}

\begin{proof}[Proof of Theorem \ref{yoyoyo}] If $\omega\in \mathcal{C}_{dK,conv}$, then
\[\lim_{n\rightarrow-\infty}\frac{\sum_{m=1}^n\log \omega_m}{n}=0<\frac{-\log\delta}{2},\]
and so $\mathcal{C}_{dK,conv}\subseteq\mathcal{C}_{dK,-}^{(\delta)}$. Suppose next that the \eqref{DKDV} dynamics (as per Theorem \ref{fthm}) are started from $\omega$. The carrier $(U_n)_{n\in\mathbb{Z}}$ associated with $\omega$ then satisfies
\[\log U_n+\frac{\log\delta}{2}=\log\sum_{m=-\infty}^n\exp\left(\frac{S_m-S_{m-1}}{2}\right)-S_n,\]
where $S$ is the path encoding of $\omega$, as given by \eqref{2pe}. We rewrite this as follows:
\begin{eqnarray*}
\log U_n+\frac{\log\delta}{2}&=&\log\sum_{m=-\infty}^n\exp\left(\frac{S_m-S_{m-1}-2S_n}{2}\right)\\
&=&\log\sum_{m=-\infty}^n\exp\left(\sum_{k=m+1}^n(\log \delta+2\log\omega_k)+\frac{\log\delta+2\log\omega_m}{2}\right)\\
&=&\log\sum_{m=-\infty}^n\delta^{n-m+1/2}\omega_m\prod_{k=m+1}^n\omega_k^2,
\end{eqnarray*}
which yields the expression at \eqref{untexp}. Moreover, we have that
\[S_{n}=-n\log\delta+2\sum_{m={n+1}}^0\log\omega_m=-n\log\delta+C+\varepsilon_n,\]
where $C\in(0,\infty)$ is a constant, and $\varepsilon_n\rightarrow0$ as $n\rightarrow-\infty$. Hence
\[U_n=\delta^{-1/2}\sum_{m=-\infty}^n\exp\left(\frac{S_m-S_{m-1}-2S_n}{2}\right)=\sum_{m=-\infty}^n\delta^{n-m+\frac{\varepsilon_{m}+\varepsilon_{m-1}-2\varepsilon_n}{2}}\rightarrow\frac{1}{1-\delta},\]
as $n\rightarrow-\infty$. Finally, if $\omega^1$ is the configuration after one time step, then \eqref{DKDV} implies $\omega^1_n=U_{n-1}\omega_n/U_n$, which together with the above limiting result yields
\[\prod_{m=-\infty}^0\omega_m^1=(1-\delta)U_{-1}\prod_{m=-\infty}^0\omega_m\in(0,\infty).\]
Thus $\omega^1\in\mathcal{C}_{dK,conv}$, and we can proceed by iteration to complete the proof of the theorem.
\end{proof}

\section*{Acknowledgments}

The research was supported by JSPS Grant-in-Aid for Scientific Research (B), 19H01792. The research of DC was also supported by JSPS Grant-in-Aid for Scientific Research (C), 19K03540, and by the Research Institute for Mathematical Sciences, an International Joint Usage/Research Center located in Kyoto University. This work was completed while MS was being kindly hosted by the Courant Institute, New York University.

\section*{Conflicts of interest}

The sources of financial support received by the authors for this research are listed in the acknowledgments section above.

\bibliographystyle{amsplain}
\bibliography{kdv01}

\providecommand{\bysame}{\leavevmode\hbox to3em{\hrulefill}\thinspace}
\providecommand{\MR}{\relax\ifhmode\unskip\space\fi MR }
\providecommand{\MRhref}[2]{%
  \href{http://www.ams.org/mathscinet-getitem?mr=#1}{#2}
}
\providecommand{\href}[2]{#2}
\begin{thebibliography}{10}

\bibitem{Bertoin}
J.~Bertoin, \emph{An extension of {P}itman's theorem for spectrally positive
  {L}\'{e}vy processes}, Ann. Probab. \textbf{20} (1992), no.~3, 1464--1483.

\bibitem{BBOC1}
P.~Biane, P.~Bougerol, and N.~O'Connell, \emph{Littelmann paths and {B}rownian
  paths}, Duke Math. J. \textbf{130} (2005), no.~1, 127--167.

\bibitem{BBOC2}
\bysame, \emph{Continuous crystal and {D}uistermaat-{H}eckman measure for
  {C}oxeter groups}, Adv. Math. \textbf{221} (2009), no.~5, 1522--1583.

\bibitem{BonaSmith}
J.~L. Bona and R.~Smith, \emph{The initial-value problem for the {K}orteweg-de
  {V}ries equation}, Philos. Trans. Roy. Soc. London Ser. A \textbf{278}
  (1975), no.~1287, 555--601.

\bibitem{Bourgain}
J.~Bourgain, \emph{Periodic {K}orteweg de {V}ries equation with measures as
  initial data}, Selecta Math. (N.S.) \textbf{3} (1997), no.~2, 115--159.

\bibitem{CKST}
D.~A. Croydon, T.~Kato, M.~Sasada, and S.~Tsujimoto, \emph{Dynamics of the
  box-ball system with random initial conditions via {P}itman's
  transformation}, to appear in Mem. Amer. Math. Soc., preprint appears at
  arXiv:1806.02147, 2018.

\bibitem{CSirf}
D.~A. Croydon and M.~Sasada, \emph{Detailed balance and invariant measures for
  discrete {K}d{V}- and {T}oda-type systems}, preprint appears at
  arXiv:2007.06203, 2020.

\bibitem{CS}
\bysame, \emph{Duality between box-ball systems of finite box and/or carrier
  capacity}, RIMS K\^{o}ky\^{u}roku Bessatsu \textbf{B79} (2020), 63--107.

\bibitem{CSmsj}
\bysame, \emph{Discrete integrable systems and {P}itman's transformation},
  Stochastic analysis, random fields and integrable probability---{F}ukuoka
  2019, Adv. Stud. Pure Math., vol.~87, Math. Soc. Japan, Tokyo, 2021,
  pp.~381--402.

\bibitem{CST}
D.~A. Croydon, M.~Sasada, and S.~Tsujimoto, \emph{Dynamics of the
  ultra-discrete {T}oda lattice via {P}itman's transformation}, RIMS
  K\^{o}ky\^{u}roku Bessatsu \textbf{B78} (2020), 235--250.

\bibitem{DMO}
M.~Draief, J.~Mairesse, and N.~O'Connell, \emph{Queues, stores, and tableaux},
  J. Appl. Probab. \textbf{42} (2005), no.~4, 1145--1167.

\bibitem{FNRW}
P.~A. Ferrari, C.~Nguyen, L.~T. Rolla, and M.~Wang, \emph{Soliton decomposition
  of the box-ball system}, Forum Math. Sigma \textbf{9} (2021), Paper No. e60,
  37.

\bibitem{HMOC}
B.~M. Hambly, J.~B. Martin, and N.~O'Connell, \emph{Pitman's {$2M-X$} theorem
  for skip-free random walks with {M}arkovian increments}, Electron. Comm.
  Probab. \textbf{6} (2001), 73--77.

\bibitem{HW}
J.~M. Harrison and R.~J. Williams, \emph{On the quasireversibility of a
  multiclass {B}rownian service station}, Ann. Probab. \textbf{18} (1990),
  no.~3, 1249--1268.

\bibitem{H:DE1}
R.~Hirota, \emph{Nonlinear partial difference equations. {I.} {A} difference
  analogue of the {Korteweg-de Vries} equation}, J. Phys. Soc. Japan
  \textbf{43} (1977), no.~4, 1424--1433.

\bibitem{H:DE2}
\bysame, \emph{Nonlinear partial difference equations. {II.} {D}iscrete-time
  {Toda} equation}, J. Phys. Soc. Japan \textbf{43} (1977), no.~6, 2074--2078.

\bibitem{HNew}
\bysame, \emph{New solutions to the ultradiscrete soliton equations}, Stud.
  Appl. Math. \textbf{122} (2009), no.~4, 361--376.

\bibitem{HTI}
R.~Hirota, S.~Tsujimoto, and T.~Imai, \emph{Difference scheme of soliton
  equations}, Future Directions of Nonlinear Dynamics in Physical and
  Biological Systems (P.~L. Christiansen, J.~C. Eilbeck, and R.~D. Parmentier,
  eds.), Plenum, 1993, pp.~7--15.

\bibitem{IKT}
R.~Inoue, A.~Kuniba, and T.~Takagi, \emph{Integrable structure of box-ball
  systems: crystal, {B}ethe ansatz, ultradiscretization and tropical geometry},
  J. Phys. A \textbf{45} (2012), no.~7, 073001, 64.

\bibitem{Jeulin}
T.~Jeulin, \emph{Un th\'{e}or\`eme de {J}. {W}. {P}itman}, S\'{e}minaire de
  {P}robabilit\'{e}s, {XIII} ({U}niv. {S}trasbourg, {S}trasbourg, 1977/78),
  Lecture Notes in Math., vol. 721, Springer, Berlin, 1979, With an appendix by
  M. Yor, pp.~521--532.

\bibitem{KMT}
M.~Kanki, J.~Mada, and T.~Tokihiro, \emph{Conserved quantities and generalized
  solutions of the ultradiscrete {K}d{V} equation}, J. Phys. A \textbf{44}
  (2011), no.~14, 145202, 13.

\bibitem{Kenig}
C.~E. Kenig, G.~Ponce, and L.~Vega, \emph{Well-posedness of the initial value
  problem for the {K}orteweg-de {V}ries equation}, J. Amer. Math. Soc.
  \textbf{4} (1991), no.~2, 323--347.

\bibitem{KMV}
R.~Killip, J.~Murphy, and M.~Visan, \emph{Invariance of white noise for {K}d{V}
  on the line}, Invent. Math. \textbf{222} (2020), no.~1, 203--282.

\bibitem{Kishi}
N.~Kishimoto, \emph{Well-posedness of the {C}auchy problem for the
  {K}orteweg-de {V}ries equation at the critical regularity}, Differential
  Integral Equations \textbf{22} (2009), no.~5-6, 447--464.

\bibitem{KdV}
D.~J. Korteweg and G.~de~Vries, \emph{On the change of form of long waves
  advancing in a rectangular canal, and on a new type of long stationary
  waves}, Philos. Mag. (5) \textbf{39} (1895), no.~240, 422--443.

\bibitem{Kotani}
S.~Kotani, \emph{Construction of {K}d{V} flow -- a unified approach}, preprint
  appears at arXiv:2107.05428, 2021.

\bibitem{Krichever}
I.~M. Krichever, \emph{Algebraic curves and nonlinear difference equations},
  Uspekhi Mat. Nauk \textbf{33} (1978), no.~4(202), 215--216.

\bibitem{KLO}
A.~Kuniba, H.~Lyu, and M.~Okado, \emph{Randomized box-ball systems, limit shape
  of rigged configurations and thermodynamic {B}ethe ansatz}, Nuclear Phys. B
  \textbf{937} (2018), 240--271.

\bibitem{LL}
O.~E. Lanford~III, J.~L. Lebowitz, and E.~H. Lieb, \emph{Time evolution of
  infinite anharmonic systems}, J. Statist. Phys. \textbf{16} (1977), no.~6,
  453--461.

\bibitem{LLP}
L.~Levine, H.~Lyu, and J.~Pike, \emph{Double jump phase transition in a soliton
  cellular automaton}, Int. Math. Res. Not. IMRN (2022), no.~1, 665--727.

\bibitem{MY0}
H.~Matsumoto and M.~Yor, \emph{Some changes of probabilities related to a
  geometric {B}rownian motion version of {P}itman's {$2M-X$} theorem},
  Electron. Comm. Probab. \textbf{4} (1999), 15--23.

\bibitem{MY}
\bysame, \emph{A version of {P}itman's {$2M-X$} theorem for geometric
  {B}rownian motions}, C. R. Acad. Sci. Paris S\'{e}r. I Math. \textbf{328}
  (1999), no.~11, 1067--1074.

\bibitem{MY1}
\bysame, \emph{An analogue of {P}itman's {$2M-X$} theorem for exponential
  {W}iener functionals. {I}. {A} time-inversion approach}, Nagoya Math. J.
  \textbf{159} (2000), 125--166.

\bibitem{MY2}
\bysame, \emph{An analogue of {P}itman's {$2M-X$} theorem for exponential
  {W}iener functionals. {II}. {T}he role of the generalized inverse {G}aussian
  laws}, Nagoya Math. J. \textbf{162} (2001), 65--86.

\bibitem{Moser}
J.~Moser, \emph{Finitely many mass points on the line under the influence of an
  exponential potential--an integrable system}, Dynamical systems, theory and
  applications ({R}encontres, {B}attelle {R}es. {I}nst., {S}eattle, {W}ash.,
  1974), Springer, Berlin, 1975, pp.~467--497. Lecture Notes in Phys., Vol. 38.

\bibitem{NTT}
A.~Nagai, D.~Takahashi, and T.~Tokihiro, \emph{Soliton cellular automaton,
  {Toda} molecule equation and sorting algorithm}, Phys. Lett. A \textbf{255}
  (1999), 265--271.

\bibitem{NTSeps}
A.~Nagai, T.~Tokihiro, and J.~Satsuma, \emph{The {T}oda molecule equation and
  the {$\epsilon$}-algorithm}, Math. Comp. \textbf{67} (1998), no.~224,
  1565--1575.

\bibitem{NTS}
\bysame, \emph{Ultra-discrete {T}oda molecule equation}, Phys. Lett. A
  \textbf{244} (1998), 383--388.

\bibitem{OC0}
N.~O'Connell, \emph{Random matrices, non-colliding processes and queues},
  S\'{e}minaire de {P}robabilit\'{e}s, {XXXVI}, Lecture Notes in Math., vol.
  1801, Springer, Berlin, 2003, pp.~165--182.

\bibitem{OC}
\bysame, \emph{Directed polymers and the quantum {T}oda lattice}, Ann. Probab.
  \textbf{40} (2012), no.~2, 437--458.

\bibitem{OCY2}
N.~O'Connell and M.~Yor, \emph{Brownian analogues of {B}urke's theorem},
  Stochastic Process. Appl. \textbf{96} (2001), no.~2, 285--304.

\bibitem{OCY}
\bysame, \emph{A representation for non-colliding random walks}, Electron.
  Comm. Probab. \textbf{7} (2002), 1--12.

\bibitem{Pitman}
J.~W. Pitman, \emph{One-dimensional {B}rownian motion and the three-dimensional
  {B}essel process}, Advances in Appl. Probability \textbf{7} (1975), no.~3,
  511--526.

\bibitem{Q}
J.~Quastel and B.~Valk\'{o}, \emph{Kd{V} preserves white noise}, Comm. Math.
  Phys. \textbf{277} (2008), no.~3, 707--714.

\bibitem{Rog}
L.~C.~G. Rogers, \emph{Characterizing all diffusions with the {$2M-X$}
  property}, Ann. Probab. \textbf{9} (1981), no.~4, 561--572.

\bibitem{RogPit}
L.~C.~G. Rogers and J.~W. Pitman, \emph{Markov functions}, Ann. Probab.
  \textbf{9} (1981), no.~4, 573--582.

\bibitem{Saisho}
Y.~Saisho and H.~Tanemura, \emph{Pitman type theorem for one-dimensional
  diffusion processes}, Tokyo J. Math. \textbf{13} (1990), no.~2, 429--440.

\bibitem{Sogo}
K.~Sogo, \emph{Toda molecule equation and quotient-difference method}, J. Phys.
  Soc. Japan \textbf{62} (1993), no.~4, 1081--1084.

\bibitem{SP}
H.~Spohn, \emph{Generalized {G}ibbs ensembles of the classical {T}oda chain},
  J. Stat. Phys. \textbf{180} (2020), no.~1-6, 4--22.

\bibitem{TM1995}
D.~Takahashi and J.~Matsukidaira, \emph{On discrete soliton equations related
  to cellular automata}, Phys. Lett. A \textbf{209} (1995), 184--188.

\bibitem{TMcar}
\bysame, \emph{Box and ball system with a carrier and ultra-discrete modified
  {KdV} equation}, {RIMS} {K}okyuroku ({K}yoto University) \textbf{1020}
  (1997), 1--14.

\bibitem{takahashi1990}
D.~Takahashi and J.~Satsuma, \emph{A soliton cellular automaton}, J. Phys. Soc.
  Japan \textbf{59} (1990), 3514--3519.

\bibitem{Toda}
M.~Toda, \emph{Vibration of a chain with nonlinear interaction}, Journal of the
  Physical Society of Japan \textbf{22} (1967), no.~2, 431--436.

\bibitem{TTeng}
T.~Tokihiro, \emph{Ultradiscrete systems (cellular automata)}, Discrete
  integrable systems, Lecture Notes in Phys., vol. 644, Springer, Berlin, 2004,
  pp.~383--424.

\bibitem{TT}
T.~Tokihiro, \emph{The mathematics of box-ball systems}, Asakura Shoten, 2010.

\bibitem{TTMS}
T.~Tokihiro, D.~Takahashi, J.~Matsukidaira, and J.~Satsuma, \emph{From soliton
  equations to integrable cellular automata through a limiting procedure},
  Phys. Lett. A \textbf{76} (1996), 3247--3250.

\bibitem{TH}
S.~Tsujimoto and R.~Hirota, \emph{Ultradiscrete {KdV} equation}, J. Phys. Soc.
  Japan \textbf{67} (1998), no.~6, 1809--1810.

\bibitem{WNSRG}
R.~Willox, Y.~Nakata, J.~Satsuma, A.~Ramani, and B.~Grammaticos, \emph{Solving
  the ultradiscrete {KdV} equation}, Journal of Physics A: Mathematical and
  Theoretical \textbf{43} (2010), no.~48, 482003.

\end{thebibliography}

\end{document}